%% file: main.tex
\documentclass[a4paper,english]{lipics}

\input{macros}

\title{Efficient Type Checking for Path Polymorphism}

\author[2]{Juan Edi}
\author[1,2]{Andr\'es Viso}
\author[1,3]{Eduardo Bonelli}
\affil[1]{Consejo Nacional de Investigaciones Cient\'{i}ficas y T\'{e}cnicas -- CONICET}
\affil[2]{Departamento de Computaci\'{o}n \\
  Facultad de Ciencias Exactas y Naturales \\
  Universidad de Buenos Aires -- UBA \\
  Buenos Aires, Argentina}
\affil[3]{Departamento de Ciencia y Tecnolog\'{i}a \\
  Universidad Nacional de Quilmes -- UNQ \\
  Bernal, Argentina}
\authorrunning{J. Edi, A. Viso, and E. Bonelli}

\Copyright{Juan Edi, Andr\'{e}s Viso, and Eduardo Bonelli}

\subjclass{
F4.1 Lambda calculus and related systems,
F.3.2 Semantics of Programming Languages,
D.3.3 Language Constructs and Features.
}
\keywords{$\lambda$-Calculus, Pattern Matching, Path Polymorphism, Type Checking}

\serieslogo{}
\volumeinfo
  {Tarmo Uustalu}
  {1}
  {21st International Conference on Types for Proofs and
Programs (TYPES 2015)}
  {69}
  {1}
  {1}
\EventShortName{TYPES 2015}
\DOI{10.4230/LIPIcs.xxx.yyy.p}

\begin{document}

\maketitle

\begin{abstract}
  A type system combining type application, constants as types, union types (associative,
  commutative and idempotent) and recursive types has recently been proposed for statically typing
  \emph{path polymorphism}, the ability to define functions that can operate uniformly over
  recursively specified applicative data structures. A typical pattern such functions resort to is
  $\dataterm{x}{y}$ which decomposes a compound, in other words any applicative tree structure, into
  its parts. We study type-checking for this type system in two stages. First we propose algorithms
  for checking type equivalence and subtyping based on coinductive characterizations of those
  relations. We then formulate a syntax-directed presentation and prove its equivalence with the
  original one. This yields a type-checking algorithm which unfortunately has exponential time
  complexity in the worst case. A second algorithm is then proposed, based on automata techniques,
  which yields a polynomial-time type-checking algorithm.
\end{abstract}

\input{introduction}

\input{reviewOfCAP}

\input{invertibility}

\input{typeChecking}

\input{efficientTypeChecking}

\input{conclusions}



\bibliography{biblio}


\end{document}


%% file: macros.tex
\usepackage[T1]{fontenc}

\usepackage{amssymb}
\usepackage{array}
\usepackage{bbold}              
\usepackage{centernot}
\usepackage{cite}
\usepackage{color}
\usepackage{enumerate}
\usepackage{hyperref}
\let\mathcur\mathscr
\usepackage{mathrsfs}           
\usepackage{multirow}
\usepackage{paralist}
\usepackage{prooftree}
\usepackage{relsize}
\usepackage{stmaryrd}
\SetSymbolFont{stmry}{bold}{U}{stmry}{m}{n}
\usepackage{tikz}
\usetikzlibrary{positioning,automata}
\usepackage{tikz-qtree}
\let\qTree=\Tree
\let\Tree=\undefined
\usepackage{wrapfig}

\theoremstyle{plain}
\newtheorem{proposition}[theorem]{Proposition}

\newcommand{\comment}[1]{}

\newcommand{\capp}{\textsf{CAP}}
\newcommand{\ppc}{\textsf{PPC}}

\newcommand{\ie}{{\em i.e.~}}
\newcommand{\eg}{{\em e.g.~}}
\newcommand{\Eg}{{\em E.g.~}}
\newcommand{\cf}{{\em cf.~}}

\newcommand{\coloneq}{\ensuremath{:=}}
\newcommand{\Coloneq}{\ensuremath{:\coloneq}}
\newcommand{\llbrace}{\{\!\!\{}
\newcommand{\rrbrace}{\}\!\!\}}

\newcommand{\varrow}{\ensuremath{\mathrel{\resizebox{1.6ex}{0.8ex}{\tikz[baseline=-.6ex]{\path[|->](0,0) edge (0.2,0);}}}}}

\newcommand{\pair}[2]{\ensuremath{\langle{#1},{#2}\rangle}}

\newcommand{\qqquad}{\qquad\quad}
\newcommand{\qqqquad}{\qquad\qquad}
\newcommand{\qqqqquad}{\qqquad\qquad}
\newcommand{\qqqqqquad}{\qqquad\qqquad}

\newcommand{\mathLarger}[1]{\ensuremath{\mathlarger{\mathlarger{#1}}}}

\newcommand{\sequP}[2]{\ensuremath{{#1}\vdash_{\mathsf{p}}{#2}}}
\newcommand{\sequT}[2]{\ensuremath{{#1}\vdash{#2}}}
\newcommand{\sequPDeriv}[2]{\ensuremath{\mathrel{\rhd}\sequP{#1}{#2}}}
\newcommand{\sequTDeriv}[2]{\ensuremath{\mathrel{\rhd}\sequP{#1}{#2}}}
\newcommand{\sequC}[2]{\ensuremath{#2}}
\newcommand{\sequTE}[2]{\ensuremath{{#1}\vdash{#2}}}
\newcommand{\sequSD}[2]{\ensuremath{{#1}\varrow{#2}}}

\newcommand{\emphdef}[1]{\textbf{{#1}}}
\newcommand{\eqdef}{\ensuremath{\triangleq}}

\newcommand{\Rule}[3]{
    \prooftree
         {#1}
    \justifies  
         {#2}
    \thickness=0.05em
    \using
         {#3}
    \endprooftree}
\newcommand{\RuleCo}[3]{
    \prooftree
         {#1}
    \Justifies  
         {#2}
    \thickness=0.1em
    \using
         {#3}
    \endprooftree}

\newcommand{\ruleName}[1]{\ensuremath{\textsc{({#1})}}}

\newcommand{\ruleBeta}{\ruleName{$\beta$}}

\newcommand{\rulePComp}{\ruleName{p-comp}}
\newcommand{\rulePConst}{\ruleName{p-const}}
\newcommand{\rulePMatch}{\ruleName{p-match}}

\newcommand{\ruleTAbs}{\ruleName{t-abs}}
\newcommand{\ruleTApp}{\ruleName{t-app}}
\newcommand{\ruleTComp}{\ruleName{t-comp}}
\newcommand{\ruleTConst}{\ruleName{t-const}}
\newcommand{\ruleTSubs}{\ruleName{t-subs}}
\newcommand{\ruleTVar}{\ruleName{t-var}}

\newcommand{\ruleTalAbs}{\ruleName{t-abs-al}}
\newcommand{\ruleTalApp}{\ruleName{t-app-al}}

\newcommand{\ruleTalComp}{\ruleName{t-comp-al}}
\newcommand{\ruleTalConst}{\ruleName{t-const-al}}

\newcommand{\ruleTalVar}{\ruleName{t-var-al}}

\newcommand{\ruleEqmuComp}{\ruleName{e-comp}}
\newcommand{\ruleEqmuContr}{\ruleName{e-contr}}
\newcommand{\ruleEqmuFold}{\ruleName{e-fold}}
\newcommand{\ruleEqmuFunc}{\ruleName{e-func}}
\newcommand{\ruleEqmuRec}{\ruleName{e-rec}}
\newcommand{\ruleEqmuRefl}{\ruleName{e-refl}}
\newcommand{\ruleEqmuSymm}{\ruleName{e-symm}}
\newcommand{\ruleEqmuTrans}{\ruleName{e-trans}}
\newcommand{\ruleEqmuUnion}{\ruleName{e-union}}
\newcommand{\ruleEqmuUnionAssoc}{\ruleName{e-union-assoc}}
\newcommand{\ruleEqmuUnionComm}{\ruleName{e-union-comm}}
\newcommand{\ruleEqmuUnionIdem}{\ruleName{e-union-idem}}

\newcommand{\ruleEqcoComp}{\ruleName{e-comp-t}}
\newcommand{\ruleEqcoFunc}{\ruleName{e-func-t}}
\newcommand{\ruleEqcoRefl}{\ruleName{e-refl-t}}
\newcommand{\ruleEqcoUnion}{\ruleName{e-union-t}}

\newcommand{\ruleEqalComp}{\ruleName{e-comp-al}}
\newcommand{\ruleEqalFunc}{\ruleName{e-func-al}}
\newcommand{\ruleEqalRecL}{\ruleName{e-rec-l-al}}
\newcommand{\ruleEqalRecR}{\ruleName{e-rec-r-al}}
\newcommand{\ruleEqalRefl}{\ruleName{e-refl-al}}
\newcommand{\ruleEqalUnion}{\ruleName{e-union-al}}

\newcommand{\ruleSubmuComp}{\ruleName{s-comp}}
\newcommand{\ruleSubmuEq}{\ruleName{s-eq}}
\newcommand{\ruleSubmuFunc}{\ruleName{s-func}}
\newcommand{\ruleSubmuHyp}{\ruleName{s-hyp}}
\newcommand{\ruleSubmuRec}{\ruleName{s-rec}}
\newcommand{\ruleSubmuRefl}{\ruleName{s-refl}}
\newcommand{\ruleSubmuTrans}{\ruleName{s-trans}}
\newcommand{\ruleSubmuUnionL}{\ruleName{s-union-l}}
\newcommand{\ruleSubmuUnionRL}{\ruleName{s-union-r1}}
\newcommand{\ruleSubmuUnionRR}{\ruleName{s-union-r2}}

\newcommand{\ruleSubcoComp}{\ruleName{s-comp-t}}
\newcommand{\ruleSubcoFunc}{\ruleName{s-func-t}}
\newcommand{\ruleSubcoRefl}{\ruleName{s-refl-t}}
\newcommand{\ruleSubcoUnion}{\ruleName{s-union-t}}

\newcommand{\ruleSubalComp}{\ruleName{s-comp-al}}
\newcommand{\ruleSubalFunc}{\ruleName{s-func-al}}
\newcommand{\ruleSubalRecL}{\ruleName{s-rec-l-al}}
\newcommand{\ruleSubalRecR}{\ruleName{s-rec-r-al}}
\newcommand{\ruleSubalRefl}{\ruleName{s-refl-al}}
\newcommand{\ruleSubalUnionL}{\ruleName{s-union-l-al}}
\newcommand{\ruleSubalUnionR}{\ruleName{s-union-r-al}}

\newcommand{\set}[1]{\ensuremath{\left\{{#1}\right\}}}
\newcommand{\lista}[1]{\ensuremath{\left[{#1}\right]}}

\newcommand{\Variable}{\ensuremath{{\mathbb{V}}}}
\newcommand{\Constant}{\ensuremath{{\mathbb{C}}}}
\newcommand{\Pattern}{\ensuremath{{\mathbb{P}}}}
\newcommand{\DataStructure}{\ensuremath{{\mathbb{D}}}}
\newcommand{\Term}{\ensuremath{{\mathbb{T}}}}
\newcommand{\MatchableForms}{\ensuremath{{\mathbb{M}}}}

\newcommand{\TypeVariable}{\ensuremath{\mathcal{V}}}
\newcommand{\TypeConstant}{\ensuremath{\mathcal{C}}}
\newcommand{\Type}{\ensuremath{\mathcal{T}}}

\newcommand{\DataTypeVariable}{\ensuremath{\TypeVariable_{D}}}
\newcommand{\DataType}{\ensuremath{\Type_{D}}}

\newcommand{\Tree}{\ensuremath{\mathfrak{T}}}

\newcommand{\powerset}[1]{\ensuremath{\wp\left({#1}\right)}}

\newcommand{\Natural}{\ensuremath{\mathbb{N}}}

\newcommand{\Phieqtypeco}{\ensuremath{\Phi_{\eqtypeco}}}
\newcommand{\Phieqtypeal}{\ensuremath{\Phi_{\eqtypeal}}}
\newcommand{\Phisubtypeco}{\ensuremath{\Phi_{\subtypeco}}}
\newcommand{\Phisubtypeal}{\ensuremath{\Phi_{\subtypeal}}}

\newcommand{\M}{\ensuremath{{\cal M}}}
\renewcommand{\S}{\ensuremath{{\cal S}}}
\newcommand{\R}{\ensuremath{{\cal R}}}

\newcommand{\X}{\ensuremath{{\cal X}}}

\newcommand{\treeFont}[1]{\mathcur{#1}}
\newcommand{\tA}{\ensuremath{\treeFont{A}}}
\newcommand{\tB}{\ensuremath{\treeFont{B}}}

\newcommand{\tD}{\ensuremath{\treeFont{D}}}

\newcommand{\contextFont}[1]{\mathcal{#1}}
\newcommand{\cA}{\ensuremath{\contextFont{A}}}
\newcommand{\cB}{\ensuremath{\contextFont{B}}}

\newcommand{\ctxtlra}{\ensuremath{\mathscr{C}}}
\newcommand{\ctxtlrb}{\ensuremath{\mathscr{D}}}

\newcommand{\project}[2]{\ensuremath{{#2}^{#1}}}
\newcommand{\eraserec}[1]{\ensuremath{{#1}^{\setminus\irectype}}}

\newcommand{\idatatype}{\ensuremath{\mathrel@}}
\newcommand{\ifunctype}{\ensuremath{\supset}}
\newcommand{\iuniontype}{\ensuremath{\oplus}}
\newcommand{\irectype}{\ensuremath{\mu}}

\newcommand{\consttype}[1]{\ensuremath{\mathbb{#1}}}
\newcommand{\datatype}[2]{\ensuremath{{#1}\idatatype{#2}}}
\newcommand{\functype}[2]{\ensuremath{{#1}\ifunctype{#2}}}
\newcommand{\uniontype}[2]{\ensuremath{{#1}\iuniontype{#2}}}
\newcommand{\maxuniontype}[2]{\ensuremath{\mathLarger{\iuniontype}_{\substack{#1}}{#2}}}
\newcommand{\rectype}[2]{\ensuremath{\irectype{#1}.{#2}}}

\newcommand{\subtype}{\ensuremath{\preceq}}

\newcommand{\eqtype}{\ensuremath{\simeq}}

\newcommand{\subtypemu}{\ensuremath{\subtype_{\mu}}}

\newcommand{\eqtypemu}{\ensuremath{\eqtype_{\mu}}}

\newcommand{\subtypeal}{\ensuremath{\subtype_{\vec{\mu}}}}

\newcommand{\eqtypeal}{\ensuremath{\eqtype_{\vec{\mu}}}}

\newcommand{\subtypeco}{\ensuremath{\subtype_{\Tree}}}

\newcommand{\eqtypeco}{\ensuremath{\eqtype_{\Tree}}}

\newcommand{\nil}{\consttype{nil}}
\newcommand{\cons}{\consttype{cons}}
\newcommand{\listType}{\ensuremath{\mathsf{List}_A}}
\newcommand{\bool}{\ensuremath{\mathsf{Bool}}}
\newcommand{\nat}{\ensuremath{\mathsf{Nat}}}

\newcommand{\iappterm}{\ensuremath{\,}}
\newcommand{\idataterm}{\ensuremath{\,}}
\newcommand{\icaseterm}{\ensuremath{\mathrel|}}
\newcommand{\ifuncterm}{\ensuremath{\shortrightarrow}}

\newcommand{\matchable}[1]{\ensuremath{{#1}}}
\newcommand{\constterm}[1]{\ensuremath{\mathtt{#1}}}
\newcommand{\appterm}[2]{\ensuremath{{#1}\iappterm{#2}}}
\newcommand{\dataterm}[2]{\ensuremath{{#1}\idataterm{#2}}}
\newcommand{\absterm}[3]{\ensuremath{{#1}\ifuncterm_{#3}{#2}}}

\newcommand{\ifThenElse}[3]{\texttt{if }{#1}\texttt{ then }{#2}\texttt{ else }{#3}}

\newcommand{\fail}{\ensuremath{\mathtt{fail}}}
\newcommand{\wait}{\ensuremath{\mathtt{wait}}}

\newcommand{\reduce}{\ensuremath{\rightarrow}}

\newcommand{\fv}[1]{\ensuremath{\mathsf{fv}\!\left({#1}\right)}}
\newcommand{\fm}[1]{\ensuremath{\mathsf{fm}\!\left({#1}\right)}}

\newcommand{\dom}[1]{\ensuremath{\mathsf{dom}\left({#1}\right)}}

\newcommand{\rename}[2]{\ensuremath{\left\{{#2}/{#1}\right\}}}
\newcommand{\substitute}[3]{\rename{#1}{#2}{#3}}
\newcommand{\basicmatch}[2]{\ensuremath{\llbrace{#2}/{#1}\rrbrace}}
\newcommand{\match}[2]{\basicmatch{#1}{#2}}

\newcommand{\card}[2]{\ensuremath{\#_{#1}\!\left({#2}\right)}}
\newcommand{\cut}[2]{\ensuremath{{#1}|_{{#2}}}}

\newcommand{\toBTree}[1]{\ensuremath{\llbracket{#1}\rrbracket^{\Tree}}}

\newcommand{\iPsicomp}{\ensuremath{{\cal P}_{\mathsf{comp}}}}
\newcommand{\Psicomp}[2]{\ensuremath{\iPsicomp\!\left({#1},{#2}\right)}}

\newcommand{\compatible}[2]{\ensuremath{{#1}\lll{#2}}}
\newcommand{\matches}[2]{\ensuremath{{#1}\mathbin{\vartriangleleft}{#2}}}
\newcommand{\nmatches}[2]{\ensuremath{{#1}\mathbin{\centernot\vartriangleleft}{#2}}}

\newcommand{\at}[2]{{#1}\ensuremath{|_{#2}}}
\newcommand{\lookup}[2]{{#1}\ensuremath{\|_{#2}}}
\newcommand{\pos}[1]{\ensuremath{\mathsf{pos}\!\left({#1}\right)}}
\newcommand{\cpos}[2]{\ensuremath{\mathsf{mmpos}\!\left({#1},{#2}\right)}}
\newcommand{\mpos}[1]{\ensuremath{\mathsf{maxpos}\!\left({#1}\right)}}

\newcommand{\fTypeCheckName}{\ensuremath{\mathsf{tc}}}
\newcommand{\fTypeCheck}[2]{\ensuremath{\fTypeCheckName({#1},{#2})}}
\newcommand{\fTypeCheckPName}{\ensuremath{\mathsf{tcp}}}
\newcommand{\fTypeCheckP}[2]{\ensuremath{\fTypeCheckPName({#1},{#2})}}
\newcommand{\fEqtypeName}{\ensuremath{\mathsf{eqtype}}}
\newcommand{\fEqtype}[3]{\ensuremath{\fEqtypeName({#1},{#2},{#3})}}
\newcommand{\fSubtypeName}{\ensuremath{\mathsf{subtype}}}
\newcommand{\fSubtype}[3]{\ensuremath{\fSubtypeName({#1},{#2},{#3})}}
\newcommand{\fCompName}{\ensuremath{\mathsf{compatible}}}
\newcommand{\fComp}[2]{\ensuremath{\fCompName({#1},{#2})}}
\newcommand{\fPCompName}{\ensuremath{\mathsf{pcomp}}}
\newcommand{\fPComp}[2]{\ensuremath{\fPCompName({#1},{#2})}}
\newcommand{\fUnfoldName}{\ensuremath{\mathsf{unfold}}}
\newcommand{\fUnfold}[1]{\ensuremath{\fUnfoldName({#1})}}

\newcommand{\fBuildUniverseName}{\ensuremath{\mathsf{buildUniverse}}}
\newcommand{\fBuildUniverse}[1]{\ensuremath{\fBuildUniverseName({#1})}}
\newcommand{\fChildrenName}{\ensuremath{\mathsf{children}}}
\newcommand{\fChildren}[1]{\ensuremath{\fChildrenName({#1})}}

\newcommand{\fGfpName}{\ensuremath{\mathsf{gfp}}}
\newcommand{\fGfp}[2]{\ensuremath{\fGfpName({#1},{#2})}}
\newcommand{\fTakeOneName}{\ensuremath{\mathsf{takeOne}}}
\newcommand{\fTakeOne}[1]{\ensuremath{\fTakeOneName({#1})}}
\newcommand{\fCheckName}{\ensuremath{\mathsf{check}}}
\newcommand{\fCheck}[2]{\ensuremath{\fCheckName({#1},{#2})}}
\newcommand{\fInvalidateName}{\ensuremath{\mathsf{invalidate}}}
\newcommand{\fInvalidate}[4]{\ensuremath{\fInvalidateName({#1},{#2},{#3},{#4})}}

\newcommand{\ttIf}{\ensuremath{\mathrel{\mathtt{if}}}}
\newcommand{\ttThen}{\ensuremath{\mathrel{\mathtt{then}}}}
\newcommand{\ttElse}{\ensuremath{\mathrel{\mathtt{else}}}}
\newcommand{\ttForeach}{\ensuremath{\mathrel{\mathtt{foreach}}}}
\newcommand{\ttLet}{\ensuremath{\mathrel{\mathtt{let}}}}
\newcommand{\ttIn}{\ensuremath{\mathrel{\mathtt{in}}}}
\newcommand{\ttAnd}{\ensuremath{\mathbin{\mathtt{and}}}}
\newcommand{\ttOr}{\ensuremath{\mathbin{\mathtt{or}}}}
\newcommand{\ttNot}{\ensuremath{\mathtt{not}\ }}
\newcommand{\ttFail}{\ensuremath{\mathtt{fail}}}
\newcommand{\ttSeq}{\ensuremath{\mathrel{\mathtt{seq}}}}

\renewcommand{\l}{\lambda}

\newcommand{\TreeN}{\ensuremath{\Tree^\mathfrak{n}}}
\newcommand{\nuniontype}[3]{\ensuremath{\mathLarger{\iuniontype}^{#2}_{#1}{#3}}}

\newcommand{\toNTree}[1]{\ensuremath{\llbracket{#1}\rrbracket^\mathfrak{n}}}
\newcommand{\subtypeup}[1]{\ensuremath{\subtype_{\TreeN}^{#1}}}
\newcommand{\subtypeaci}{\ensuremath{\subtypeup{\varnothing}}}
\newcommand{\Phisubtypeup}[1]{\ensuremath{\Phi_{\subtypeup{#1}}}}
\newcommand{\Phisubtypeaci}{\ensuremath{\Phisubtypeup{\varnothing}}}
\newcommand{\ruleSubupRefl}{\ruleName{s-refl-up}}
\newcommand{\ruleSubupFunc}{\ruleName{s-func-up}}
\newcommand{\ruleSubupComp}{\ruleName{s-comp-up}}
\newcommand{\ruleSubupUnion}{\ruleName{s-union-up}}
\newcommand{\ruleSubupUnionL}{\ruleName{s-union-l-up}}
\newcommand{\ruleSubupUnionR}{\ruleName{s-union-r-up}}

\newcommand{\tT}{\ensuremath{\treeFont{T}}}
\newcommand{\tS}{\ensuremath{\treeFont{S}}}
\newcommand{\ttCase}{\ensuremath{\mathrel{\mathtt{case}}}}
\newcommand{\ttOf}{\ensuremath{\mathrel{\mathtt{of}}}}
\newcommand{\ttReturn}{\ensuremath{\mathrel{\mathtt{return}}}}
\newcommand{\ttWhile}{\ensuremath{\mathrel{\mathtt{while}}}}

\newcommand{\setinsert}[2]{\mathsf{insert}({#1},\ {#2})}
\newcommand{\parents}[1]{\mathsf{parents}({#1})}
\newcommand{\setmove}[3]{\mathsf{move}({#1},\ {#2},\ {#3})}

\newcommand{\ttTrue}{\ensuremath{\mathbin{\mathtt{true}}}}

\newcommand{\triple}[3]{\ensuremath{\langle{#1},{#2},{#3}\rangle}}
\newcommand{\eqtypeup}[1]{\ensuremath{\eqtype_{\TreeN}^{#1}}}
\newcommand{\eqtypeaci}{\ensuremath{\eqtypeup{\varnothing}}}

\newcommand{\Phieqtypeup}[1]{\ensuremath{\Phi_{\eqtypeup{#1}}}}
\newcommand{\Phieqtypeaci}{\ensuremath{\Phieqtypeup{\varnothing}}}
\newcommand{\ruleEqupRefl}{\ruleName{e-refl-up}}
\newcommand{\ruleEqupFunc}{\ruleName{e-func-up}}
\newcommand{\ruleEqupComp}{\ruleName{e-comp-up}}
\newcommand{\ruleEqupUnion}{\ruleName{e-union-up}}
\newcommand{\ruleEqupUnionL}{\ruleName{e-union-l-up}}
\newcommand{\ruleEqupUnionR}{\ruleName{e-union-r-up}}

\renewcommand{\O}{\ensuremath{{\mathcal{O}}}}


%% file: introduction.tex
\section{Introduction}
\label{sec:intro}

The \emph{lambda-calculus} plays an important role in the study of
programming languages (PLs). Programs are represented as syntactic terms and
execution by repeated simplification of these terms using a reduction
rule called \emph{$\beta$-reduction}. The study of the lambda-calculus
has produced deep results in both the theory and the implementation of
PLs. Many variants of the lambda-calculus have been introduced with
the purpose of studying specific PL features. One such feature of
interest is \emph{pattern-matching}. Pattern-matching is used extensively in
PLs as a means for writing more succinct and, at the same time,
elegant programs. This is 
particularly so in the functional programming community, but by no means
restricted to that community. 

In the standard lambda-calculus,
functions are represented as expressions of the form $\l x.t$, $x$
being the formal parameter and $t$ the body. Such a function may be
applied to any term, regardless of its form. This is expressed
by the above mentioned $\beta$-reduction rule: $(\l x.t)\,s
\reduce_{\beta} \substitute{x}{s}{t}$, where $\substitute{x}{s}{t}$ stands for the result of
replacing all free occurrences of $x$ in $t$ with $s$. Note that, in
this rule, no requirement on the form of $s$ is placed.
\emph{Pattern calculi} are generalizations of the $\beta$-reduction
rule in which
abstractions $\l x.t$ are replaced by $\l p.t$ where $p$ is called a
\emph{pattern}. An example is $\l\pair{x}{y}.x$ for projecting the first
component of a pair, the pattern $p$ being
  $\pair{x}{y}$. An expression such as $\appterm{(\l\pair{x}{y}.x)}{s}$ will only
  be able to reduce if $s$ indeed is of the form $\pair{s_1}{s_2}$; it
  will otherwise be blocked.

Patterns may be catalogued in at least two
  dimensions. One is their \emph{structure} and another their \emph{time of
  creation}. The structure of patterns may be very general.  
Such is the case of variables: any term can match a variable, as
in the standard lambda-calculus. The structure of a pattern
may also be very specific.  Such is the case when 
arbitrary
terms are allowed to be
patterns~\cite{vanOostrom90,DBLP:journals/tcs/KlopOV08}. Regarding the
time of creation, patterns may either be \emph{static} or
\emph{dynamic}. Static patterns are those that are created at compile time,
such as the pattern $\pair{x}{y}$ mentioned above. Dynamic patterns are those that may be generated at
run-time~\cite{DBLP:journals/jfp/JayK09,DBLP:books/daglib/0023687}. For example,
consider the term $\l x.(\l(\appterm{x}{y}).y)$; note that it has an
occurrence of a pattern $x\, y$ with a free variable, namely the $x$
in $\dataterm{x}{y}$, that is bound to the outermost lambda. If this
term is applied to a constant $\constterm{c}$, then one obtains
$\l\appterm{\constterm{c}}{y}.y$. However, if we apply it to the constant
$\constterm{d}$, then we obtain $\l\appterm{\constterm{d}}{y}.y$. Both
patterns $\appterm{\constterm{c}}{y}$ and $\appterm{\constterm{d}}{y}$
are created during execution. Note that one could also replace the $x$
in the pattern $\dataterm{x}{y}$ with an abstraction. This leads to computations
that evaluate to patterns. 

Expressive pattern features may
easily break desired properties, such as confluence, and are not easy
to endow with type systems.
This work is an attempt at devising type systems for such expressive pattern
calculi. We originally set out to type-check the \emph{Pure Pattern Calculus}
($\ppc$)~\cite{DBLP:journals/jfp/JayK09,DBLP:books/daglib/0023687}. \ppc\ is a
lambda-calculus that embodies the essence of dynamic patterns by
stripping away everything inessential to the reduction and matching
process of dynamic patterns. It admits terms such as $\l
x.(\l(\appterm{x}{y}).y)$. We soon realized that
typing \ppc\ was too challenging and noticed that the static fragment of
$\ppc$, which we dub \emph{Calculus of Applicative Patterns} ($\capp$), was
already challenging in itself.  $\capp$ also admits patterns such as
$\dataterm{\matchable{x}}{\matchable{y}}$ however all variables in this
pattern are considered \emph{bound}. Thus, in a term such as
$\l(\appterm{x}{y}).s$ both occurrences of $x$ and $y$
are bound in  $s$, disallowing reduction inside patterns. Such
patterns, however, allow arguments that are applications to be
decomposed, as long as these applications encode \emph{data
  structures}. They are therefore useful
for writing functions that operate on
\emph{semi-structured data}. 

The main obstacle for typing \capp\ is dealing
in the type system with a form of polymorphism called
\emph{path polymorphism}~\cite{DBLP:journals/jfp/JayK09,DBLP:books/daglib/0023687},
that arises from these kinds of patterns. We next briefly describe
path polymorphism and the requirements it
places on typing considerations.

\textbf{Path Polymorphism.} In $\capp$ data structures are trees.
These trees are built using application
and variable arity constants or constructors. Examples of two such
trees follow, where the first one represents a list and the second
a binary tree: $$
\arraycolsep=1.5pt\def\arraystretch{1}
\begin{array}{c}
\dataterm{\dataterm{\constterm{cons}}{(\dataterm{\constterm{vl}}{1})}}{(\dataterm{\dataterm{\constterm{cons}}{(\dataterm{\constterm{vl}}{2})}}{\constterm{nil}})} \\
\dataterm{\dataterm{\dataterm{\constterm{node}}{(\dataterm{\constterm{vl}}{3})}}{(\dataterm{\dataterm{\dataterm{\constterm{node}}{(\dataterm{\constterm{vl}}{4})}}{\constterm{nil}}}{\constterm{nil}})}}{(\dataterm{\dataterm{\dataterm{\constterm{node}}{(\dataterm{\constterm{vl}}{5})}}{\constterm{nil}}}{\constterm{nil}})}
\end{array} $$
  The constructor \constterm{vl} is used to tag values ($1$ and $2$ in the first
  case, and $3$, $4$ and $5$ in the second). A ``map'' function for
  updating the values of any of these two structures by applying some
  user-supplied function $f$ follows, where type annotations are omitted for
  clarity: 
\begin{equation}
\mathsf{upd} = \absterm{\matchable{f}}{
\arraycolsep=1.5pt\def\arraystretch{1}
\begin{array}[t]{clll}
(          & \dataterm{\constterm{vl}}{\matchable{z}} & \ifuncterm_{} & \dataterm{\constterm{vl}}{(\appterm{f}{z})} \\
\icaseterm & \dataterm{\matchable{x}}{\matchable{y}}  & \ifuncterm_{} & \dataterm{(\appterm{\appterm{\mathsf{upd}}{f}}{x})}{(\appterm{\appterm{\mathsf{upd}}{f}}{y})} \\
\icaseterm & \matchable{w}                            & \ifuncterm_{} & w)
\end{array}
}{}
\label{eq:intro:upd}
\end{equation}
The expression $\appterm{\mathsf{upd}}{(+1)}$ may thus be applied to
any of the two data structures illustrated above. For example, we can evaluate
$\appterm{\appterm{\mathsf{upd}}{(+1)}}{\dataterm{\dataterm{\constterm{cons}}{(\dataterm{\constterm{vl}}{1})}}{(\dataterm{\dataterm{\constterm{cons}}{(\dataterm{\constterm{vl}}{2})}}{\constterm{nil}})}}$
and also 
$\appterm{\appterm{\mathsf{upd}}{(+1)}}{\dataterm{\dataterm{\dataterm{\constterm{node}}{(\dataterm{\constterm{vl}}{3})}}{(\dataterm{\dataterm{\dataterm{\constterm{node}}{(\dataterm{\constterm{vl}}{4})}}{\constterm{nil}}}{\constterm{nil}})}}{(\dataterm{\dataterm{\dataterm{\constterm{node}}{(\dataterm{\constterm{vl}}{5})}}{\constterm{nil}}}{\constterm{nil}})}}$. The expression to
the right of ``='' is called an \emph{abstraction} (or \emph{case}) and
consists of a unique \emph{branch}; this branch in turn is formed from a
pattern (\matchable{f}), and a body (in this case the body is itself another
abstraction that consists of three branches). An  argument to an abstraction is
matched against the patterns, in the order in which they are written, and the
appropriate body is selected. 

Notice the pattern $\dataterm{\matchable{x}}{\matchable{y}}$. During evaluation
of $\appterm{\appterm{\mathsf{upd}}{(+1)}}{\dataterm{\dataterm{\constterm{cons}}{(\dataterm{\constterm{vl}}{1})}}{(\dataterm{\dataterm{\constterm{cons}}{(\dataterm{\constterm{vl}}{2})}}{\constterm{nil}})}}$ the variables $\matchable{x}$ and
$\matchable{y}$ may be instantiated with different applicative terms in each
recursive call to $\mathsf{upd}$. For example:
\begin{center}
\begin{tabular}{l|c|c}
 & $\matchable{x}$ & $\matchable{y}$ \\ 
\hline
$\appterm{\appterm{\mathsf{upd}}{(+1)}}{s}$
  & $\dataterm{\constterm{cons}}{(\dataterm{\constterm{vl}}{1})}$
  & $\dataterm{\dataterm{\constterm{cons}}{(\dataterm{\constterm{vl}}{2})}}{\constterm{nil}}$ \\
$\appterm{\appterm{\mathsf{upd}}{(+1)}}{(\dataterm{\constterm{cons}}{(\dataterm{\constterm{vl}}{1})})}$
  & $\constterm{cons}$
  & $\dataterm{\constterm{vl}}{1}$ \\
$\appterm{\appterm{\mathsf{upd}}{(+1)}}{(\dataterm{\dataterm{\constterm{cons}}{(\dataterm{\constterm{vl}}{2})}}{\constterm{nil}})}$
  & $\dataterm{\constterm{cons}}{(\dataterm{\constterm{vl}}{2})}$
  & $\constterm{nil}$
\end{tabular}
\end{center}
The type assigned to $\matchable{x}$ and $\matchable{y}$ should encompass all
terms in its respective column.

\textbf{Singleton Types and Type Application.} A further
consideration in typing $\capp$ is that terms such as the ones
depicted below should clearly not be typable.
\begin{equation}
\appterm{(\absterm{\constterm{nil}}{0}{})}{\constterm{cons}}
\qqqquad
\appterm{(\absterm{\dataterm{\constterm{vl}}{x}}{x+1}{\set{x:\nat}})}
        {(\dataterm{\constterm{vl}}{\constterm{true}})}
\label{eq:intro:cx}
\end{equation}
In the first case, \constterm{cons}
will never match \constterm{nil}. The type system will resort to
singleton types in order to determine this: \constterm{cons} will be
assigned a type of the form \consttype{cons} which will fail to match
\consttype{nil}. The second expression in (\ref{eq:intro:cx})
breaks \emph{Subject Reduction} (SR): reduction will produce
$\constterm{true}+1$. Applicative types of the form
$\datatype{\consttype{vl}}{\consttype{true}}$ will allow us to check
for these situations, $\datatype{}{}$ being a new type constructor
that applies datatypes to arbitrary types. Also, note the use of
typing environments (the expression $\set{x:\nat}$) to declare the types of the variables of patterns
in branches. These are supplied by the programmer.

\textbf{Union and Recursive Types.} On the assumption that the programmer has provided an exhaustive
coverage, the type assigned by $\capp$ to the variable $\matchable{x}$
in the pattern $x\,y$ in  $\mathsf{upd}$ is:
$$ \rectype{\alpha}{
\uniontype{
 \uniontype{(\datatype{\consttype{vl}}{A})}{(\datatype{\alpha}{\alpha})}
}{(
 \uniontype{\consttype{cons}}{\uniontype{\consttype{node}}{\consttype{nil}}}
)}} $$ Here $\irectype$ is the recursive type constructor and $\iuniontype$ the
union type constructor. $\consttype{vl}$ is the singleton type used for typing the constant
$\constterm{vl}$ and $\idatatype$ denotes type application, as
mentioned above. The union type constructor is used to collect the
types of all the branches. The variable $y$ in the pattern $\dataterm{x}{y}$ will
also be assigned the same type as $x$. Thus variables in applicative patterns are
assigned union types. $\mathsf{upd}$ itself is assigned type
$\functype{(\functype{A}{B})}{(\functype{F_A}{F_B})}$, where $F_X$ is
$\rectype{\alpha}{
  \uniontype{
    \uniontype{(\datatype{\consttype{vl}}{X})}{(\datatype{\alpha}{\alpha})}
  }{(
    \uniontype{\consttype{cons}}{\uniontype{\consttype{node}}{\consttype{nil}}}
  )}
}$.

\textbf{Type-Checking for $\capp$.} Based on these, and other
similar considerations, we proposed \emph{typed $\capp$}~\cite{DBLP:journals/entcs/VisoBA16}, referred
to simply as $\capp$ in the sequel. The system consists of typing rules that
combine singleton types, type application, union types, recursive types
and subtyping. Also it enjoys several properties, the salient one being
safety (subject reduction and progress). Safety relies on a notion of \emph{typed
pattern compatibility} based on subtyping that guarantees that examples such as
(\ref{eq:intro:cx}--right) and the following one do not break safety: 
\begin{equation}
\appterm{
  ((\absterm{\dataterm{\constterm{vl}}{\matchable{x}}}
            {\ifThenElse{x}{1}{0}}
            {\set{x:\bool}})
  \icaseterm
  (\absterm{\dataterm{\constterm{vl}}{\matchable{y}}}
           {y+1}
           {\set{y:\nat}}))
}{(\appterm{\constterm{vl}}{4})}
\label{eq:example:compatib:i}
\end{equation}
Assumptions on associativity and commutativity of typing
operators in~\cite{DBLP:journals/entcs/VisoBA16}, make it non-trivial
to deduce a type-checking algorithm from the typing rules. 
The proposed type system is, moreover, not syntax-directed. Also, it relies
on coinductive notions of type equivalence and subtyping which in the presence
of recursive subtypes are not obviously decidable. A practical implementation
of $\capp$ is instrumental since a robust theoretical analysis without such an
implementation is of little use.

\textbf{Goal and Summary of Contributions.}
  This paper addresses this implementation. It
does so in two stages:
\begin{itemize}
  \item The first stage presents a na\"ive but correct, high-level description
  of a type-checking algorithm, the principal aim being clarity. We propose an
  invertible presentation of the coinductive notions of type-equivalence and
  subtyping of~\cite{DBLP:journals/entcs/VisoBA16} and also a syntax-directed variant of the
  presentation in~\cite{DBLP:journals/entcs/VisoBA16}. This leads to algorithms for checking
  subtyping membership and equivalence modulo associative, commutative and
  idempotent (ACI) unions, both based on an invertible presentation of the
  functional generating the associated coinductive notions.
  
  \item The second stage builds on ideas from the first algorithm with the aim
  of improving efficiency. $\mu$-types are interpreted as infinite $n$-ary
  trees and represented using automata, avoiding having to explicitly handle
  unfoldings of recursive types, and leading to a significant improvement in
  the complexity of the key steps of the type-checking process, namely equality
  and subtype checking.  
\end{itemize}

  
  

\textbf{Related work. }
For literature on (typed) pattern calculi the reader is referred
to~\cite{DBLP:journals/entcs/VisoBA16}. The algorithms for checking equality of recursive
types or subtyping of recursive types have been studied in the 1990s by
Amadio and Cardelli~\cite{DBLP:journals/toplas/AmadioC93}; Kozen, Palsberg, and
Schwartzbach~\cite{DBLP:journals/mscs/KozenPS95}; Brandt and Henglein~\cite{DBLP:conf/tlca/BrandtH97}; Jim and
Palsberg~\cite{Jim97typeinference} among others.
Additionally, Zhao and Palsberg~\cite{DBLP:journals/iandc/PalsbergZ01} studied
the possibilities of incorporating associative and commutative (AC) products to
the equality check, on an automata-based approach that the authors themselves
claimed was not extensible to subtyping~\cite{Zhao:thesis}. Later on Di Cosmo,
Pottier, and R{\'e}my~\cite{DBLP:conf/tlca/CosmoPR05} presented another
automata-based algorithm for subtyping that properly handles AC products with a
complexity cost of $\O(n^2n'^2d^{5/2})$, where $n$ and $n'$ are the sizes of
the analyzed types, and $d$ is a bound on the arity of the involved products.

\textbf{Structure of the paper.} Sec.~\ref{sec:reviewOfCAP} reviews the syntax
and operational semantics of $\capp$, its type system and the main properties.
Further details may be consulted in~\cite{DBLP:journals/entcs/VisoBA16}.
Sec.~\ref{sec:invertibility} proposes invertible generating functions for
coinductive notions of type-equivalence and subtyping that lead to inefficient
but elegant algorithms for checking these relations.
Sec.~\ref{sec:typeChecking} proposes a syntax-directed type system for $\capp$.
Sec.~\ref{sec:efficientTypeChecking} studies a more efficient type-checking
algorithm based on automaton. Finally, we conclude in
Sec.~\ref{sec:conclusions}. An implementation of the algorithms
described here is available online~\cite{EV:2015:Prototipo}.


%% file: reviewOfCAP.tex
\section{Review of \capp}
\label{sec:reviewOfCAP}

\subsection{Syntax and Operational Semantics}

We assume given an infinite set of term variables $\Variable$ and constants
$\Constant$. \capp\ has four syntactic categories, namely \emphdef{patterns}
($p, q, \ldots$), \emphdef{terms} ($s, t, \ldots$), \emphdef{data structures}
($d, e, \ldots$) and \emphdef{matchable forms} ($m, n, \ldots$): $$
\begin{array}{l@{\qquad}l}
\begin{array}[t]{rlll}
p & \Coloneq & \matchable{x}   & \text{(matchable)} \\
  & |        & \constterm{c}   & \text{(constant)} \\
  & |        & \dataterm{p}{p} & \text{(compound)} \\
\end{array}
&
\begin{array}[t]{rlll}
t & \Coloneq  & x                                                                         & \text{(variable)} \\
\hphantom{m}
  & |         & \constterm{c}                                                             & \text{(constant)} \\
  & |         & \appterm{t}{t}                                                            & \text{(application)} \\
  & |         & \absterm{p}{t}{\theta} \icaseterm \dots \icaseterm \absterm{p}{t}{\theta} & \text{(abstraction)}
\end{array}
\\
\\
\begin{array}[t]{rlll}
d & \Coloneq & \constterm{c}   & \text{(constant)} \\
  & |        & \dataterm{d}{t} & \text{(compound)} 
\end{array}
&
\begin{array}[t]{rlll}
m & \Coloneq & d                                                                         & \text{(data structure)} \\
  & |        & \absterm{p}{t}{\theta} \icaseterm \dots \icaseterm \absterm{p}{t}{\theta} & \text{(abstraction)}
\end{array}
\end{array} $$

The set of patterns, terms, data structures and matchable forms are denoted
$\Pattern$, $\Term$, $\DataStructure$ and $\MatchableForms$, resp. Variables
occurring in patterns are called \emphdef{matchables}. We often abbreviate
$\absterm{p_1}{s_1}{\theta_1}\icaseterm\ldots\icaseterm\absterm{p_n}{s_n}{\theta_n}$
with $(\absterm{p_i}{s_i}{\theta_i})_{i \in 1..n}$. The $\theta_i$ are typing
contexts annotating the type assignments for the variables in $p_i$ (\cf
Sec.~\ref{sec:typinsSchemes}). The \emphdef{free variables} of a term $t$
(notation $\fv{t}$) are defined as expected; in a pattern $p$ we call them
\emphdef{free matchables} ($\fm{p}$). All free matchables in each $p_i$ are
assumed to be bound in their respective bodies $s_i$. Positions in patterns and
terms are defined as expected and denoted $\pi,\pi',\ldots$ ($\epsilon$ denotes
the root position). We write \pos{s} for the set of positions of $s$ and
$\at{s}{\pi}$ for the subterm of $s$ occurring at position $\pi$.

A \emphdef{substitution} ($\sigma, \sigma_i, \ldots$) is a partial function
from term variables to terms. If it assigns $u_i$ to $x_i$, $i\in 1..n$, then
we write $\{u_1/x_1,\ldots,u_n/x_n\}$. Its domain (\dom{\sigma}) is
$\set{x_1,\ldots,x_n}$. Also, $\set{}$ is the identity substitution. We write
$\sigma s$ for the result of applying $\sigma$ to term $s$. We say a pattern
$p$ \emphdef{subsumes} a pattern $q$, written $\matches{p}{q}$ if there exists
$\sigma$ such that $\sigma p = q$.  \emphdef{Matchable forms} are required for
defining the \emphdef{matching operation}, described next.

Given a pattern $p$ and a term $s$, the matching operation $\basicmatch{p}{s}$
determines whether $s$ matches $p$. It may have one of three outcomes: success,
fail (in which case it returns the special symbol $\fail$) or undetermined (in
which case it returns the special symbol $\wait$). We say $\basicmatch{p}{s}$
is \emphdef{decided} if it is either successful or it fails. In the former it
yields a substitution $\sigma$; in this case we write $\basicmatch{p}{s} =
\sigma$. The disjoint union of matching outcomes is given as follows
(``\eqdef'' is used for definitional equality): $$
\begin{array}{c}
\begin{array}{rll}
\fail\uplus o            & \eqdef & \fail \\
o \uplus \fail           & \eqdef & \fail \\
\sigma_1 \uplus \sigma_2 & \eqdef & \sigma
\end{array}
\hspace{.5cm}
\begin{array}{rll}
\wait \uplus \sigma & \eqdef & \wait \\
\sigma \uplus \wait & \eqdef & \wait \\
\wait \uplus \wait  & \eqdef & \wait
\end{array}
\end{array} $$
where $o$ denotes any possible output and $\sigma_1\uplus\sigma_2 \eqdef
\sigma$ if the domains of $\sigma_1$ and $\sigma_2$ are disjoint. This always
holds given that patterns are assumed to be linear (at most one occurrence of
any matchable). The matching operation is defined as follows, where the
defining clauses below are evaluated from top to bottom\footnote{This is
simplification to the static patterns case of the matching operation introduced
in~\cite{DBLP:journals/jfp/JayK09}.}: $$
\begin{array}{llll}
\basicmatch{\matchable{x}}{u}                 & \eqdef & \rename{x}{u} \\
\basicmatch{\constterm{c}}{\constterm{c}}     & \eqdef & \set{} \\
\basicmatch{\dataterm{p}{q}}{\dataterm{u}{v}} & \eqdef & \basicmatch{p}{u} \uplus \basicmatch{q}{v} &
  \quad \text{if $\dataterm{u}{v}$ is a \emph{matchable form}} \\
\basicmatch{p}{u}                             & \eqdef & \fail &
  \quad \text{if $u$ is a \emph{matchable form}} \\
\basicmatch{p}{u}                             & \eqdef & \wait
\end{array} $$
For example: $\basicmatch{\constterm{c}}{\absterm{\matchable{x}}{s}{}} =
\fail$; $\basicmatch{\constterm{c}}{\constterm{d}} = \fail$;
$\basicmatch{\constterm{c}}{x} = \wait$ and
$\basicmatch{\dataterm{\constterm{c}}{\constterm{c}}}{\dataterm{x}{\constterm{d}}}
= \fail$. We now turn to the only reduction axiom of $\capp$: $$
\begin{array}{c}
\Rule{\match{p_i}{u} = \fail \text{ for all } i < j
      \quad
      \match{p_j}{u} = \sigma_j
      \quad
      j \in 1..n}
     {\appterm{(\absterm{p_i}{s_i}{\theta_i})_{i \in 1..n}}{u} \reduce \sigma_j s_j}
     {\ruleBeta}
\end{array} $$
It may be applied under any context and states that if the argument $u$ to an
abstraction $(\absterm{p_i}{s_i}{\theta_i})_{i\in 1..n}$ fails to match all
patterns $p_i$ with $i<j$ and successfully matches pattern $p_j$ (producing a
substitution $\sigma_j$), then the term
$\appterm{(\absterm{p_i}{s_i}{\theta_i})_{i \in 1..n}}{u}$ reduces to
$\sigma_j s_j$.

For instance, consider the function $$\mathsf{head} = 
((\absterm{\constterm{nil}}
          {\constterm{nothing}}
          {\set{}})
\icaseterm
(\absterm{\dataterm{\dataterm{\constterm{cons}}{x}}{\mathit{xs}}}
         {\dataterm{\constterm{just}}{x}}
         {\set{x:\nat, \mathit{xs}:\rectype{\alpha}{\uniontype{\consttype{nil}}{\datatype{\datatype{\consttype{cons}}{\nat}}{\alpha}}}}}))
$$ Then, $\appterm{\mathsf{head}}{\constterm{nil}} \reduce \constterm{nothing}$
with $\match{\constterm{nil}}{\constterm{nil}} = \set{}$, while
$\appterm{\mathsf{head}}{(\dataterm{\dataterm{\constterm{cons}}{4}}{\constterm{nil}})} \reduce \dataterm{\constterm{just}}{4}$
since $\match{\constterm{nil}}{\dataterm{\dataterm{\constterm{cons}}{x}}{\constterm{nil}}} = \fail$ and $\match{\dataterm{\dataterm{\constterm{cons}}{x}}{\mathit{xs}}}{\dataterm{\dataterm{\constterm{cons}}{4}}{\constterm{nil}}} = \{4/x,\constterm{nil}/\mathit{xs}\}$.

\begin{proposition}
Reduction in \capp\ is confluent~\cite{DBLP:journals/entcs/VisoBA16}.
\end{proposition}

\input{types}

\subsection{Typing and Safety}
\label{sec:typinsSchemes}

\begin{figure} $$
\begin{array}{c}
\multicolumn{1}{l}{\textbf{Patterns}}
\\
\\
\Rule{\theta(x) = A}
     {\sequP{\theta}{\matchable{x} : A}}
     {\rulePMatch}
\quad
\Rule{}
     {\sequP{\theta}{\constterm{c} : \consttype{c}}}
     {\rulePConst} 
\quad
\Rule{\sequP{\theta}{p : D} \quad \sequP{\theta}{q : A}}
     {\sequP{\theta}{\dataterm{p}{q} : \datatype{D}{A}}}
     {\rulePComp}
\\
\\
\multicolumn{1}{l}{\textbf{Terms}}
\\
\\
\Rule{\Gamma(x) = A}
     {\sequT{\Gamma}{x : A}}
     {\ruleTVar}
\quad
\Rule{}
     {\sequT{\Gamma}{\constterm{c} : \consttype{c}}}
     {\ruleTConst}
\quad
\Rule{\sequT{\Gamma}{r : D} \quad \sequT{\Gamma}{u : A}}
     {\sequT{\Gamma}{\dataterm{r}{u} : \datatype{D}{A}}}
     {\ruleTComp}
\\
\\
\Rule{\begin{array}{c}
        \lista{\sequC{\theta_i}{p_i : A_i}}_{i \in 1..n}
        \text{ compatible}\vspace{.1cm}
        \\
        (\sequP{\theta_i}{p_i : A_i})_{i\in 1..n}
        \quad
        (\dom{\theta_i} = \fm{p_i})_{i\in 1..n}
        \quad
        (\sequT{\Gamma, \theta_i}{s_i : B})_{i\in 1..n}
      \end{array}
     }
     {\sequT{\Gamma}{(\absterm{p_i}{s_i}{\theta_i})_{i \in 1..n} :
      \functype{(\maxuniontype{i \in 1..n}{A_i})}{B}}
     }
     {\ruleTAbs}
\\
\\
\Rule{\sequT{\Gamma}{r : \functype{\maxuniontype{i \in 1..n}{A_i}}{B}}
      \quad
      \sequT{\Gamma}{u : A_k}
      \quad
      k \in 1..n
      }
      {\sequT{\Gamma}{\appterm{r}{u} : B}}
      {\ruleTApp} 
\quad
\Rule{\sequT{\Gamma}{s : A}
      \quad
      \sequTE{}{A \subtypemu A'}
     }
     {\sequT{\Gamma}{s : A'}}
     {\ruleTSubs}
\end{array} $$
\caption{Typing rules for patterns and terms.}
\label{fig:typingSchemesForPatternsAndTerms}
\end{figure}

A \emphdef{typing context} $\Gamma$ (or $\theta$) is a partial function from
term variables to $\mu$-types; $\Gamma(x) = A$ means that $\Gamma$ maps $x$ to
$A$.  We have two typing judgments, one for patterns $\sequP{\theta}{p : A}$
and one for terms $\sequT{\Gamma}{s : A}$. Accordingly, we have two sets of
typing rules: Fig.~\ref{fig:typingSchemesForPatternsAndTerms}, top and bottom.
We write $\sequPDeriv{\theta}{p:A}$ to indicate that the typing judgment
$\sequP{\theta}{p:A}$ is derivable (likewise for $\sequTDeriv{\Gamma}{s:A}$).
The typing schemes speak for themselves except for two of them which we now
comment. The first is $\ruleTApp$. Note that we do not impose any additional
restrictions on $A_i$, in particular it may be a union-type itself. This implies
that the argument $u$ can have a union type too.
Regarding $\ruleTAbs$ it requests a number of conditions. First of all, each of
the patterns $p_i$ must be typable under the typing context $\theta_i$, $i \in
1..n$. Also, the set of free matchables in each $p_i$ must be exactly the
domain of $\theta_i$. Another condition, indicated by $(\sequT{\Gamma,
\theta_i}{s_i : B})_{i\in 1..n}$, is that the bodies of each of the branches
$s_i$, $i\in 1..n$, must be typable under the context extended with the
corresponding $\theta_i$. More noteworthy is the condition that the list
$\lista{\sequC{\theta_i}{p_i : A_i}}_{i \in 1..n}$ be \emph{compatible}. 

Compatibility is a condition that ensures that Subject Reduction is not violated. We briefly
recall it; see~\cite{DBLP:journals/entcs/VisoBA16} for further details and examples. As already
mentioned in example (\ref{eq:example:compatib:i}) of the introduction, if
$p_i$ subsumes $p_j$ (\ie $\matches{p_i}{p_j}$) with $i < j$, then the branch \absterm{p_j}{s_j}{\theta_j}
will never be evaluated since the argument will already match $p_i$. Thus, in
this case, in order to ensure SR we demand that $A_j \subtypemu A_i$. If $p_i$
does not subsume $p_j$ (\ie $\nmatches{p_i}{p_j}$) with $i<j$ we analyze the
cause of failure of subsumption in order to determine whether requirements on
$A_i$ and $A_j$ must be put forward, focusing on those cases where $\pi \in
\pos{p_i} \cap \pos{p_j}$ is an offending position in both patterns. The
following table exhaustively lists them:
\begin{center}
\begin{tabular}{l|c|c|l}
    & $\at{p_i}{\pi}$                       & $\at{p_j}{\pi}$      & \\
\hline
(a) & \multirow{3}{*}{\constterm{c}}        & \matchable{y}        & restriction required \\
(b) &                                       & \constterm{d}        & no overlapping ($\nmatches{p_j}{p_i}$) \\
(c) &                                       & $\appterm{q_1}{q_2}$ & no overlapping \\
\hline
(d) & \multirow{3}{*}{$\appterm{q_1}{q_2}$} & \matchable{y}        & restriction required \\
(e) &                                       & \constterm{d}        & no overlapping
\end{tabular}
\end{center}
\noindent In cases (b), (c) and (e), no extra condition on the types of $p_i$ and
$p_j$ is necessary, since their respective sets of possible arguments are
disjoint. The cases where $A_i$ and $A_j$ must be related are (a) and (d): for
those we require $A_j\subtypemu A_i$. In summary, the cases requiring
conditions on their types are: 1) $\matches{p_i}{p_j}$; and 2)
$\nmatches{p_i}{p_j}$ and $\matches{p_j}{p_i}$.

\begin{definition}
\label{def:constructorAtPosition}
Given a pattern $\sequP{\theta}{p:A}$ and $\pi\in\pos{p}$, we say $A$
\emph{admits a symbol $\odot$ (with $\odot\in \TypeVariable \cup \TypeConstant
\cup \set{\ifunctype,\idatatype}$)} \emph{at position} $\pi$ iff $\odot \in
\lookup{A}{\pi}$, where: $$
\begin{array}{c}
\begin{array}{r@{\quad\eqdef\quad}l@{\quad}l}
\lookup{a}{\epsilon}                 & \set{a} \\
\lookup{(A_1 \star A_2)}{\epsilon}   & \set{\star},       & \star \in \set{\ifunctype, \idatatype} \\
\lookup{(A_1 \star A_2)}{i\pi}       & \lookup{A_i}{\pi}, & \star \in \set{\ifunctype, \idatatype}, i \in \set{1,2} \\
\lookup{(\uniontype{A_1}{A_2})}{\pi} & \lookup{A_1}{\pi} \cup \lookup{A_2}{\pi} \\
\lookup{(\rectype{V}{A'})}{\pi}      & \lookup{(\substitute{V}{\rectype{V}{A'}}{A'})}{\pi}
\end{array}
\end{array} $$
\end{definition}
Note that $\sequPDeriv{\theta}{p:A}$ and contractiveness of $A$ imply
$\lookup{A}{\pi}$ is well-defined for $\pi\in\pos{p}$.

\begin{definition}
The \emph{maximal positions} in a set of positions $P$ are: $$\mpos{P} \eqdef
\set{\pi \in P \mathrel| \nexists \pi'\neq\epsilon. \pi\pi'\in P}$$
The \emph{mismatching positions} between two patterns are defined
below where, recall from the introduction, $\at{p}{\pi}$ stands for the sub-pattern
at position $\pi$ of $p$:
$$\cpos{p}{q} \eqdef
\set{\pi \mathrel| \pi \in \mpos{\pos{p} \cap \pos{q}} \land
\nmatches{\at{p}{\pi}}{\at{q}{\pi}}}$$
\end{definition}

For instance, given patterns $\constterm{nil}$ and
$\dataterm{\dataterm{\constterm{cons}}{x}}{\mathit{xs}}$ with set of positions
$\set{\epsilon}$ and $\set{\epsilon, 1, 2, 11, 12}$ respectively, we have
$\mpos{\constterm{nil}} = \set{\epsilon}$ and
$\mpos{\dataterm{\dataterm{\constterm{cons}}{x}}{\mathit{xs}}} = \set{11, 12}$,
while the only mismatching position among them is the root, \ie
$\cpos{\constterm{nil}}{\dataterm{\dataterm{\constterm{cons}}{x}}{\mathit{xs}}} = \set{\epsilon}$.

\begin{definition}
\label{def:compatibility}
Define the \emph{compatibility predicate} as $$\Psicomp{\sequC{\theta}{p :
A}}{\sequC{\theta'}{q : B}} \eqdef \forall \pi \in \cpos{p}{q}.\lookup{A}{\pi}
\cap \lookup{B}{\pi} \neq \varnothing$$

We say $\sequC{\theta}{p : A}$ is \emph{compatible} with $\sequC{\theta'}{q :
B}$, notation $\compatible{\sequC{\theta}{p : A}}{\sequC{\theta'}{q : B}}$, iff
$$\Psicomp{\sequC{\theta}{p : A}}{\sequC{\theta'}{q : B}} \implies B \subtypemu
A$$

A list of patterns $\lista{\sequC{\theta_i}{p_i : A_i}}_{i \in 1..n}$ is
compatible if $\forall i, j \in 1..n. i < j \implies
\compatible{\sequC{\theta_i}{p_i : A_i}}{\sequC{\theta_j}{p_j : A_j}}$.
\end{definition}

Following the example above, consider types $\consttype{nil}$ and
$\datatype{\datatype{\consttype{cons}}{\nat}}{(\rectype{\alpha}{\uniontype{\consttype{nil}}{\datatype{\datatype{\consttype{cons}}{\nat}}{\alpha}}})}$
for patterns $\constterm{nil}$ and $\dataterm{\dataterm{\constterm{cons}}{x}}{xs}$
respectively. Compatibility requires no further restriction in this case since
$\cpos{\constterm{nil}}{\dataterm{\dataterm{\constterm{cons}}{x}}{xs}} = \set{\epsilon}$
and $$\lookup{\consttype{nil}}{\epsilon} = \set{\consttype{nil}} \qquad
\lookup{(\datatype{\datatype{\consttype{cons}}{\nat}}{(\rectype{\alpha}{\uniontype{\consttype{nil}}{\datatype{\datatype{\consttype{cons}}{\nat}}{\alpha}}})})}{\epsilon}
= \set{\idatatype}$$ hence $\iPsicomp$ is false and the property gets validated trivially.

On the contrary, recall motivating example (\ref{eq:example:compatib:i}) on Sec.~\ref{sec:intro}.
$\compatible{\sequC{\set{x:\bool}}{\dataterm{\constterm{vl}}{\matchable{x}} :
\datatype{\consttype{vl}}{\bool}}}{\sequC{\set{y : \nat}}{\dataterm{\constterm{vl}}{\matchable{y}}
: \datatype{\consttype{vl}}{\nat}}}$ requires $\datatype{\consttype{vl}}{\nat} \subtypemu
\datatype{\consttype{vl}}{\bool}$ since
$\cpos{\dataterm{\constterm{vl}}{\matchable{x}}}{\dataterm{\constterm{vl}}{\matchable{y}}}
= \varnothing$ (\ie $\iPsicomp$ is trivially true). This actually fails because $\nat
\not\subtypemu \bool$. Thus, this pattern combination is rejected by rule $\ruleTAbs$.

Types are preserved along reduction. The proof relies crucially on compatibility.

\begin{proposition}[Subject Reduction]
\label{prop:subjectReduction}
If $\sequTDeriv{\Gamma}{s : A}$ and $s \reduce s'$, then
$\sequTDeriv{\Gamma}{s' : A}$.
\end{proposition}

Let the set of \emphdef{values} be defined as $v \Coloneq \dataterm{x}{v_1
\ldots v_n} \mathrel| \dataterm{\constterm{c}}{v_1 \ldots v_n} \mathrel|
(\absterm{p_i}{s_i}{\theta_i})_{i \in 1..n}$. The following property guarantees
that no functional application gets stuck. Essentially this means that, in a
well-typed closed term, a function which is applied to an argument has at least
one branch that is capable of handling it.

\begin{proposition}[Progress]
\label{prop:progress}
If $\sequTDeriv{}{s : A}$ and $s$ is not a value, then $\exists s'$ such that $s
\reduce s'$.
\end{proposition}


%% file: types.tex
\subsection{Types}
\label{sec:typingMu}

In order to ensure that patterns such as
$\dataterm{\matchable{x}}{\matchable{y}}$ decompose only data structures rather
than arbitrary terms, we shall introduce two sorts of typing expressions:
\emph{types} and \emph{datatypes}, the latter being strictly included in the
former. We assume given countably infinite sets $\DataTypeVariable$ of
\emphdef{datatype variables} ($\alpha, \beta, \ldots$), $\TypeVariable_{A}$ of
\emphdef{type variables} $(X, Y, \ldots$) and $\TypeConstant$ of \emphdef{type
constants} $(\consttype{c}, \consttype{d}, \ldots)$. We define $\TypeVariable
\eqdef \TypeVariable_{A} \cup \DataTypeVariable$ and use meta-variables $V, W,
\ldots$ to denote an arbitrary element in it. Likewise, we write $a, b, \ldots$
for elements in $\TypeVariable \cup \TypeConstant$. The sets $\DataType$ of
\emphdef{$\mu$-datatypes} and $\Type$ of \emphdef{$\mu$-types}, resp., are
inductively defined as follows: $$
\begin{array}{l@{\qquad}l}
\begin{array}{rlll}
D & \Coloneq & \alpha              & \quad \text{(datatype variable)} \\
  & |        & \consttype{c}       & \quad \text{(atom)} \\
  & |        & \datatype{D}{A}     & \quad \text{(compound)} \\
  & |        & \uniontype{D}{D}    & \quad \text{(union)} \\
  & |        & \rectype{\alpha}{D} & \quad \text{(recursion)} 
\end{array}
&
\begin{array}{rlll}
A & \Coloneq & X                & \quad \text{(type variable)} \\
  & |        & D                & \quad \text{(datatype)} \\
  & |        & \functype{A}{A}  & \quad \text{(type abstraction)} \\
  & |        & \uniontype{A}{A} & \quad \text{(union)} \\
  & |        & \rectype{X}{A}   & \quad \text{(recursion)}
\end{array}
\end{array} $$

\begin{remark}
A type of the form $\rectype{\alpha}{A}$ is not valid in general since it may
produce invalid unfoldings. For example,
$\rectype{\alpha}{\functype{\alpha}{\alpha}} =
\functype{(\rectype{\alpha}{\functype{\alpha}{\alpha}})}{(\rectype{\alpha}{\functype{\alpha}{\alpha}})}$,
which fails to preserve sorting.
On the other hand, types of the form $\rectype{X}{D}$ are not necessary since
they denote the solution to the equation $X = D$, hence $X$ is a
variable representing a datatype, a role already fulfilled by $\alpha$.
\end{remark}

We consider $\iuniontype$ to bind tighter than $\ifunctype$, while $\idatatype$
binds tighter than $\iuniontype$. \Eg 
$\functype{\uniontype{\datatype{D}{A}}{A'}}{B}$ means
$\functype{(\uniontype{(\datatype{D}{A})}{A'})}{B}$. We write $A \neq \iuniontype$
to mean that the root symbol of $A$ is different from $\iuniontype$; and similarly
with the other type constructors.
Expressions such as $\uniontype{\uniontype{A_1}{\ldots}}{A_n}$ will be
abbreviated $\maxuniontype{i \in 1..n}{A_i}$; this is sound since $\mu$-types
will be considered modulo associativity of $\iuniontype$. A type of the form
$\maxuniontype{i \in 1..n}{A_i}$  where each $A_i \neq \iuniontype$, $i \in
1..n$, is called a \emphdef{maximal union}. We often write $\rectype{V}{A}$ to
mean either $\rectype{\alpha}{D}$ or $\rectype{X}{A}$. A \emphdef{non-union
$\mu$-type} $A$ is a $\mu$-type of one of the following forms: $\alpha$,
$\consttype{c}$, $\datatype{D}{A'}$, $X$, $\functype{A'}{A''}$ or $\rectype{V}{B}$
with $B$ a non-union $\mu$-type. We assume $\mu$-types are
\emphdef{contractive}: $\rectype{V}{A}$ is contractive if $V$ occurs in $A$
only under a type constructor $\ifunctype$ or $\idatatype$, if at all. 
For
instance,  $\rectype{X}{\functype{X}{\consttype{c}}}$, $\rectype{X}{\functype{X}{X}}$ and
$\rectype{X}{\datatype{\consttype{c}}{\uniontype{X}{X}}}$ are contractive
while $\rectype{X}{X}$ and $\rectype{X}{\uniontype{X}{X}}$ are not. We
henceforth redefine $\Type$ to be the set of \emphdef{contractive $\mu$-types}.

\begin{figure} $$
\begin{array}{c}
\Rule{}
     {\sequTE{}{A \eqtypemu A}}
     {\ruleEqmuRefl}
\quad
\Rule{\sequTE{}{A \eqtypemu B}
      \quad
      \sequTE{}{B \eqtypemu C}}
     {\sequTE{}{A \eqtypemu C}}{\ruleEqmuTrans}
\quad
\Rule{\sequTE{}{A \eqtypemu B}}
     {\sequTE{}{B \eqtypemu A}}
     {\ruleEqmuSymm}
\\
\\
\Rule{\sequTE{}{D \eqtypemu D'}
      \quad
      \sequTE{}{A \eqtypemu A'}}
     {\sequTE{}{\datatype{D}{A} \eqtypemu \datatype{D'}{A'}}}
     {\ruleEqmuComp}
\qquad
\Rule{\sequTE{}{A \eqtypemu A'}
      \quad
      \sequTE{}{B \eqtypemu B'}}
     {\sequTE{}{\functype{A}{B} \eqtypemu \functype{A'}{B'}}}
     {\ruleEqmuFunc}
\\
\\
\Rule{}
     {\sequTE{}{\uniontype{A}{A} \eqtypemu A}}
     {\ruleEqmuUnionIdem}
\qquad
\Rule{}
     {\sequTE{}{\uniontype{A}{B} \eqtypemu \uniontype{B}{A}}}
     {\ruleEqmuUnionComm}
\\
\\
\Rule{}
     {\sequTE{}{\uniontype{A}{(\uniontype{B}{C})} \eqtypemu \uniontype{(\uniontype{A}{B})}{C}}}
     {\ruleEqmuUnionAssoc}
\\
\\
\Rule{\sequTE{}{A \eqtypemu A'}
      \quad
      \sequTE{}{B \eqtypemu B'}}
     {\sequTE{}{\uniontype{A}{B} \eqtypemu \uniontype{A'}{B'}}}
     {\ruleEqmuUnion}
\qquad
\Rule{\sequTE{}{A \eqtypemu B}}
     {\sequTE{}{\rectype{V}{A} \eqtypemu \rectype{V}{B}}}
     {\ruleEqmuRec}
\\
\\
\Rule{}
     {\sequTE{}{\rectype{V}{A} \eqtypemu \substitute{V}{\rectype{V}{A}}{A}}}
     {\ruleEqmuFold}
\quad
\Rule{\sequTE{}{A \eqtypemu \substitute{V}{A}{B}}
      \quad
      \rectype{V}{B} \text{ contractive}}
     {\sequTE{}{A \eqtypemu \rectype{V}{B}}}
     {\ruleEqmuContr}
\end{array} $$
\caption{Type equivalence for $\mu$-types.}
\label{fig:equivalenceSchemesMu}
\end{figure}

\begin{figure} $$
\begin{array}{c}
\Rule{}
     {\sequTE{\Sigma}{A \subtypemu A}}
     {\ruleSubmuRefl}
\qquad
\Rule{}
     {\sequTE{\Sigma, V \subtypemu W}{V\subtypemu W}}
     {\ruleSubmuHyp}
\qquad
\Rule{\sequTE{}{A \eqtypemu B}}
     {\sequTE{\Sigma}{A \subtypemu B}}
     {\ruleSubmuEq}
\\
\\
\Rule{\sequTE{\Sigma}{A \subtypemu B} 
      \quad
      \sequTE{\Sigma}{B \subtypemu C}}
     {\sequTE{\Sigma}{A \subtypemu C}}
     {\ruleSubmuTrans} 
\qquad
\Rule{\sequTE{\Sigma}{D \subtypemu D'} 
      \quad 
      \sequTE{\Sigma}{A \subtypemu A'}}
     {\sequTE{\Sigma}{\datatype{D}{A} \subtypemu \datatype{D'}{A'}}}
     {\ruleSubmuComp}
\\
\\
\Rule{\sequTE{\Sigma}{A \subtypemu A'}
      \quad 
      \sequTE{\Sigma}{B \subtypemu B'}}
     {\sequTE{\Sigma}{\functype{A'}{B} \subtypemu \functype{A}{B'}}}
     {\ruleSubmuFunc} 
\qquad
\Rule{\sequTE{\Sigma}{A \subtypemu C}
      \quad 
      \sequTE{\Sigma}{B \subtypemu C}}
     {\sequTE{\Sigma}{\uniontype{A}{B} \subtypemu C}}
     {\ruleSubmuUnionL}
\\
\\
\Rule{\sequTE{\Sigma}{A \subtypemu B}}
     {\sequTE{\Sigma}{A \subtypemu \uniontype{B}{C}}}
     {\ruleSubmuUnionRL}
\qquad
\Rule{\sequTE{\Sigma}{A \subtypemu C}}
     {\sequTE{\Sigma}{A \subtypemu \uniontype{B}{C}}}
     {\ruleSubmuUnionRR} 
\\
\\
\Rule{\sequTE{\Sigma, V \subtypemu W}{A \subtypemu B}
      \quad
      W \notin \fv{A} \quad V \notin \fv{B}
     }
     {\sequTE{\Sigma}{\rectype{V}{A} \subtypemu \rectype{W}{B}}}
     {\ruleSubmuRec}
\end{array} $$
\caption{Strong subtyping for $\mu$-types.}
\label{fig:subtypingSchemesMu}
\end{figure}

$\mu$-types come equipped with a notion of \emphdef{type equivalence}
$\eqtypemu$ (Fig.~\ref{fig:equivalenceSchemesMu}) and \emphdef{subtyping}
$\subtypemu$ (Fig.~\ref{fig:subtypingSchemesMu}). In
Fig.~\ref{fig:subtypingSchemesMu} a subtyping context $\Sigma$ is a set of
assumptions over type variables of the form $V \subtypemu W$ with $V, W  \in
\TypeVariable$. $\ruleEqmuRec$ actually encodes two rules, one for datatypes
($\rectype{\alpha}{D}$) and one for arbitrary types ($\rectype{X}{A}$).
Likewise for $\ruleEqmuFold$ and $\ruleEqmuContr$.
Regarding the subtyping rules, we adopt those for union
of~\cite{DBLP:conf/csl/Vouillon04}. It should be noted that the na\"ive variant
of $\ruleSubmuRec$ in which $\sequTE{\Sigma}{\rectype{V}{A} \subtypemu
\rectype{V}{B}}$ is deduced from \sequTE{\Sigma}{A\subtypemu B}, is known to be
unsound~\cite{DBLP:journals/toplas/AmadioC93}. We often abbreviate $\sequTE{}{A
\subtypemu B}$ as $A \subtypemu B$.


%% file: invertibility.tex
\section{Checking Equivalence and Subtyping}
\label{sec:invertibility}

As mentioned in the related work, there are roughly two approaches to
implementing equivalence and subtype checking in the presence of recursive
types, one based on automata theory and another based on  coinductive
characterizations of the associated relations. The former leads to efficient
algorithms~\cite{DBLP:journals/iandc/PalsbergZ01} while the latter is more
abstract in nature and hence closer to the formalism itself although they may
not be as efficient. In the particular case of subtyping for recursive types in
the presence of ACI operators, the automata approach of~\cite{DBLP:journals/iandc/PalsbergZ01} is known
not to be applicable~\cite{Zhao:thesis} while the coinductive approach,
developed in this section, yields a correct algorithm. In
Sec.~\ref{sec:efficientTypeChecking} we explore an alternative approach for
subtyping based on automata inspired
from~\cite{DBLP:conf/tlca/CosmoPR05}.  We next further describe the
reasoning behind the coinductive approach.

\subsection{Preliminaries}

Consider \emphdef{type constructors} $\idatatype$ and $\ifunctype$ together
with \emphdef{type connector} $\iuniontype$ and the ranked alphabet
$\mathfrak{L} \eqdef \set{a^0 \mathrel| a \in \TypeVariable \cup \TypeConstant}
\cup \set{\idatatype^2, \ifunctype^2, \iuniontype^2}$. We write $\Tree$ for the
set of (possibly) \emphdef{infinite types} with symbols in $\mathfrak{L}$.
This is a standard construction~\cite{Terese:2003,DBLP:journals/tcs/Courcelle83} given
by the metric completion based on a simple depth function measuring the
distance from the root to the minimum conflicting node in two trees.
Perhaps worth mentioning is that the type connector $\iuniontype$ does not
contribute to the depth (hence the reason for calling it a connector rather
than a constructor) excluding types consisting of infinite branches of
$\iuniontype$, such as
$\uniontype{(\uniontype{\ldots}{\ldots})}{(\uniontype{\ldots}{\ldots})}$, from
$\Tree$. A tree is \emphdef{regular} if the set of all its subtrees is finite.
We shall always work with regular trees and also denote them $\Tree$.

\begin{definition}
\label{def:treeCut}
The \emphdef{truncation} of a tree $\tA$ at depth $k \in \Natural$ (notation
$\cut{\tA}{k}$) is defined inductively\footnote{This definition is
well-founded~\cite{DBLP:journals/entcs/VisoBA16}.} as follows: $$
\begin{array}{r@{\quad\eqdef\quad}l@{\quad}l}
\cut{\tA}{0}                          & \circ \\
\cut{a}{k+1}                          & a                                   & \text{for $a \in \TypeVariable \cup \TypeConstant$} \\
\cut{(\tA_1 \star \tA_2)}{k+1}        & \cut{\tA_1}{k} \star \cut{\tA_2}{k} & \text{for $\star \in \set{\idatatype, \ifunctype}$} \\
\cut{(\uniontype{\tA_1}{\tA_2})}{k+1} & \uniontype{\cut{\tA_1}{k+1}}{\cut{\tA_2}{k+1}}
\end{array} $$ where $\circ \in \TypeConstant$ is a distinguished type
constant used to identify the nodes where the tree was truncated.
\end{definition}

\begin{definition}
\label{def:toBTree}
The function $\toBTree{\bullet} : \Type \to \Tree$, mapping $\mu$-types to
types, is defined inductively as follows: $$
\begin{array}{r@{\quad\eqdef\quad}l@{\quad}l}
\toBTree{a}(\epsilon)             & a \\
\toBTree{A_1 \star A_2}(\epsilon) & \star              & \text{for $\star \in \set{\idatatype, \ifunctype, \iuniontype}$} \\
\toBTree{A_1 \star A_2}(i\pi)     & \toBTree{A_i}(\pi) & \text{for $\star \in \set{\idatatype, \ifunctype, \iuniontype}$} \\
\toBTree{\rectype{V}{A}}(\pi)     & \toBTree{\substitute{V}{\rectype{V}{A}}{A}}(\pi) \\
\end{array} $$
\end{definition}

Coinductive characterizations of subsets of $\Type\times\Type$ whose generating
function $\Phi$ is \emph{invertible} admit a simple (although not necessarily
efficient) algorithm for subtype membership checking and consists in ``running
$\Phi$ backwards''~\cite[Sec.~21.5]{DBLP:books/daglib/0005958}. This strategy is
supported by the fact that contractiveness of $\mu$-types guarantees a finite
state space to explore (\ie unfolding these types results in regular trees);
invertibility further guarantees that there is at most one way in which a
member of $\nu\Phi$, the greatest fixed-point of $\Phi$, can be generated.
Invertibility of $\Phi : \powerset{\Type\times\Type} \to
\powerset{\Type\times\Type}$ means that for any $\pair{A}{B} \in \Type$, the set
$\set{\X \in \powerset{\Type\times\Type} \mathrel| \pair{A}{B} \in \Phi(\X)}$
is either empty or contains a unique member.

\subsection{Equivalence Checking}

Fig.~\ref{fig:eqtypeSchemesAl} presents a coinductive definition of
type equality over $\mu$-types. This relation $\eqtypeal$ is defined
by means of rules that  are interpreted coinductively
(indicated by the double lines). The rule $\ruleEqalUnion$ makes use
of functions $f$ and $g$ to encode the ACI nature of $\oplus$. Letters $\ctxtlra, \ctxtlrb$,
used in rules \ruleEqalRecL\ and \ruleEqalRecR, denote contexts of the form: $$A_1
\iuniontype \ldots A_{i-1} \iuniontype \Box \iuniontype A_{i+1} \iuniontype
\ldots \iuniontype A_n$$ where $\Box$ denotes the hole of the context, $A_j
\neq \iuniontype$ for all $j \in 1..n \setminus i$ and $A_l \neq \irectype$ for
all $l \in 1..i-1$. Note that, in particular, $\ctxtlra$ may take the form
$\Box$. These contexts help identify the first occurrence of a $\irectype$
constructor within a maximal union. In turn, this allows us to guarantee the invertibility of the
generating function associated to these rules.

\begin{proposition}
\label{prop:eqtypealInvertibility}
The generating function associated with the rules of Fig.~\ref{fig:eqtypeSchemesAl} is invertible.
\end{proposition}

Moreover, $\eqtypeal$ coincides with $\eqtypemu$:

\begin{proposition}
\label{prop:eqtypealSoundnessAndCompleteness}
$A \eqtypeal B$ iff $A \eqtypemu B$.
\end{proposition}

This will allow us to check $A \eqtypemu B$ by using the invertibility of the
generating function (implicit in the rules of Fig.~\ref{fig:eqtypeSchemesAl})
for $\eqtypeal$. The proof of Prop.~\ref{prop:eqtypealSoundnessAndCompleteness}
relies on an intermediate relation $\eqtypeco$ over the possibly infinite trees
resulting from the complete unfolding of $\mu$-types. This relation is defined
using the same rules as in Fig.~\ref{fig:eqtypeSchemesAl} except for two
important differences: 1) the relation is defined over regular trees in \Tree,
and 2) rules $\ruleEqalRecL$ and $\ruleEqalRecR$ are dropped.

\begin{figure} $$
\begin{array}{c}
\RuleCo{}
       {a \eqtypeal a}
       {\ruleEqalRefl}
\\
\\
\RuleCo{D \eqtypeal D' \quad A \eqtypeal A'}
       {\datatype{D}{A} \eqtypeal \datatype{D'}{A'}}
       {\ruleEqalComp}
\qquad
\RuleCo{A \eqtypeal A'
        \quad 
        B \eqtypeal B'}
       {\functype{A}{B} \eqtypeal \functype{A'}{B'}}
       {\ruleEqalFunc}
\\
\\
\RuleCo{\ctxtlra[\substitute{V}{\rectype{V}{A}}{A}] \eqtypeal B
       \quad
       }
       {\ctxtlra[\rectype{V}{A}] \eqtypeal B}
       {\ruleEqalRecL}
\\
\\
\RuleCo{A \eqtypeal \ctxtlrb[\substitute{W}{\rectype{W}{B}}{B}]
        \quad
        A \neq \ctxtlra[\rectype{V}{C}]
       }
       {A \eqtypeal \ctxtlrb[\rectype{W}{B}]}
       {\ruleEqalRecR}
\\
\\
\RuleCo{\begin{array}{ll}
          A_i \eqtypeal B_{f(i)} & \quad f : 1..n \to 1..m \\
          A_{g(j)} \eqtypeal B_j & \quad g : 1..m \to 1..n
        \end{array}
        \quad
        A_i, B_j \neq \irectype,\iuniontype
        \quad
        n + m > 2}
       {\maxuniontype{i \in 1..n}{A_i} \eqtypeal \maxuniontype{j \in 1..m}{B_j}}
       {\ruleEqalUnion}
\end{array} $$
\caption{Coinductive axiomatization of type equality for contractive $\mu$-types.}
\label{fig:eqtypeSchemesAl}
\end{figure}

The proof is structured as follows. First we characterize equality of
$\mu$-types in terms of equality of their infinite unfoldings~\cite{DBLP:journals/entcs/VisoBA16}:

\begin{proposition}
\label{prop:eqtypeSoundnessAndCompleteness}
$A \eqtypemu B$ iff $\toBTree{A} \eqtypeco \toBTree{B}$.
\end{proposition}

The proof of Prop.~\ref{prop:eqtypealSoundnessAndCompleteness} thus reduces to
showing that $A \eqtypeal B$ coincides with $\toBTree{A} \eqtypeco
\toBTree{B}$. In order to do so, we appeal to the following result that states
that inspecting all finite truncations suffices to determine whether
$\toBTree{A} \eqtypeco \toBTree{B}$ holds:

\begin{lemma}
\label{lem:cutEquivalenceCo}
$\forall k \in \Natural.\cut{\tA}{k} \eqtypeco \cut{\tB}{k}$ iff $\tA \eqtypeco \tB$.
\end{lemma}

\begin{proof}
This is proved by showing that the relations $\S \eqdef \set{\pair{\tA}{\tB}
\mathrel| \forall k \in \Natural. \cut{\tA}{k} \eqtypeco \cut{\tB}{k}}$ and
$\R \eqdef \set{\pair{\cut{\tA}{k}}{\cut{\tB}{k}} \mathrel| \tA \eqtypeco \tB,
k \in \Natural}$ are $\Phieqtypeco$-dense. Then we conclude by the coinductive
principle. For full details refer to~\cite{DBLP:journals/entcs/VisoBA16}.
\end{proof}

The proofs of Lem.~\ref{lem:eqtypealSoundnessAndCompleteness}
and~\ref{lem:subtypealSoundnessAndCompleteness} rely on some key lemmas which
we now state, preceded by some preliminary notions. Let $\cA,\cB,\ldots$ refer
to (multi-hole) $\irectype\iuniontype$-contexts: $$\cA \Coloneq \Box \mathrel|
\rectype{V}{\cA} \mathrel| \uniontype{\cA}{\cA}$$

When $\cA$ consists solely of $\irectype$ type constructors, then we say that
it is a $\irectype$-context. As for types, we use notation $\maxuniontype{i \in
1..n}{\cA_i}$ when referring to a context irrespective of how unions are
associated within it. Any $A \in \Type$, can be uniquely written as a
\emph{maximal $\irectype\iuniontype$-context} $A = \cA[A_1, \ldots, A_n]$,
where each $A_i \neq \irectype,\iuniontype$, $i \in 1..n$. The
\emph{$\irectype\iuniontype$-depth} of $A$ is the depth of $\cA$.

\begin{definition}
\label{def:projection}
The \emphdef{projection} of a hole's position $\pi$ in $\cA$, notation
$\project{\pi}{\cA[\vec{A}]}$, is defined by induction on $\cA$ as follows: $$
\begin{array}{r@{\quad\eqdef\quad}l@{\quad}l}
\project{\epsilon}{\Box[A]}                                      & A \\
\project{1\pi}{(\rectype{V}{\cA[\vec{A}]})}                      & \project{\pi}{\substitute{V}{\rectype{V}{\cA[\vec{A}]}}{\cA[\vec{A}]}} \\
\project{i\pi}{(\uniontype{\cA_1[\vec{A_1}]}{\cA_2[\vec{A_2}]})} & \project{\pi}{\cA_i[\vec{A_i}]} & i = 1,2
\end{array} $$
\end{definition}

We write $\eraserec{\cA}$ for the $\irectype\iuniontype$-context resulting from
dropping all the $\irectype$ type constructors from $\cA$: $$
\begin{array}{r@{\quad\eqdef\quad}l}
\eraserec{\Box}                       & \Box \\
\eraserec{(\rectype{V}{\cA})}         & \eraserec{\cA} \\
\eraserec{(\uniontype{\cA_1}{\cA_n})} & \uniontype{\eraserec{\cA_1}}{\eraserec{\cA_n}}
\end{array} $$

\begin{lemma}
\label{lem:forgettingMu}
Let $\cA$ be a $\irectype\iuniontype$-context. Then, $$\toBTree{\cA[\vec{A}]} =
\eraserec{\cA}[\toBTree{\project{\pi_1}{\cA[\vec{A}]}}, \ldots,
\toBTree{\project{\pi_n}{\cA[\vec{A}]}}]$$ where $\pi_i$ is the position of the
$i$-th hole (in the left-to-right order) in $\cA$.
\end{lemma}

\begin{proof}
Induction on the context $\cA$.
\begin{itemize}
  \item $\cA = \Box$: $$\toBTree{\cA[A]} = \toBTree{\Box[A]} = \toBTree{A} =
  \Box[\toBTree{A}] = \eraserec{\Box}[\toBTree{\project{\epsilon}{\Box[A]}}] =
  \eraserec{\cA}[\toBTree{\project{\epsilon}{\cA[A]}}]$$
  
  \item $\cA = \rectype{V}{\cA'}$: $$
\begin{array}{rcl@{\quad}l}
\toBTree{\cA[\vec{A}]} & = & \toBTree{\rectype{V}{\cA'}[\vec{A}]} \\
                       & = & \toBTree{\substitute{V}{\rectype{V}{\cA'}[\vec{A}]}{\cA'[\vec{A}]}} & \text{Def.~\ref{def:toBTree}} \\
                       & = & \toBTree{(\cA'[\substitute{V}{\rectype{V}{\cA'}[\vec{A}]}{\vec{A}}])} \\
                       & = & \eraserec{\cA'}[\toBTree{\project{\pi_1}{\cA'[\substitute{V}{\rectype{V}{\cA'}[\vec{A}]}{\vec{A}}]}}, \ldots, \\
                       &   & \hphantom{\eraserec{\cA'}[} \toBTree{\project{\pi_n}{\cA'[\substitute{V}{\rectype{V}{\cA'}[\vec{A}]}{\vec{A}}]}}] & \text{inductive hypothesis} \\
                       & = & \eraserec{\cA'}[\toBTree{\substitute{V}{\rectype{V}{\cA'}[\vec{A}]}{\project{\pi_1}{\cA'[\vec{A}]}}}, \ldots, \\
                       &   & \hphantom{\eraserec{\cA'}[} \toBTree{\substitute{V}{\rectype{V}{\cA'}[\vec{A}]}{\project{\pi_n}{\cA'[\vec{A}]}}}] \\
                       & = & \eraserec{\cA}[\toBTree{\project{1\pi_1}{\cA[\vec{A}]}}, \ldots, \toBTree{\project{1\pi_n}{\cA[\vec{A}]}}] & \text{Def.~\ref{def:projection} and $\eraserec{\cA}$}
\end{array} $$
  
  \item $\cA = \uniontype{\cA_1}{\cA_2}$: $$
\begin{array}{rcl@{\quad}l}
\toBTree{\cA[\vec{A}]} & = & \toBTree{\uniontype{\cA_1[\vec{A_1}]}{\cA_2[\vec{A_2}]}} & \vec{A} = \vec{A_1},\vec{A_2} \\
                       & = & \uniontype{\toBTree{\cA_1[\vec{A_1}]}}{\toBTree{\cA_2[\vec{A_2}]}} & \text{Def.~\ref{def:toBTree}} \\
                       & = & \eraserec{\cA_1}[\toBTree{\project{\pi'_1}{\cA_1[\vec{A_1}]}}, \ldots, \toBTree{\project{\pi'_{n_1}}{\cA_1[\vec{A_1}]}}] \mathrel{\iuniontype} \\
                       &   & \eraserec{\cA_2}[\toBTree{\project{\pi''_1}{\cA_2[\vec{A_2}]}}, \ldots, \toBTree{\project{\pi''_{n_2}}{\cA_2[\vec{A_2}]}}] & \text{inductive hypothesis} \\
                       & = & \eraserec{\cA}[\toBTree{\project{\pi'_1}{\cA_1[\vec{A_1}]}}, \ldots, \toBTree{\project{\pi'_{n_1}}{\cA_1[\vec{A_1}]}}, \\
                       &   & \hphantom{\eraserec{\cA}[} \toBTree{\project{\pi''_1}{\cA_2[\vec{A_2}]}}, \ldots, \toBTree{\project{\pi''_{n_2}}{\cA_2[\vec{A_2}]}}] & \text{Def. of $\eraserec{\cA}$} \\
                       & = & \eraserec{\cA}[\toBTree{\project{\pi_1}{\cA[\vec{A}]}}, \ldots, \toBTree{\project{\pi_n}{\cA[\vec{A}]}}] & \text{Def.~\ref{def:projection}}
\end{array} $$ where $\lista{\pi_1, \ldots, \pi_n} = \lista{1\pi'_1, \ldots,
  1\pi'_{n_1}, 2\pi''_1, \ldots, 2\pi''_{n_2}}$.
\end{itemize}
\end{proof}

We also need to deduce the form of a type $A$ whose associated tree is a
maximal union type.

\begin{lemma}
\label{lem:decomposeUnion}
Let $A \in \Type$, $\tA_i \in \Tree$ for $i \in 1..n$ such that $\toBTree{A} =
\maxuniontype{i \in 1..n}{\tA_i}$ and $\tA_i \neq \iuniontype$. Then, there
exists $n' \leq n$, $\cA, \cA_1, \ldots, \cA_{n'}$ maximal
$\iuniontype$-contexts, $A_1, \ldots, A_{n'} \neq \iuniontype$ and functions
$s,t : 1..n' \to 1..n$ such that $$A = \cA[\vec{A}] \qquad\text{and}\qquad
\toBTree{A_l} = \cA_l[\tA_{s(l)}, \ldots, \tA_{t(l)}]$$ for every $l \in
1..n'$.
\end{lemma}

\begin{proof}
As noted before, every $A \in \Type$ can be uniquely written as a maximal
$\irectype\iuniontype$-context $\cA'$ and types $A'_1, \ldots, A'_m$ such that
$A'_j \neq \irectype,\iuniontype$ for all $j \in 1..m$, \ie $A =
\cA'[\vec{A'}]$. It is immediate to see, by Def.~\ref{def:toBTree}
and~\ref{def:projection}, that $m = n$ and
\begin{equation}
\label{eq:decomposeUnion:tree}
\toBTree{\project{\pi_i}{\cA'[\vec{A'}]}} = \tA_i
\end{equation}
for every $i \in 1..n$.

Then, we can decompose $\cA'$ into two parts, namely a maximal
$\iuniontype$-context $\cA$ with $n' \leq n$ holes, and multiple
$\irectype\iuniontype$-contexts $\cA'_1, \ldots, \cA'_{n'}$ such that $\cA' =
\cA[\vec{\cA'}]$ and each $\cA'_l$ is either $\Box$ or starts with $\irectype$
for every $l \in 1..n'$. Thus, there exists types $A_1, \ldots, A_{n'}$ and
functions $s,t : 1..n' \to 1..n$ such that
\begin{equation}
\label{eq:decomposeUnion:ctxt}
A_l = \cA'_l[A'_{s(l)},\ldots,A'_{t(l)}] \qquad\text{and}\qquad A = \cA[\vec{A}]
\end{equation}
Note that $A_l \neq \iuniontype$ since $\cA'_l = \Box,\irectype$ for every
$l \in 1..n'$.

Finally, let $\rho_1, \ldots, \rho_{n_l}$ be the positions of the holes in
$\cA'_l$, we have $$
\begin{array}{r@{\quad=\quad}l@{\quad}l}
\toBTree{A_l}
& \toBTree{\cA'_l[A'_{s(l)},\ldots,A'_{t(l)}]}
 & \text{by (\ref{eq:decomposeUnion:ctxt})-left} \\
& \eraserec{\cA'_l}[\toBTree{\project{\rho_1}{\cA'_l[A'_{s(l)},     \ldots, A'_{t(l)}]}},
            \ldots, \toBTree{\project{\rho_{n_l}}{\cA'_l[A'_{s(l)}, \ldots, A'_{t(l)}]}}]
 & \text{by Lem.~\ref{lem:forgettingMu}} \\
& \eraserec{\cA'_l}[\toBTree{\project{\pi_{s(l)}}{\cA'[\vec{A'}]}}, \ldots,
                    \toBTree{\project{\pi_{t(l)}}{\cA'[\vec{A'}]}}]
 & \cA' = \cA[\vec{\cA'}] \\
& \eraserec{\cA'_l}[\tA_{s(l)}, \ldots, \tA_{t(l)}]
 & \text{by (\ref{eq:decomposeUnion:tree})}
\end{array} $$
for every $l \in 1..n'$. Then, it is enough to take $\cA_l = \eraserec{\cA'_l}$.
\end{proof}

Finally, we show that $A \eqtypeal B$ can be determined by checking all finite
truncations of its associated regular trees:

\begin{lemma}
\label{lem:eqtypealSoundnessAndCompleteness}
$A \eqtypeal B$ iff $\forall k \in \Natural. \cut{\toBTree{A}}{k} \eqtypeco
\cut{\toBTree{B}}{k}$.
\end{lemma}

\begin{proof}
For this proof we use notation $\card{\irectype}{A}$ for the number of
occurrences of the $\irectype$ type constructor in the maximal
$\irectype\iuniontype$-context $\cA$ defining $A = \cA[\vec{A}]$.

$\Rightarrow)$ Given $A \eqtypeal B$, there exists $\S \in
\powerset{\Type\times\Type}$ such that $\S \subseteq \Phieqtypeal(\S)$ and
$\pair{A}{B} \in \S$. Define: $$\R(\S) \eqdef
\set{\pair{\cut{\toBTree{A'}}{k}}{\cut{\toBTree{B'}}{k}} \mathrel|
\pair{A'}{B'} \in \S, k \in \Natural}$$ We show that $\R(\S) \subseteq
\Phieqtypeco(\R(\S))$.

Define $\S_c \eqdef \set{\pair{A'}{B'} \in \S \mathrel| c = \card{\irectype}{A}
+ \card{\irectype}{B}}$. For each $c \in \Natural$, we prove: $$\pair{A'}{B'}
\in \S_c \implies \forall k \in \Natural.
\pair{\cut{\toBTree{A'}}{k}}{\cut{\toBTree{B'}}{k}} \in \Phieqtypeco(\R(\S))$$

Note that if $k = 0$, the result follows directly from $\ruleEqcoRefl$ since
$\pair{\cut{\toBTree{A'}}{k}}{\cut{\toBTree{B'}}{k}} = \pair{\circ}{\circ}
\in \Phieqtypeco(\R(\S))$ trivially. So we in fact prove $$\pair{A'}{B'} \in
\S_c \implies \forall k > 0 \in \Natural.
\pair{\cut{\toBTree{A'}}{k}}{\cut{\toBTree{B'}}{k}} \in \Phieqtypeco(\R(\S))$$

We proceed by induction on $c$.
\begin{itemize}
  \item $c = 0$. Since $\pair{A}{B} \in \S_0$ implies $\pair{A}{B} \in \S
  \subseteq \Phieqtypeal(\S)$, one of the following cases must hold:
  \begin{itemize}
    \item $\pair{A}{B} = \pair{a}{a}$. Then
    $\pair{\cut{\toBTree{A}}{k}}{\cut{\toBTree{A}}{k}} = \pair{a}{a}$ for every
    $k > 0$ and we conclude by $\ruleEqcoRefl$, $\Phieqtypeco(\R(\S))$.
    
    \item $\pair{A}{B} = \pair{\functype{A'}{B'}}{\functype{A''}{B''}}$ with
    $\pair{A'}{A''} \in \S$ and $\pair{B'}{B''} \in \S$. Then, by definition of
    $\R$ we have both
    $\pair{\cut{\toBTree{A'}}{k-1}}{\cut{\toBTree{A''}}{k-1}},
    \pair{\cut{\toBTree{B'}}{k-1}}{\cut{\toBTree{B''}}{k-1}} \in \R(\S)$ for
    every $k > 0$. Finally, by $\ruleEqcoFunc$ and Def.~\ref{def:treeCut},
    $$\pair{\cut{\toBTree{\functype{A'}{B'}}}{k}}{\cut{\toBTree{\functype{A''}{B''}}}{k}}
    \in \Phieqtypeco(\R(\S))$$ for every $k > 0$.
    
    \item $\pair{A}{B} = \pair{\datatype{D}{A'}}{\datatype{D'}{B'}}$ with
    $\pair{D}{D'} \in \S$ and $\pair{A'}{B'} \in \S$. Similarly to the previous
    case, we get $\pair{\cut{\toBTree{D}}{k-1}}{\cut{\toBTree{D'}}{k-1}},
    \pair{\cut{\toBTree{A'}}{k-1}}{\cut{\toBTree{B'}}{k-1}}\in \R(\S)$ for any
    $k > 0$ from Def.~\ref{def:toBTree}, and conclude by applying
    $\ruleEqcoComp$ with Def.~\ref{def:treeCut}.
    
    \item $\pair{A}{B} = \pair{\maxuniontype{i \in 1..n}{A_i}}{\maxuniontype{j
    \in 1..m}{B_j}}$ with $n+m > 2$, $A_i,B_j \neq \irectype,\iuniontype$ and
    there exists functions $f : 1..n \to 1..m, g : 1..m \to 1..n$ s.t.
    $\pair{A_i}{B_{f(i)}} \in \S$ and $\pair{A_{g(j)}}{B_j} \in \S$ for every
    $i \in 1..n, j \in 1..m$.
    
    Once again, we have
    $\pair{\cut{\toBTree{A_i}}{k}}{\cut{\toBTree{B_{f(i)}}}{k}},
    \pair{\cut{\toBTree{A_{g(j)}}}{k}}{\cut{\toBTree{B_j}}{k}} \in \R(\S)$ for
    every $i \in 1..n, j \in 1..m, k > 0$ by Def.~\ref{def:toBTree} and $\R$.
    
    Moreover, since $A_i,B_j \neq \irectype,\iuniontype$ we can assure
    $\toBTree{A_i},\toBTree{B_j} \neq \iuniontype$ too. Thus, we are able to
    apply $\ruleEqcoUnion$ with $f$ and $g$ to conclude,
    $$\pair{\cut{\toBTree{\maxuniontype{i
    \in 1..n}{A_i}}}{k}}{\cut{\toBTree{\maxuniontype{j \in 1..m}{B_j}}}{k}}
    \in \Phieqtypeco(\R(\S))$$ for every $k > 0$.
  \end{itemize}
  
  \item $c > 1$. Then:
  \begin{itemize}
    \item $\pair{A}{B} = \pair{\ctxtlra[\rectype{V}{A'}]}{B}$ with
    $\pair{\ctxtlra[\substitute{V}{\rectype{V}{A'}}{A'}]}{B} \in \S$. By
    contractiveness of $\mu$-types, we have
    $\card{\irectype}{\ctxtlra[\substitute{V}{\rectype{V}{A'}}{A'}]} <
    \card{\irectype}{\ctxtlra[\rectype{V}{A'}]}$. Hence we can apply the
    inductive hypothesis to get
    $$\pair{\cut{\toBTree{\ctxtlra[\substitute{V}{\rectype{V}{A'}}{A'}]}}{k}}{\cut{\toBTree{B}}{k}}
    \in \Phieqtypeco(\R(\S))$$ for every $k > 0$. Moreover, since
    $\toBTree{\ctxtlra[\rectype{V}{A'}]} =
    \toBTree{\ctxtlra[\substitute{V}{\rectype{V}{A'}}{A'}]}$ by
    Def.~\ref{def:toBTree}, we can safely conclude.
    
    \item $\pair{A}{B} = \pair{A}{\ctxtlrb[\rectype{W}{B'}]}$ with
    $\pair{A}{\ctxtlrb[\substitute{W}{\rectype{W}{B'}}{B'}]} \in \S$ and $A
    \neq \ctxtlra[\rectype{V}{C}]$. As before, we conclude directly from the
    inductive hypothesis by resorting to contractiveness of $\mu$-types. Thus,
    $$\pair{\cut{\toBTree{A}}{k}}{\cut{\toBTree{\ctxtlrb[\rectype{W}{B'}]}}{k}}
    \in \Phieqtypeco(\R(\S))$$ for every $k > 0$.
  \end{itemize}
\end{itemize}

$\Leftarrow)$ We show that $\R \eqdef \set{\pair{A'}{B'} \mathrel| A', B' \in
\Type, \forall k \in \Natural. \cut{\toBTree{A'}}{k} \eqtypeco
\cut{\toBTree{B'}}{k}}$ is $\Phieqtypeal$-dense.

Let $\pair{A}{B} \in \R$. We proceed by induction on $c = \card{\irectype}{A} +
\card{\irectype}{B}$. From now on we consider $k > 0$ to avoid the border case
$\pair{\cut{\toBTree{A}}{0}}{\cut{\toBTree{B}}{0}} = \pair{\circ}{\circ}$.
\begin{itemize}
  \item $c = 0$. Since $\pair{\cut{\toBTree{A}}{k}}{\cut{\toBTree{B}}{k}} \in
  \mathbin{\eqtypeco}$ and $\mathbin{\eqtypeco} = \Phieqtypeco(\eqtypeco)$, one
  of the following cases must hold:
  \begin{itemize}
    \item $\pair{\cut{\toBTree{A}}{k}}{\cut{\toBTree{B}}{k}} = \pair{a}{a}$.
    Then, since $c = 0$ it must be the case $A = a = B$. Thus, by
    $\ruleEqalRefl$, $\pair{A}{B} \in \Phieqtypeal(\R)$.
    
    \item $\pair{\cut{\toBTree{A}}{k}}{\cut{\toBTree{B}}{k}} =
    \pair{\functype{\cut{\tA'}{k-1}}{\cut{\tA''}{k-1}}}
         {\functype{\cut{\tB'}{k-1}}{\cut{\tB''}{k-1}}}$
    with
\begin{equation}
\label{eq:eqCompl:func}
\cut{\tA'}{k-1} \eqtypeco \cut{\tB'}{k-1}
\qquad\text{and}\qquad
\cut{\tA''}{k-1} \eqtypeco \cut{\tB''}{k-1}
\end{equation}
    Once again, since the outermost constructor of both $A$ and $B$ is not
    $\irectype$, there exists $A',A'',B',B'' \in \Type$ such that $A =
    \functype{A'}{A''}$ and $B = \functype{B'}{B''}$ with $\toBTree{A'} = \tA'$,
    $\toBTree{A''} = \tA''$, $\toBTree{B'} = \tB'$ and $\toBTree{B''} = \tB''$.
    Moreover, from (\ref{eq:eqCompl:func}) and the definition of $\R$ we get
    $\pair{A'}{B'}, \pair{A''}{B''} \in \R$. We conclude with $\ruleEqalFunc$,
    $\pair{A}{B} \in \Phieqtypeal(\R)$.
    
    \item $\pair{\cut{\toBTree{A}}{k}}{\cut{\toBTree{B}}{k}} =
    \pair{\datatype{\cut{\tA'}{k-1}}{\cut{\tA''}{k-1}}}
         {\datatype{\cut{\tB'}{k-1}}{\cut{\tB''}{k-1}}}$
    with $$\cut{\tA'}{k-1} \eqtypeco \cut{\tB'}{k-1} \qquad\text{and}\qquad
    \cut{\tA''}{k-1} \eqtypeco \cut{\tB''}{k-1}$$ As in the previous case, we
    know there exists $D,D',A',B' \in \Type$ such that $A = \datatype{D}{A'}$,
    $B = \datatype{D'}{B'}$ and $\pair{D}{D'}, \pair{A'}{B'} \in \R$. Here we
    use $\ruleEqalComp$ to conclude.
    
    \item $\pair{\cut{\toBTree{A}}{k}}{\cut{\toBTree{B}}{k}} =
    \pair{\maxuniontype{i \in 1..n}{\cut{\tA_i}{k}}}
         {\maxuniontype{j \in 1..m}{\cut{\tB_j}{k}}}$
    with $n+m > 2$, $\cut{\tA_i}{k}, \cut{\tB_j}{k} \neq \iuniontype$ and there
    exists functions $f : 1..n \to 1..m, g : 1..m \to 1..n$ such that
\begin{equation}
\label{eq:eqCompl:union:hyp}
\cut{\tA_i}{k} \eqtypeco \cut{\tB_{f(i)}}{k}
\qquad\text{and}\qquad
\cut{\tA_{g(j)}}{k} \eqtypeco \cut{\tB_j}{k}
\end{equation}
    for every $i \in 1..n, j \in 1..m$.
    
    By Lem.~\ref{lem:decomposeUnion} on $A$, there exists $n' \leq n$,
    maximal $\iuniontype$-contexts $\cA, \cA_1, \ldots, \cA_{n'}$, types $A_1,
    \ldots, A_{n'} \neq \iuniontype$ and functions $s,t : 1..n' \to 1..n$ such
    that $$A = \cA[\vec{A}] \qquad\text{and}\qquad \toBTree{A_l} =
    \cA_l[\tA_{s(l)}, \ldots, \tA_{t(l)}]$$ for every $l \in 1..n'$. Moreover,
    since $\card{\irectype}{A} = 0$ and $\cA$ is maximal, it can only be the
    case that $n' = n$ and $\cA_l = \Box$ for every $l \in 1..n'$. Thus, $s =
    t = id$ and
\begin{equation}
\label{eq:eqCompl:union:a}
\toBTree{A_i} = \tA_i
\qquad\text{with}\qquad
A_i \neq \irectype,\iuniontype
\end{equation}
    for every $i \in 1..n$.
    
    With a similar analysis we have $B = \cB[\vec{B}]$ with $\cB$ a maximal
    $\iuniontype$-context, and 
\begin{equation}
\label{eq:eqCompl:union:b}
\toBTree{B_j} = \tB_j
\qquad\text{with}\qquad
B_j \neq \irectype,\iuniontype
\end{equation}
    for every $j \in 1..m$. Then, from (\ref{eq:eqCompl:union:hyp}),
    (\ref{eq:eqCompl:union:a})-left and (\ref{eq:eqCompl:union:b})-left we get
    $$\cut{\toBTree{A_i}}{k} \eqtypeco \cut{\toBTree{B_{f(i)}}}{k}
    \qquad\text{and}\qquad \cut{\toBTree{A_{g(j)}}}{k} \eqtypeco
    \cut{\toBTree{B_j}}{k}$$ for every $i \in 1..n, j \in 1..m$. Finally, we
    apply $\ruleEqalUnion$ with $f$, $g$, (\ref{eq:eqCompl:union:a})-right and
    (\ref{eq:eqCompl:union:b})-right to conclude
    $\pair{\maxuniontype{i \in 1..n}{A_i}}
          {\maxuniontype{j \in 1..m}{B_j}} \in \Phieqtypeal(\R)$.
  \end{itemize}
  
  \item $c > 0$. Then we need to distinguish two cases:
  \begin{itemize}
    \item $\card{\irectype}{A} = 0$. Then, it is necessarily the case
    $\card{\irectype}{B} > 0$. Thus, $B = \ctxtlrb[\rectype{W}{B'}]$ and, by
    Def.~\ref{def:toBTree}, we have
    $$\cut{\toBTree{\ctxtlrb[\rectype{W}{B'}]}}{k} =
    \cut{\toBTree{\ctxtlrb[\substitute{W}{\rectype{W}{B'}}{B'}]}}{k}$$ for
    every $k \in \Natural$. Moreover, by contractiveness of $\mu$-types we
    know that $$\card{\irectype}{\ctxtlrb[\substitute{W}{\rectype{W}{B'}}{B'}]}
    < \card{\irectype}{\ctxtlrb[\rectype{W}{B'}]}$$ and we can apply the
    inductive hypothesis to get
    $\pair{A}{\ctxtlrb[\substitute{W}{\rectype{W}{B'}}{B'}]} \in \R$.
    
    Note that $\card{\irectype}{A} = 0$ implies $A \neq
    \ctxtlra[\rectype{V}{A'}]$. Then, we are under the hypothesis of rule
    $\ruleEqalRecR$, and we conclude $\pair{A}{\ctxtlrb[\rectype{W}{B'}]} \in
    \Phieqtypeal(\R)$.
    
    \item $\card{\irectype}{A} > 0$. Then, $A = \ctxtlra[\rectype{V}{A'}]$.
    Similarly to the previous case, by Def.~\ref{def:toBTree} we have
    $$\cut{\toBTree{\ctxtlra[\rectype{V}{A'}]}}{k} =
    \cut{\toBTree{\ctxtlra[\substitute{V}{\rectype{V}{A'}}{A'}]}}{k}$$ for
    every $k \in \Natural$. By contractiveness of $\mu$-types we can safely
    apply the inductive hypothesis to get
    $\pair{\ctxtlra[\substitute{V}{\rectype{V}{A'}}{A'}]}{B} \in \R$, and
    finally conclude with rule $\ruleEqalRecL$,
    $\pair{\ctxtlra[\rectype{V}{A'}]}{B} \in \Phieqtypeal(\R)$.
  \end{itemize}
\end{itemize}
\end{proof}


Thus we can resort to invertibility of the generating function to check for
$\eqtypeal$. Fig.~\ref{algo:equivalenceChecking} presents the algorithm. It
uses $\ttSeq e_1 \ldots e_n$ which sequentially evaluates each of its
arguments, returning the value of the first of these that does not
fail. Evaluation of $\fEqtype{\emptyset}{A}{B}$ can have one of two
outcomes: \ttFail, meaning that $A \not\eqtypeal B$, or a set $S\in
\powerset{\Type\times\Type}$ that is $\Phi$-dense with $(A,B)\in S$,
proving that $A\eqtypeal B$.

\begin{figure} 
{\small $$
\begin{array}{rcl}
\multicolumn{3}{l}{\fEqtype{S}{A}{B}\eqdef} \\
 & & \ttIf \pair{A}{B} \in S \\
 & & \quad \ttThen S \\
 & & \quad \ttElse\ \ttLet S_0 = S \cup \set{\pair{A}{B}} \ttIn \\
 & & \qquad \ttCase \pair{A}{B} \ttOf \\
 & & \qqquad \pair{a}{a} \rightarrow \\
 & & \qqqquad S_0\\

 & & \qqquad \pair{\datatype{A'}{A''}}{\datatype{B'}{B''}} \rightarrow \\
 & & \qqqquad \ttIf A', B' \text{ are datatypes} \\
 & & \qqqqquad \ttThen\ \ttLet S_1 = \fEqtype{S_0}{A'}{B'} \ttIn \\
 & & \qqqqqquad \fEqtype{S_1}{A''}{B''}\\
 & & \qqqqquad \ttElse \ttFail \\

 & & \qqquad \pair{\functype{A'}{A''}}{\functype{B'}{B''}} \rightarrow \\
 & & \qqqquad \ttLet S_1 = \fEqtype{S_0}{A'}{B'} \ttIn \\
 & & \qqqqquad \fEqtype{S_1}{A''}{B''} \\

 & & \qqquad \pair{\ctxtlra[\rectype{V}{A'}]}{B} \rightarrow \\
 & & \qqqquad \fEqtype{S_0}{\ctxtlra[\substitute{V}{\rectype{V}{A'}}{A'}]}{B} \\

 & & \qqquad \pair{A}{\ctxtlrb[\rectype{W}{B'}]} \rightarrow \\
 & & \qqqquad \fEqtype{S_0}{A}{\ctxtlrb[\substitute{W}{\rectype{W}{B'}}{B'}]} \\

 & & \qqquad \pair{\maxuniontype{i \in 1..n}{A_i}}{\maxuniontype{j \in 1..m}{B_j}} \rightarrow \\

 & & \qqqquad \ttLet S_1       = (\ttSeq \fEqtype{S_0}{A_1}{B_1},           \ldots, \fEqtype{S_0}{A_1}{B_m}) \ttIn \\
 & & \qqqquad \ldots \\
 & & \qqqquad \ttLet S_n       = (\ttSeq \fEqtype{S_{n-1}}{A_n}{B_1},       \ldots, \fEqtype{S_{n-1}}{A_n}{B_m}) \ttIn \\
 & & \qqqquad \ttLet S_{n+1}   = (\ttSeq \fEqtype{S_n}{A_1}{B_1},           \ldots, \fEqtype{S_n}{A_n}{B_1}) \ttIn \\
 & & \qqqquad \ldots \\
 & & \qqqquad \ttLet S_{n+m-1} = (\ttSeq \fEqtype{S_{n+m-2}}{A_1}{B_{m-1}}, \ldots, \fEqtype{S_{n+m-2}}{A_n}{B_{m-1}}) \ttIn \\
 & & \qqqqquad \ttSeq \fEqtype{S_{n+m-1}}{A_1}{B_m}, \ldots, \fEqtype{S_{n+m-1}}{A_n}{B_m} \\

 & & \qqquad \mathtt{otherwise} \rightarrow \\
 & & \qqqquad \ttFail
\end{array} $$
} 
\caption{Equivalence checking algorithm.}
\label{algo:equivalenceChecking}
\end{figure}

\subsection{Subtype Checking}

The approach to subtype checking is similar to that of type equivalence.
First consider the relation $\subtypeal$ over $\mu$-types defined in
Fig.~\ref{fig:subtypingSchemesAl}. It captures $\subtypemu$:

\begin{proposition}
\label{prop:subtypealSoundnessAndCompleteness}
$A \subtypeal B$ iff $A \subtypemu B$.
\end{proposition}

The proof strategy is similar to that of
Prop.~\ref{prop:eqtypealSoundnessAndCompleteness}. In this case we resort to a
proper subtyping relation for infinite trees that essentially results from
dropping rules $\ruleSubalRecL$ and $\ruleSubalRecR$ in
Fig.~\ref{fig:subtypingSchemesAl}.

\begin{figure} 
{\small $$
\begin{array}{c}
\RuleCo{}{a \subtypeal a}{\ruleSubalRefl}
\\
\\
\qquad
\RuleCo{D \subtypeal D' \quad A \subtypeal A'}
       {\datatype{D}{A} \subtypeal \datatype{D'}{A'}}
       {\ruleSubalComp}
\qquad
\RuleCo{A' \subtypeal A \quad B \subtypeal B'}
       {\functype{A}{B} \subtypeal \functype{A'}{B'}}
       {\ruleSubalFunc}
\\
\\
\RuleCo{\substitute{V}{\rectype{V}{A}}{A} \subtypeal B}
       {\rectype{V}{A} \subtypeal B}
       {\ruleSubalRecL}
\qquad
\RuleCo{A \subtypeal \substitute{W}{\rectype{W}{B}}{B}
        \quad
        A \neq \irectype
       }
       {A \subtypeal \rectype{W}{B}}
       {\ruleSubalRecR}
\\
\\
\RuleCo{A_i \subtypeal B \text{ for all $i \in 1..n$}
        \quad
        n > 1
        \quad
        B \neq \irectype
        \quad
        A_i \neq \iuniontype}
       {\maxuniontype{i \in 1..n}{A_i} \subtypeal B}
       {\ruleSubalUnionL}
\\
\\
\RuleCo{A \subtypeal B_k \text{ for some $k \in 1..m$}
        \quad
        m > 1
        \quad
        A \neq \irectype, \iuniontype
        \quad
        B_j \neq \iuniontype}
       {A \subtypeal \maxuniontype{j \in 1..m}{B_j}}
       {\ruleSubalUnionR}
\end{array} $$
} 
\caption{Coinductive axiomatization of subtyping for contractive $\mu$-types.}
\label{fig:subtypingSchemesAl}
\end{figure}

In this case we resort to Prop.~\ref{prop:subtypeSoundnessAndCompleteness} and
Lem.~\ref{lem:cutSubtypingCo} to establish a relation with $\subtypemu$ via
the infinite trees semantics and its finite truncations~\cite{DBLP:journals/entcs/VisoBA16}.

\begin{proposition}
\label{prop:subtypeSoundnessAndCompleteness}
$A \subtypemu B$ iff $\toBTree{A} \subtypeco \toBTree{B}$.
\end{proposition}

\begin{lemma}
\label{lem:cutSubtypingCo}
$\forall k \in \Natural.\cut{\tA}{k} \subtypeco \cut{\tB}{k}$ iff $\tA \subtypeco \tB$.
\end{lemma}

The following lemma allows us to relate the different non-union projections of
two maximal union types that belong to a $\Phisubtypeal$-dense relation.

\begin{lemma}
\label{lem:subtypeUnion}
Let $A, B \in \Type$. Suppose $\pair{A}{B} \in \S$ and $\S \subseteq
\Phisubtypeal(\S)$. Let $\cA$ be the maximal $\irectype\iuniontype$-context
such that $A = \cA[A_1, \ldots, A_n]$, for some $A_i$, $i \in 1..n$, and $\cB$
the maximal $\irectype\iuniontype$-context such that $B = \cB[B_1, \ldots,
B_m]$. Then, there exists $f : 1..n \to 1..m$ such that for each $i \in 1..n$,
$\pair{\project{\pi_i}{\cA[\vec{A}]}}{\project{\rho_{f(i)}}{\cB[\vec{B}]}} \in
\S$.
\end{lemma}

\begin{proof}
Induction on the sum of the $\irectype\iuniontype$-depths.
\begin{itemize}
  \item If the sum is 0, then $\cA = \Box = \cB$ and the result is immediate.
  
  \item If it is greater than 0, we analyze the shape of $\cA$:
  \begin{itemize}
    \item $\cA = \Box$. Then $n = 1$ and $A_n = A \neq \irectype,\iuniontype$
    and we have two possibilities for $\cB$:
    \begin{enumerate}
      \item $\cB = \rectype{W}{\cB'}$. Then the only rule that applies is
      $\ruleSubalRecR$. By hypothesis of the rule,
      $\pair{A}{\substitute{W}{\rectype{W}{\cB'[\vec{B}]}}{\cB'[\vec{B}]}} \in
      \S$. Moreover, since
      $$\substitute{W}{\rectype{W}{\cB'[\vec{B}]}}{\cB'[\vec{B}]} =
      \cB'[\substitute{W}{\rectype{W}{\cB'[\vec{B}]}}{\vec{B}}]$$ and
      $\mu$-types are contractive (hence none of the $\vec{B}$ can be
      variables), we can assure $\cB'$ is a maximal
      $\irectype\iuniontype$-context too. Then, by inductive hypothesis, there
      exists $f : 1..n \to 1..m$ such that
      $$\pair{\project{\pi_i}{\cA[\vec{A}]}}{\project{\rho_{f(i)}}{\substitute{W}{\rectype{W}{\cB'[\vec{B}]}}{\cB'[\vec{B}]}}}
      \in \S$$ for each $i \in 1..n$. Finally, we conclude by
      Def.~\ref{def:projection},
      $\pair{\project{\pi_i}{\cA[\vec{A}]}}{\project{1\rho_{f(i)}}{(\rectype{W}{\cB'[\vec{B}]})}}
      \in \S$.

      \item $\cB = \uniontype{\cB_1}{\cB_2}$. From $B = \cB[\vec{B}]$, with a
      similar analysis to the one made for the proof of
      Lem.~\ref{lem:decomposeUnion}, we can decompose $\cB$ into a maximal
      $\iuniontype$-context $\cB'$ and multiple $\irectype\iuniontype$-contexts
      $\cB'_1, \ldots, \cB'_{m'}$ with $m' \leq m$. Then, there exists types
      $B'_1, \ldots, B'_{m'}$ and functions $s,t : 1..m' \to 1..m$ such that
\begin{equation}
\label{eq:subtypeUnion:unionR}
B = \cB'[\vec{B'}]
\qquad\text{and}\qquad
B'_l = \cB'_l[B_{s(l)}, \ldots, B_{t(l)}]
\end{equation}
      for every $l \in 1..m'$. Moreover, since $\cB = \uniontype{\cB_1}{\cB_2}$
      and $\cB'$ is maximal, we have $m' > 1$ and $B'_l \neq \iuniontype$ for
      every $l \in 1..m'$.
      
      Then, we can assure that the only rule that applies for $\pair{A}{B} \in
      \S$ is $\ruleSubalUnionR$. By hypothesis of it, $\pair{A}{B'_k} \in \S$
      for some $k \in 1..m'$. Let $m_k = t(k)-s(k)$ (\ie the number of holes in
      $\cB'_k$). By inductive hypothesis and
      (\ref{eq:subtypeUnion:unionR})-right with $l = k$, there exists $f : 1..n
      \to 1..m_k$ such that
      $$\pair{\project{\pi_i}{\cA[\vec{A}]}}{\project{\tau^k_{f(i)}}{\cB'_k[B_{s(k)},\ldots,B_{t(k)}]}}
      \in \S$$ where $\tau^k_h$ is the position of the $(h-s(k))$-th hole in
      $\cB'_k$. Let $\rho'_k$ be the position of the $k$-th hole in $\cB'$ and
      take $\rho_{h} = \rho'_k\tau^k_{h}$. Then, by Def.~\ref{def:projection}
      $$\project{\tau_{h}}{\cB'_k[B_{s(k)},\ldots,B_{t(k)}]} =
      \project{\rho_h}{\cB'[\vec{B'}]}$$ Finally, since $\cB'[\vec{B'}] = B =
      \cB[\vec{B}]$, we resort to the hypothesis and
      (\ref{eq:subtypeUnion:unionR})-left to conclude
      $\pair{\project{\pi_i}{\cA[\vec{A}]}}{\project{\rho_{f(i)}}{\cB[\vec{B}]}}
      \in \S$.
    \end{enumerate}
    
    \item $\cA = \rectype{V}{\cA'}$. Then the only rule that applies is
    $\ruleSubalRecL$ and we have
    $$\pair{\substitute{V}{\rectype{V}{\cA'[\vec{A}]}}{\cA'[\vec{A}]}}{\cB[\vec{B}]}
    \in \S$$ As before, by contractiveness of $\mu$-types, we get
    $\pair{\cA'[\substitute{V}{\rectype{V}{\cA'[\vec{A}]}}{\vec{A}}]}{\cB[\vec{B}]}
    \in \S$ with $\cA'$ still a maximal $\irectype\iuniontype$-context. Then, by
    inductive hypothesis, there exists $f : 1..n \to 1..m$ such that
    $$\pair{\project{\pi_i}{\substitute{V}{\rectype{V}{\cA'[\vec{A}]}}{\cA'[\vec{A}]}}}{\project{\rho_{f(i)}}{\cB[\vec{B}]}}
    \in \S$$ for each $i \in 1..n$. Once again we conclude by
    Def.~\ref{def:projection},
    $\pair{\project{1\pi_i}{(\rectype{V}{\cA'[\vec{A}]})}}{\project{\rho_{f(i)}}{\cB[\vec{B}]}}
    \in \S$.
  
    \item $\cA = \uniontype{\cA_1}{\cA_2}$. Here we have two possibilities for
    $B$:
    \begin{enumerate}
      \item $B = \irectype$ (\ie $\cB = \rectype{W}{\cB'}$). Then, it can only
      be the case of $\ruleSubalRecR$ and analysis here is analogous to the one
      presented before.
      
      \item $B \neq \irectype$. This case is similar to the one for $\cB =
      \uniontype{\cB_1}{\cB_2}$ presented above. Here we have $A =
      \cA[\vec{A}]$ and we can decompose $\cA$ into a maximal
      $\iuniontype$-context $\cA'$ and multiple $\irectype\iuniontype$-contexts
      $\cA'_1, \ldots, \cA'_{n'}$ with $n' \leq n$. Then, there exists types
      $A'_1, \ldots, A'_{n'}$ and functions $s,t : 1..n' \to 1..n$ such that
\begin{equation}
\label{eq:subtypeUnion:unionL}
A = \cA'[\vec{A'}]
\qquad\text{and}\qquad
A'_l = \cA'_l[A_{s(l)}, \ldots, A_{t(l)}]
\end{equation}
      for every $l \in 1..n'$. Again we get $A'_l \neq \iuniontype$ and $n' >
      1$, thus we can only apply $\ruleSubalUnionL$ for $\pair{A}{B} \in \S$.
      By hypothesis of the rule we know $\pair{A'_l}{B} \in \S$ for every $l
      \in 1..n'$. Now we take $n_l = t(l)-s(l)$ (\ie the number of holes in
      $\cA'_l$) and we apply the inductive hypothesis on each assumption of
      the rule. Thus, there exists $f_l : 1..n_l \to 1..m$ such that
      $$\pair{\project{\tau^l_i}{\cA'_l[A_{s(l)},\ldots,A_{t(l)}]}}{\project{\rho_{f_l(i)}}{\cB[\vec{B}]}}
      \in \S$$ for every $l \in 1..n'$. Then, let $\pi'_l$ be the position of
      the $l$-th hole in $\cA'$ and $\pi_i = \pi'_l\tau^l_i$ to get
      $$\project{\tau^l_i}{\cA'_l[A_{s(l)},\ldots,A_{t(l)}]} =
      \project{\pi_i}{\cA'[\vec{A'}]}$$ Finally, we combine all functions $f_l$
      into one $f : 1..n \to 1..m$, and resort to
      (\ref{eq:subtypeUnion:unionL})-left (\ie $\cA'[\vec{A'}] = A =
      \cA[\vec{A}]$) to conclude with
      $\pair{\project{\pi_i}{\cA[\vec{A}]}}{\project{\rho_{f(i)}}{\cB[\vec{B}]}}
      \in \S$.
    \end{enumerate}
  \end{itemize}
\end{itemize}
\end{proof}

Finally, we prove:

\begin{lemma}
\label{lem:subtypealSoundnessAndCompleteness}
$A \subtypeal B$ iff $\forall k \in \Natural.\cut{\toBTree{A}}{k} \subtypeco
\cut{\toBTree{B}}{k}$.
\end{lemma}

\begin{proof}
In the following proof we denote with $\card{\irectype}{A}$ the number
$\irectype$ type constructors at the head of type $A$: $$
\card{\irectype}{A} \eqdef \left\{
\begin{array}{ll}
0                        & \quad\text{if } A \neq \irectype \\
1 + \card{\irectype}{A'} & \quad\text{if } A = \rectype{V}{A'}
\end{array}\right. $$

$\Rightarrow)$ Given $A \subtypeal B$, there exists $\S \in
\powerset{\Type\times\Type}$ such that $\S \subseteq \Phisubtypeal(\S)$ and
$\pair{A}{B} \in \S$. Define: $$\R(\S) \eqdef
\set{\pair{\cut{\toBTree{A'}}{k}}{\cut{\toBTree{B'}}{k}} \mathrel|
\pair{A'}{B'} \in \S, k \in \Natural}$$ We show that $\R(\S) \subseteq
\Phisubtypeco(\R(\S))$.

Define $\S_c \eqdef \set{\pair{A'}{B'} \in \S \mathrel| c =
\card{\irectype}{A'} + \card{\irectype}{B'}}$. For each $c \in \Natural$, we
prove: $$\pair{A'}{B'} \in \S_c \implies \forall k \in
\Natural.\pair{\cut{\toBTree{A'}}{k}}{\cut{\toBTree{B'}}{k}} \in
\Phisubtypeco(\R(\S))$$

Note that if $k = 0$, the result follows directly from $\ruleSubcoRefl$ since
$\pair{\cut{\toBTree{A'}}{k}}{\cut{\toBTree{B'}}{k}} = \pair{\circ}{\circ}
\in \Phisubtypeco(\R(\S))$ trivially. So we in fact prove $$\pair{A'}{B'} \in
\S_c \implies \forall k > 0 \in
\Natural.\pair{\cut{\toBTree{A'}}{k}}{\cut{\toBTree{B'}}{k}} \in
\Phisubtypeco(\R(\S))$$

We proceed by induction on $c$.
\begin{itemize}
  \item $c = 0$. Since $\pair{A}{B} \in \S$ and $\S$ is $\Phisubtypeal$-dense
  one of the following cases must occur:
  \begin{itemize}
    \item $\pair{A}{B} = \pair{a}{a}$. Then for any $k > 0$,
    $\pair{\cut{\toBTree{A}}{k}}{\cut{\toBTree{B}}{k}} = \pair{a}{a} \in
    \Phisubtypeco(\R(\S))$ by $\ruleSubcoRefl$.
    
    \item $\pair{A}{B} = \pair{\functype{A'}{B'}}{\functype{A''}{B''}}$ and
    $\pair{A''}{A'} \in \S$ and $\pair{B'}{B''} \in \S$. From $\pair{A''}{A'}
    \in \S$ we have $\pair{\cut{\toBTree{A''}}{k-1}}{\cut{\toBTree{A'}}{k-1}}
    \in \R(\S)$, for any $k > 0$, as follows from Def.~\ref{def:toBTree}.
    Similarly,
    $\pair{\cut{\toBTree{B'}}{k-1}}{\cut{\toBTree{B''}}{k-1}} \in \R(\S)$. Then
    by $\ruleSubcoFunc$ and Def.~\ref{def:treeCut},
    $$\pair{\cut{\toBTree{\functype{A'}{B'}}}{k}}{\cut{\toBTree{\functype{A''}{B''}}}{k}}
    \in \Phisubtypeco(\R(\S))$$ for any $k > 0$.
    
    \item $\pair{A}{B} = \pair{\datatype{D}{A'}}{\datatype{D'}{B'}}$ and
    $\pair{D}{D'} \in \S$ and $\pair{A'}{B'} \in \S$. As before, we get
    $\pair{\cut{\toBTree{D}}{k-1}}{\cut{\toBTree{D'}}{k-1}},
    \pair{\cut{\toBTree{A'}}{k-1}}{\cut{\toBTree{B'}}{k-1}}\in \R(\S)$ for any
    $k > 0$ from the definition of $T$, and conclude by applying
    $\ruleSubcoComp$ with Def.~\ref{def:treeCut}.
    
    \item $\pair{A}{B} = \pair{\maxuniontype{i \in 1..n}{A_i}}{B}$ with $n >
    1$, $B \neq \irectype$, $A_i \neq \iuniontype$ and $\pair{A_i}{B} \in \S$
    for all $i \in 1..n$. $A$ can alternately be written as $\cA[\vec{A}]$
    where $\cA$ is a maximal $\iuniontype$-context.
    
    Let $\cA'$ be the $\irectype\iuniontype$-context such that $A =
    \cA'[\vec{A'}]$ with $A'_i \neq \irectype,\iuniontype$ for all $i \in
    1..n'$ (note that $n' \geq n$). Similarly, let $\cB$ the
    $\irectype\iuniontype$-context such that $B = \cB[\vec{B}]$ with $B_j \neq
    \irectype,\iuniontype$ for all $j \in 1..m$. Note that, since $A =
    \cA[\vec{A}]$, there exist $s,t : 1..n \to 1..n'$ and
    $\irectype\iuniontype$-contexts $\cA_1,\ldots,\cA_n$ such that
    $$\cA' = \cA[\vec{\cA}] \qquad\text{and}\qquad A_i = \cA_i[A'_{s(i)},
    \ldots, A'_{t(i)}]$$

    From $\pair{A}{B}\in \S$ and Lem.~\ref{lem:subtypeUnion}, there exists $f :
    1..n' \to 1..m$ such that
    $$\pair{\project{\pi_i}{\cA'[\vec{A'}]}}{\project{\rho_{f(i)}}{\cB[\vec{B}]}}
    \in \S$$ where $\pi_i$ is the position of the $i$-th hole in $\cA'$ and
    $\rho_{f(i)}$ is the position of the $f(i)$-th hole in $\cB$. Now we reason
    as follows: $$\kern-4em
\begin{array}{rcl@{\quad}l}
\pair{\project{\pi_i}{\cA'[\vec{A'}]}}{\project{\rho_{f(i)}}{\cB[\vec{B}]}} \in \S
  & \implies & \pair{\cut{\toBTree{\project{\pi_i}{\cA'[\vec{A'}]}}}{k}}{\cut{\toBTree{\project{\rho_{f(i)}}{\cB[\vec{B}]}}}{k}} \in \R(\S) & \text{for all $i \in 1..n', k \in \Natural$} \\
  & \implies & \langle \eraserec{\cA'}[\cut{\toBTree{\project{\pi_1}{\cA'[\vec{A'}]}}}{k},\ldots,\cut{\toBTree{\project{\pi_{n'}}{\cA'[\vec{A'}]}}}{k}], \\
  &          & \hphantom{\langle} \eraserec{\cB}[\cut{\toBTree{\project{\rho_1}{\cB[\vec{B}]}}}{k},\ldots,\cut{\toBTree{\project{\rho_m}{\cB[\vec{B}]}}}{k}] \rangle \in \Phisubtypeco(\R(\S)) & \text{$\ruleSubcoUnion$ with $f$} \\
  & \implies & \langle \cut{(\eraserec{\cA'}[\toBTree{\project{\pi_1}{\cA'[\vec{A'}]}},\ldots,\toBTree{\project{\pi_{n'}}{\cA'[\vec{A'}]}}])}{k}, \\
  &          & \hphantom{\langle} \cut{(\eraserec{\cB}[\toBTree{\project{\rho_1}{\cB[\vec{B}]}},\ldots,\toBTree{\project{\rho_m}{\cB[\vec{B}]}}])}{k} \rangle \in \Phisubtypeco(\R(\S)) \\
  & \implies & \pair{\cut{\toBTree{\cA'[\vec{A'}]}}{k}}{\cut{\toBTree{\cB[\vec{B}]}}{k}} \in \Phisubtypeco(\R(\S)) & \text{Lem.~\ref{lem:forgettingMu}}
\end{array} $$
    Note that since $A'_i\neq \irectype,\iuniontype$ for all $i\in 1..n'$, then
    $\project{\pi_i}{\cA'[\vec{A'}]} \neq \irectype,\iuniontype$ too. Then, we
    can assure that $\toBTree{\project{\pi_i}{\cA'[\vec{A'}]}} \neq
    \iuniontype$ for every $i \in 1..n'$. Similarly,
    $\toBTree{\project{\rho_j}{\cB[\vec{B}]}} \neq \iuniontype$ for every $j
    \in 1..m$. Thus, $\ruleSubcoUnion$ can be applied safely.
    
    \item $\pair{A}{B} = \pair{A}{\maxuniontype{j \in 1..m}{B_l}}$ with
    $\pair{A}{B_l} \in \S$ for some $l \in 1..m$, $m > 1$, $A \neq
    \irectype,\iuniontype$ and $B_j \neq \iuniontype$ for all $j \in 1..m$. We
    call $\cB$ the maximal $\iuniontype$-context defining $B$ (\ie $B =
    \cB[\vec{B}]$ with $B_j \neq \iuniontype$ for every $j \in 1..m$).

    Let $\cB'$ be the $\irectype\iuniontype$-context such that $B =
    \cB'[\vec{B'}]$ with $B'_j \neq \irectype,\iuniontype$ for all $l \in
    1..m'$ (here $m' \geq m$). Similarly, $A = \cA[\vec{A}]$ but we know that
    $\cA = \Box$ since $A \neq \irectype,\iuniontype$ by hypothesis. As in the
    previous case, from $B = \cB[\vec{B}]$, there exists $s,t : 1..m \to 1..m'$
    and $\irectype\iuniontype$-contexts $\cB_1,\ldots,\cB_m$ such that $$\cB' =
    \cB[\vec{\cB}] \qquad\text{and}\qquad B_j = \cB_j[B'_{s(j)}, \ldots,
    B'_{t(j)}]$$

    Take $\vec{B'_j} = B'_{s(j)},\ldots,B'_{t(j)}$ for every $j \in 1..m$.
    Then, from $\pair{A}{B_l} \in \S$ and Lem.\ref{lem:subtypeUnion}, there
    exists $f : 1..n \to 1..m'$ such that
    $$\pair{\project{\pi_i}{\cA[\vec{A}]}}{\project{\tau_{f(i)}}{\cB_l[\vec{B'_l}]}}
    \in \S$$

    Now, from $B'_j \neq \irectype,\iuniontype$ we can assure
    $\project{\rho'_j}{\cB'[\vec{B'}]} \neq \irectype,\iuniontype$ for every $j
    \in 1..m'$. Thus, $\project{\rho'_{f(1)}}{\cB'[\vec{B'}]} \neq
    \irectype,\iuniontype$ taking $\rho'_{f(1)} = \rho_l\tau_{f(1)}$ where
    $\rho_l$ is the position of the $l$-th hole in $\cB$. Finally, we can
    assure both $\toBTree{A}$ and
    $\toBTree{\project{\rho'_{f(1)}}{\cB'[\vec{B'}]}}$ are non-union. We
    conclude as follows (recall $n = 1$ and $\cA = \Box$, thus
    $\project{\pi_1}{\cA[\vec{A}]} = A$): $$\kern-6em
\begin{array}{rcl@{\quad}l}
\pair{\project{\pi_1}{\cA[\vec{A}]}}{\project{\rho'_{f(1)}}{\cB'[\vec{B'}]}} \in \S
  & \implies & \pair{\cut{\toBTree{\project{\pi_1}{\cA[\vec{A}]}}}{k}}{\cut{\toBTree{\project{\rho'_{f(1)}}{\cB'[\vec{B'}]}}}{k}} \in \R(\S) & \text{for all $k \in \Natural$} \\
  & \implies & \langle \cut{\toBTree{\project{\pi_1}{\cA[\vec{A}]}}}{k}, \\
  &          & \hphantom{\langle} \eraserec{\cB'}[\cut{\toBTree{\project{\rho'_1}{\cB'[\vec{B'}]}}}{k},\ldots,\cut{\toBTree{\project{\rho'_m}{\cB'[\vec{B'}]}}}{k}] \rangle \in \Phisubtypeco(\R(\S)) & \text{$\ruleSubcoUnion$ with $f$} \\
  & \implies & \langle \cut{\toBTree{\project{\pi_1}{\cA[\vec{A}]}}}{k}, \\
  &          & \hphantom{\langle} \cut{(\eraserec{\cB'}[\toBTree{\project{\rho'_1}{\cB'[\vec{B'}]}},\ldots,\toBTree{\project{\rho'_m}{\cB'[\vec{B'}]}}])}{k} \rangle \in \Phisubtypeco(\R(\S)) \\
  & \implies & \pair{\cut{\toBTree{\cA[\vec{A}]}}{k}}{\cut{\toBTree{\cB'[\vec{B'}]}}{k}} \in \Phisubtypeco(\R(\S)) & \text{Lem.~\ref{lem:forgettingMu}}
\end{array} $$
  \end{itemize}

  \item $c > 0$. Then:
  \begin{itemize}
    \item $\pair{A}{B} = \pair{\rectype{V}{A'}}{B}$ and
    $\pair{\substitute{V}{\rectype{V}{A'}}{A'}}{B} \in \S$. By inductive
    hypothesis,
    $$\pair{\cut{\toBTree{\substitute{V}{\rectype{V}{A'}}{A'}}}{k}}{\cut{\toBTree{B}}{k}}
    \in \Phisubtypeco(\R(\S))$$ Since
    $\pair{\cut{\toBTree{\substitute{V}{\rectype{V}{A'}}{A'}}}{k}}{\cut{\toBTree{B}}{k}}
    = \pair{\cut{\toBTree{\rectype{V}{A'}}}{k}}{\cut{\toBTree{B}}{k}}$ by
    Def.~\ref{def:toBTree}, we conclude.

    \item $\pair{A}{B} = \pair{A}{\rectype{W}{B'}}$ and
    $\pair{A}{\substitute{W}{\rectype{W}{B'}}{B'}} \in \S$ and $A \neq
    \irectype$. As before, from the inductive hypothesis we get
    $$\pair{\cut{\toBTree{A}}{k}}{\cut{\toBTree{\substitute{W}{\rectype{W}{B'}}{B'}}}{k}}
    \in \Phisubtypeco(\R(\S))$$ and we conclude by Def.~\ref{def:toBTree}.
  \end{itemize}
\end{itemize}

$\Leftarrow)$ We show that $\R \eqdef \set{\pair{A'}{B'} \mathrel| A', B' \in
\Type, \forall k \in \Natural. \cut{\toBTree{A'}}{k} \subtypeco
\cut{\toBTree{B'}}{k}}$ is $\Phisubtypeal$-dense.

Let $\pair{A}{B} \in \R$. We proceed by induction on $c = \card{\irectype}{A} +
\card{\irectype}{B}$. From now on we consider $k > 0$ to avoid the border case
$\pair{\cut{\toBTree{A}}{0}}{\cut{\toBTree{B}}{0}} = \pair{\circ}{\circ}$.

\begin{itemize}
  \item $c = 0$. Since $\pair{\cut{\toBTree{A}}{k}}{\cut{\toBTree{B}}{k}} \in
  \mathbin{\subtypeco}$ and $\mathbin{\subtypeco} = \Phisubtypeco(\subtypeco)$,
  one of the following cases must hold:
  \begin{itemize}
    \item $\pair{\cut{\toBTree{A}}{k}}{\cut{\toBTree{B}}{k}} = \pair{a}{a}$.
    Then, since $c = 0$ it must be the case $A = a = B$. Thus, by
    $\ruleSubalRefl$, $\pair{A}{B} \in \Phisubtypeal(\R)$.
    
    \item $\pair{\cut{\toBTree{A}}{k}}{\cut{\toBTree{B}}{k}} =
    \pair{\functype{\cut{\tA'}{k-1}}{\cut{\tA''}{k-1}}}
         {\functype{\cut{\tB'}{k-1}}{\cut{\tB''}{k-1}}}$
    with
\begin{equation}
\label{eq:subCompl:func}
\cut{\tB'}{k-1} \subtypeco \cut{\tA'}{k-1}
\qquad\text{and}\qquad
\cut{\tA''}{k-1} \subtypeco \cut{\tB''}{k-1}
\end{equation}
    Once again, since the outermost constructor of both $A$ and $B$ is not
    $\irectype$, there exists $A',A'',B',B'' \in \Type$ such that $A =
    \functype{A'}{A''}$ and $B = \functype{B'}{B''}$ with $\toBTree{A'} = \tA'$,
    $\toBTree{A''} = \tA''$, $\toBTree{B'} = \tB'$ and $\toBTree{B''} = \tB''$.
    Moreover, from (\ref{eq:subCompl:func}) we get $\pair{B'}{A'},
    \pair{A''}{B''} \in \R$. We conclude with $\ruleSubalFunc$, $\pair{A}{B}
    \in \Phisubtypeal(\R)$.
    
    \item $\pair{\cut{\toBTree{A}}{k}}{\cut{\toBTree{B}}{k}} =
    \pair{\datatype{\cut{\tA'}{k-1}}{\cut{\tA''}{k-1}}}
         {\datatype{\cut{\tB'}{k-1}}{\cut{\tB''}{k-1}}}$
    with $$\cut{\tA'}{k-1} \subtypeco \cut{\tB'}{k-1} \qquad\text{and}\qquad
    \cut{\tA''}{k-1} \subtypeco \cut{\tB''}{k-1}$$ As in the previous case, we
    know there exists $D,D',A',B' \in \Type$ such that $A = \datatype{D}{A'}$,
    $B = \datatype{D'}{B'}$ and $\pair{D}{D'}, \pair{A'}{B'} \in \R$. Here we
    use $\ruleSubalComp$ to conclude.
    
    \item $\pair{\cut{\toBTree{A}}{k}}{\cut{\toBTree{B}}{k}} =
    \pair{\maxuniontype{i \in 1..n}{\cut{\tA_i}{k}}}
         {\maxuniontype{j \in 1..m}{\cut{\tB_j}{k}}}$
    such that $n+m > 2$, $\cut{\tA_i}{k}, \cut{\tB_j}{k} \neq \iuniontype$ and
    there exists $f : 1..n \to 1..m$ s.t. $\cut{\tA_i}{k} \subtypeco
    \cut{\tB_{f(i)}}{k}$.
    
    We distinguish two cases:
    \begin{enumerate}
      \item $n = 1$. Then $\cut{\toBTree{A}}{k} = \cut{\tA_1}{k} \neq
      \irectype,\iuniontype$ and $m > 1$. Moreover, by
      Lem.~\ref{lem:decomposeUnion} on $B$, there exists $m' \leq m$, $\cB,
      \cB_1, \ldots, \cB_{m'}$ maximal $\iuniontype$-contexts, $B_1, \ldots,
      B_{m'} \neq \iuniontype$ and functions $s,t : 1..m' \to 1..m$ such that
\begin{equation}
\label{eq:subCompl:unionR}
B = \cB[\vec{B}]
\qquad\text{and}\qquad
\toBTree{B_l} = \cB_l[\tB_{s(l)}, \ldots, \tB_{t(l)}]
\end{equation}
      for every $l \in 1..m'$.
      
      Let $h \in 1..m'$ such that $f(1) \in s(h)..t(h)$. From
      $$\cut{\toBTree{A}}{k} = \cut{\tA_1}{k} \subtypeco \cut{\tB_{f(1)}}{k}$$
      and (\ref{eq:subCompl:unionR})-right, we get $$\cut{\toBTree{A}}{k}
      \subtypeco \maxuniontype{j \in s(h)..t(h)}{\cut{\tB_j}{k}} =
      \cut{\toBTree{B_h}}{k}$$ by $\ruleSubcoUnion$. Thus, $\pair{A}{B_h} \in
      \R$.
      
      Finally, we have $A \neq \irectype,\iuniontype$, $B_l \neq \iuniontype$
      and, since $\card{\irectype}{B} = 0$, $m \geq m' > 1$. We take $f' : 1..1
      \to 1..m'$ such that $f'(1) = h$ and conclude with
      (\ref{eq:subCompl:unionR})-left, by $\ruleSubalUnionR$, $\pair{A}{B} \in
      \Phisubtypeal(\R)$.
      
      \item $n > 1$. By Lem.~\ref{lem:decomposeUnion} on $A$, there exists $n'
      \leq n$, $\cA, \cA_1, \ldots, \cA_{n'}$ maximal $\iuniontype$-contexts,
      $A_1, \ldots, A_{n'} \neq \iuniontype$ and functions $s,t : 1..n' \to
      1..n$ such that
\begin{equation}
\label{eq:subCompl:unionL}
A = \cA[\vec{A}]
\qquad\text{and}\qquad
\toBTree{A_l} = \cA_l[\tA_{s(l)}, \ldots, \tA_{t(l)}]
\end{equation}
      for every $l \in 1..n'$.
      
      From (\ref{eq:subCompl:unionL})-right, by $\ruleSubcoUnion$ using
      function $f$ from the hypothesis, we get $$\cut{\toBTree{A_l}}{k} =
      \maxuniontype{i \in s(h)..t(h)}{\cut{\tA_i}{k}} \subtypeco
      \cut{\toBTree{B}}{k}$$ and thus, $\pair{A_l}{B} \in \R$ for every $l \in
      1..n'$.
      
      Then, we conclude by $\ruleSubalUnionL$ resorting to
      (\ref{eq:subCompl:unionL})-left and the fact that $A,B \neq \irectype$ by
      hypothesis (hence $n \geq n' > 1$) and $A_l \neq \iuniontype$,
      $\pair{A}{B} \in \Phisubtypeal(\R)$.
    \end{enumerate}
  \end{itemize}

  \item $c > 0$. Then we must analyze two cases:
  \begin{enumerate}
    \item $A = \rectype{V}{A'}$. Then $\cut{\toBTree{A}}{k} =
    \cut{\toBTree{\substitute{V}{\rectype{V}{A'}}{A'}}}{k}$ by
    Def.~\ref{def:toBTree}. Thus,
    $\pair{\substitute{V}{\rectype{V}{A'}}{A'}}{B} \in \R$ too. We conclude by
    $\ruleSubalRecL$, $\pair{A}{B} \in \Phisubtypeal(\R)$.
    
    \item $A \neq \irectype$ and $B = \rectype{W}{B'}$. Again, $\pair{A}{B} \in
    \R$ implies $\pair{A}{\substitute{W}{\rectype{W}{B'}}{B'}} \in \R$. This
    time we conclude by $\ruleSubalRecR$, $\pair{A}{B} \in \Phisubtypeal(\R)$.
  \end{enumerate}
\end{itemize}
\end{proof}


Unfortunately, the generating function determined by the rules in
Fig.~\ref{fig:subtypingSchemesAl}, let us call it $\Phisubtypeal$, is not
invertible. Notice that $\ruleSubalUnionR$ overlaps with itself. For example,
$\consttype{c} \subtypeal
\uniontype{(\uniontype{\consttype{c}}{\consttype{d}})}{(\uniontype{\consttype{e}}{\consttype{c}})}$
belongs to two $\Phisubtypeal$-saturated sets (\ie sets $\X$ such that $\X \subseteq
\Phisubtypeal(\X)$):
\begin{center}
$\begin{array}{rcl}
\X_1  & =  &\set{
\pair{\consttype{c}}{(\consttype{c}\iuniontype\consttype{d})\iuniontype
  (\consttype{e}\iuniontype\consttype{c})}, 
\pair{\consttype{c}}{
(\consttype{c}\iuniontype\consttype{d})}, 
\pair{\consttype{c}}{\consttype{c}}}\\
\X_2  & =  & \set{
\pair{\consttype{c}}{(\consttype{c}\iuniontype\consttype{d})\iuniontype
  (\consttype{e}\iuniontype\consttype{c})}, 
\pair{\consttype{c}}{
(\consttype{e}\iuniontype\consttype{c})}, 
\pair{\consttype{c}}{\consttype{c}}}
\end{array}$
\end{center}

However, since this is the only source of non-invertibility we easily derive
a subtype membership checking function $\fSubtype{\bullet}{\bullet}{\bullet}$
that, in the case of $\ruleSubalUnionR$, simply checks all cases
(Fig.~\ref{algo:subtypeChecking}).

\begin{figure} $$
\begin{array}{rcl}
\multicolumn{3}{l}{\fSubtype{S}{A}{B}\eqdef} \\
 & & \ttIf \pair{A}{B} \in S \\
 & & \quad \ttThen S \\
 & & \quad \ttElse\ \ttLet S_0 = S \cup \set{\pair{A}{B}} \ttIn \\
 & & \qquad \ttCase \pair{A}{B} \ttOf \\
 & & \qqquad \pair{a}{a} \rightarrow \\
 & & \qqqquad S_0\\

 & & \qqquad \pair{\datatype{A'}{A''}}{\datatype{B'}{B''}} \rightarrow \\
 & & \qqqquad \ttIf A', B' \text{ are datatypes} \\
 & & \qqqqquad \ttThen\ \ttLet S_1 = \fSubtype{S_0}{A'}{A''} \ttIn \\
 & & \qqqqqquad \fSubtype{S_1}{B'}{B''}\\
 & & \qqqqquad \ttElse \ttFail \\

 & & \qqquad \pair{\functype{A'}{A''}}{\functype{B'}{B''}} \rightarrow \\
 & & \qqqquad \ttLet S_1 = \fSubtype{S_0}{B'}{A'} \ttIn \\
 & & \qqqqquad \fSubtype{S_1}{A''}{B''} \\

 & & \qqquad \pair{\rectype{V}{A'}}{B} \rightarrow \\
 & & \qqqquad \fSubtype{S_0}{\substitute{V}{\rectype{V}{A'}}{A'}}{B} \\

 & & \qqquad \pair{A}{\rectype{W}{B'}} \rightarrow \\
 & & \qqqquad \fSubtype{S_0}{A}{\substitute{W}{\rectype{W}{B'}}{B'}} \\

 & & \qqquad \pair{\maxuniontype{i \in 1..n}{A_i}}{B} \rightarrow \\
 & & \qqqquad \ttLet S_1 = \fSubtype{S_0}{A_1}{B} \ttIn \\
 & & \qqqquad \ttLet S_2 = \fSubtype{S_1}{A_2}{B} \ttIn \\
 & & \qqqquad \ldots \\
 & & \qqqquad \ttLet S_{n-1} = \fSubtype{S_{n-2}}{A_{n-1}}{B} \ttIn \\
 & & \qqqqquad \fSubtype{S_{n-1}}{A_n}{B} \\

 & & \qqquad \pair{A}{\maxuniontype{j \in 1..m}{B_j}} \rightarrow \\
 & & \qqqquad \ttSeq \fSubtype{S_0}{A}{B_1}, \ldots, \fSubtype{S_0}{A}{B_m} \\

 & & \qqquad \mathtt{otherwise} \rightarrow \\
 & & \qqqquad \ttFail
\end{array} $$
\caption{Subtype checking algorithm.}
\label{algo:subtypeChecking}
\end{figure}


%% file: typeChecking.tex
\section{Type Checking}
\label{sec:typeChecking}

A syntax-directed presentation for typing in $\capp$, inferring judgments of
the form \sequSD{\Gamma}{s : A}, may be obtained from the rules of
Fig.~\ref{fig:typingSchemesForPatternsAndTerms} by dropping subsumption. This
requires ``hard-wiring'' it back in into $\ruleTApp$. Unfortunately, the na\"ive
syntax-directed variant:
\begin{center}
$\Rule{\sequSD{\Gamma}{r : (\maxuniontype{i \in 1..n}{A_i}) \ifunctype B}
      \quad
      \sequSD{\Gamma}{u : A'}
      \quad
      A'\subtypemu A_k, \mbox{ for some }k \in 1..n
      }
      {\sequSD{\Gamma}{r \iappterm u : B}}
      {\ruleTalApp'}
$ 
\end{center} 
fails to capture all the required terms. In other words, there are $\Gamma, s$
and $A$ such that \sequT{\Gamma}{s : A} but no $A'\subtypemu A$ such that
\sequSD{\Gamma}{s : A'}. For example, take $\Gamma(x) \eqdef
\uniontype{(\functype{\uniontype{\consttype{c}}{\consttype{e}}}{\consttype{d}})}{(\functype{\uniontype{\consttype{c}}{\consttype{f}}}{\consttype{d}})}$,
$s \eqdef \appterm{x}{\constterm{c}}$ and $A \eqdef \consttype{d}$. More
generally, from \sequSD{\Gamma}{r : A} and $A \subtypemu \maxuniontype{i \in
1..n}{A_i} \ifunctype B$ we cannot infer that $A$ is a functional type due to
the presence of union types. A complete
(Prop.~\ref{prop:completenessTypeChecking}) syntax directed presentation is
obtained by dropping $\ruleTSubs$ from
Fig.~\ref{fig:typingSchemesForPatternsAndTerms} and replacing $\ruleTAbs$ and
$\ruleTApp$ by $\ruleTalAbs$ and $\ruleTalApp$, resp., of
Fig.~\ref{fig:algorithmicTypingSchemes}.

\begin{figure} $$
\begin{array}{c}
\Rule{\Gamma(x) = A}
     {\sequSD{\Gamma}{x : A}}
     {\ruleTalVar}
\quad
\Rule{}
     {\sequSD{\Gamma}{\constterm{c} : \consttype{c}}}
     {\ruleTalConst} 
\quad
\Rule{\sequSD{\Gamma}{r : D} \quad \sequSD{\Gamma}{u : A}}
     {\sequSD{\Gamma}{\dataterm{r}{u} : \datatype{D}{A}}}
     {\ruleTalComp}
\\
\\
\Rule{\begin{array}{c}
        \lista{\sequC{\theta_i}{p_i : A_i}}_{i \in 1..n}
        \text{ compatible}
        \\
        (\sequP{\theta_i}{p_i : A_i})_{i\in 1..n}
        \quad
        (\dom{\theta_i} = \fm{p_i})_{i\in 1..n}
        \quad
        (\sequSD{\Gamma, \theta_i}{s_i : B_i})_{i\in 1..n}
      \end{array}
     }
     {\sequSD{\Gamma}{(\absterm{p_i}{s_i}{\theta_i})_{i \in 1..n} :
      \functype{\maxuniontype{i \in 1..n}{A_i}}{\maxuniontype{i \in 1..n}{B_i}}}}
     {\ruleTalAbs}
\\
\\
\Rule{\begin{array}{l}
        \sequSD{\Gamma}{r : A}
        \quad
        A \eqtypemu \maxuniontype{i\in 1..n}{(\functype{A_i}{B_i})}
        \quad
        A_i \neq \iuniontype
        \\
        \sequSD{\Gamma}{u : C}
        \quad
        (\sequTE{}{C \subtypemu A_i})_{i \in 1..n}
      \end{array}
     }
     {\sequSD{\Gamma}{\appterm{r}{u} : \maxuniontype{i \in 1..n}{B_i}}}
     {\ruleTalApp} 
\end{array} $$
\caption{Syntax-directed typing rules for terms.}
\label{fig:algorithmicTypingSchemes}
\end{figure}

\begin{proposition}
\label{prop:completenessTypeChecking}
\begin{enumerate}
  \item If $\sequSD{\Gamma}{s : A}$, then $\sequT{\Gamma}{s : A}$.
  
  \item If $\sequT{\Gamma}{s : A}$, then $\exists A'$ such that $A' \subtypemu
  A$ and $\sequSD{\Gamma}{s : A'}$.
\end{enumerate}
\end{proposition}

\begin{proof}
\begin{enumerate}
  \item By induction on $\sequSD{\Gamma}{s : A}$ analyzing the last rule
  applied.
  \begin{itemize}
    \item $\ruleTalVar$, $\ruleTalConst$ and $\ruleTalComp$ are immediate.
    
    \item $\ruleTalAbs$: then $A = \functype{\maxuniontype{i \in
    1..n}{A_i}}{\maxuniontype{i \in 1..n}{B_i}}$ and, by hypothesis of the
    rule, we have $\lista{\sequC{\theta_i}{p_i : A_i}}_{i \in 1..n}$
    compatible, $\sequP{\theta_i}{p_i : A_i}$, $\dom{\theta_i} = \fm{p_i}$ and
    $\sequSD{\Gamma, \theta_i}{s_i : B_i}$ for every $i \in 1..n$. Then, by
    inductive hypothesis, $\sequT{\Gamma, \theta_i}{s_i : B_i}$ and, by
    $\ruleTSubs$, $\sequT{\Gamma, \theta_i}{s_i : \maxuniontype{i \in
    1..n}{B_i}}$ for each $i \in 1..n$. Finally we conclude with $\ruleTAbs$.
    
    \item $\ruleTalApp$: then $A = \maxuniontype{i\in 1..n}{B_i}$ and we have
    $\sequSD{\Gamma}{r : A'}$ and $\sequSD{\Gamma}{u : C}$ with $A' \eqtypemu
    \maxuniontype{i \in 1..n}{(\functype{A_i}{B_i})}$, $A_i \neq \iuniontype$
    and $C \subtypemu A_i$ for every $i \in i..n$. From $\sequSD{\Gamma}{r :
    A'}$, by inductive hypothesis, we have $\sequT{\Gamma}{r : A'}$. Moreover,
    from $A' \eqtypemu \maxuniontype{i \in 1..n}{(\functype{A_i}{B_i})}$ with
    $\ruleSubmuEq$ and $\ruleTSubs$, we get $\sequT{\Gamma}{r : \maxuniontype{i
    \in 1..n}{(\functype{A_i}{B_i})}}$. By $\ruleTSubs$ once again, we deduce
    \begin{equation}
      \sequT{\Gamma}{r : \maxuniontype{i \in
      1..n}{(\functype{A_i}{\maxuniontype{j \in 1..n}{B_j}})}}
      \label{al-app:i}
    \end{equation}
    Now, from $C \subtypemu A_i$ we get $\functype{A_i}{\maxuniontype{j \in
    1..n}{B_j}} \subtypemu \functype{C}{\maxuniontype{j \in 1..n}{B_j}}$ for
    every $i \in 1..n$, and therefore $$\maxuniontype{i \in
    1..n}{(\functype{A_i}{\maxuniontype{j \in 1..n}{B_j}})} \subtypemu
    \functype{C}{\maxuniontype{j \in 1..n}{B_j}}$$ From this and
    (\ref{al-app:i}) we deduce $\sequT{\Gamma}{r : \functype{C}{\maxuniontype{i
    \in 1..n}{B_i}}}$ using $\ruleTSubs$. Finally, from $\sequSD{\Gamma}{u :
    C}$ and the inductive hypothesis we have $\sequT{\Gamma}{u : C}$, thus we
    conclude by applying $\ruleTApp$, $\sequT{\Gamma}{\appterm{r}{u} :
    \maxuniontype{i \in 1..n}{B_i}}$.
  \end{itemize}
    
  \item By induction on $\sequT{\Gamma}{s : A}$ analyzing the last rule applied.
  \begin{itemize}
    \item $\ruleTVar$, $\ruleTConst$ and $\ruleTComp$ are immediate.
    
    \item $\ruleTAbs$: then $A = \functype{\maxuniontype{i \in 1..n}{A_i}}{B}$
    and, by hypothesis of the rule, we have $\lista{\sequC{\theta_i}{p_i :
    A_i}}_{i \in 1..n}$ compatible, $\sequP{\theta_i}{p_i : A_i}$,
    $\dom{\theta_i} = \fm{p_i}$, $\sequT{\Gamma, \theta_i}{s_i : B}$ for every
    $i \in 1..n$. Then, by inductive hypothesis, $\exists B_i$ such that $B_i
    \subtypemu B$ and $\sequSD{\Gamma, \theta_i}{s_i : B_i}$ for each $i \in
    1..n$. Finally we conclude by applying $\ruleTalAbs$, assigning $s$ the
    type $\functype{\maxuniontype{i \in 1..n}{A_i}}{\maxuniontype{i \in
    1..n}{B_i}}$. Note that $\maxuniontype{i \in 1..n}{B_i}\subtypemu B$ and
    thus we have $\functype{\maxuniontype{i \in 1..n}{A_i}}{\maxuniontype{i \in
    1..n}{B_i}} \subtypemu \functype{\maxuniontype{i \in 1..n}{A_i}}{B} = A$.
    
    \item $\ruleTApp$: then $\sequT{\Gamma}{r : \functype{\maxuniontype{i \in
    1..n}{A_i}}{A}}$ and $\sequT{\Gamma}{u : A_k}$ for some $k \in 1..n$. By
    inductive hypothesis we have $\sequSD{\Gamma}{r : B}$ with $B \subtypemu
    \functype{\maxuniontype{i \in 1..n}{A_i}}{A}$. Without loss of generality
    we can assume $B = \maxuniontype{j \in 1..m}{B_j}$ with $B_j \neq
    \iuniontype$. Moreover, by contractiveness, we can further assume that $B
    \eqtypemu \maxuniontype{j \in 1..m}{B_j}$ and $B_j \neq
    \irectype,\iuniontype$ for every $j \in 1..m$.
    
    By Prop.~\ref{prop:subtypealSoundnessAndCompleteness}, from
    $\ruleSubalUnionL$ we get $B_j \subtypemu \functype{\maxuniontype{i \in
    1..n}{A_i}}{A}$ for every $j \in 1..m$. Moreover, since $B_j \neq
    \irectype,\iuniontype$ it can only be the case of $\ruleSubalFunc$ thus, by
    Prop.~\ref{prop:subtypealSoundnessAndCompleteness} once again, we have $B_j
    = \functype{B'_j}{B''_j}$ with
    \begin{equation}
      \maxuniontype{i \in 1..n}{A_i} \subtypemu B'_j \qquad\text{and}\qquad
      B''_j \subtypemu A
      \label{eq:i}
    \end{equation}
    for every $j \in 1..m$. Once again, by contractiveness we can assume $B'_j
    \neq \irectype$. Resorting to the maximal union type notation once again,
    we can assume w.l.o.g that $\maxuniontype{i \in 1..n}{A_i} =
    \maxuniontype{i \in 1..n'}{A'_i}$ such that $A_k = \maxuniontype{i \in
    I_k}{A'_i}$ with $I_k \subseteq 1..n'$ and $A'_i \neq \iuniontype$ for
    every $i \in 1..n'$. Then, by
    Prop.~\ref{prop:subtypealSoundnessAndCompleteness} and $\ruleSubalUnionL$,
    from (\ref{eq:i}-left) we get $A'_i \subtypemu B'_j$ for every $i \in
    1..n'$. Hence, from $I_k \subseteq 1..n'$ we deduce
    \begin{equation}
      A_k \subtypemu B'_j
      \label{eq:ii}
    \end{equation}
    for every $j \in 1..m$. By inductive hypothesis from $\sequT{\Gamma}{u :
    A_k}$ we get $\exists C$ such that $C \subtypemu A_k$ and
    $\sequSD{\Gamma}{u : C}$. Then, from (\ref{eq:ii}) we deduce $C \subtypemu
    B'_j$ for every $j \in 1..m$ by transitivity of subtyping. Finally, from
    (\ref{eq:i}-right) we get $\maxuniontype{j \in 1..m}{B''_j} \subtypemu A$
    and we conclude with $\ruleTalApp$ since $B \eqtypemu \maxuniontype{j \in
    1..m}{B_j} = \maxuniontype{j \in 1..m}{(\functype{B'_j}{B''_j})}$.
    
    \item $\ruleTSubs$: here we have $\sequT{\Gamma}{s : B}$ with $B \subtypemu
    A$. By inductive hypothesis we get $\sequSD{\Gamma}{s : A'}$ with $A'
    \subtype B$ and we conclude by transitivity of subtyping.
  \end{itemize}
\end{enumerate}
\end{proof}

From this we may obtain a simple type-checking function
$\fTypeCheck{\Gamma}{s}$ (Fig.~\ref{fig:typeChecking}-top) such that
$\fTypeCheck{\Gamma}{s} = A$ iff $\sequSD{\Gamma}{s : A'}$, for some $A' \eqtypemu
A$. The interesting clause is that of application, 
where the decision of whether $\ruleTalComp$ or $\ruleTalApp$ may be applied
depends on the result of the recursive call. If the term $r$ is
assigned a datatype, then a new compound datatype is built; if its type
can be rewritten as a union of functional types, then a proper type is
constructed with each of the codomains of the latter, as established in rule
$\ruleTalApp$.

\begin{figure} $$
\begin{array}{rcl}
\fTypeCheck{\Gamma}{x} & \eqdef & \Gamma(x) \\
 
\fTypeCheck{\Gamma}{\constterm{c}} & \eqdef & \consttype{c} \\
 
\fTypeCheck{\Gamma}{(p_i \ifuncterm_{\theta_i} s_i)_{i \in 1..n}} & \eqdef & \ttLet A_i = \fTypeCheckP{\theta_i}{p_i}, B_i = \fTypeCheck{\Gamma,\theta_i}{s_i} \ttIn \\
& & \quad \ttIf \forall i \in 1..n.\forall j \in i+1..n.\fComp{{p_i : A_i}}{p_j : A_j} \\
& & \qquad \ttThen \functype{\maxuniontype{i\in 1..n}{A_i}}{\maxuniontype{i\in 1..n}{B_i}} \\
& & \qquad \ttElse \ttFail \\

\fTypeCheck{\Gamma}{\appterm{r}{u}} & \eqdef & \ttLet A = \fTypeCheck{\Gamma}{r}, C = \fTypeCheck{\Gamma}{u} \ttIn \\
& & \quad\ttIf A\mbox{ is a datatype} \\
& & \qquad \ttThen \datatype{A}{C} \\
& & \qquad \ttElse\ \ttLet \maxuniontype{i\in 1..n}{(\functype{A_i}{B_i})} = \fUnfold{A} \ttIn \\
& & \qqquad \ttIf \forall i \in 1..n.\fSubtype{\varnothing}{C}{A_i} \\
& & \qqquad \ttThen \maxuniontype{i\in 1..n}{B_i}\\
& & \qqquad \ttElse \ttFail\\
& & \\

\fTypeCheckP{\Gamma}{x} & \eqdef & \Gamma(x) \\

\fTypeCheckP{\Gamma}{\constterm{c}} & \eqdef & \consttype{c} \\

\fTypeCheckP{\Gamma}{\dataterm{p}{q}} & \eqdef & \ttLet A = \fTypeCheckP{\Gamma}{p}, B = \fTypeCheckP{\Gamma}{q} \ttIn \\
& & \quad\ttIf A\mbox{ is a datatype} \\
& & \qquad \ttThen \datatype{A}{B} \\
& & \qquad \ttElse\ \ttFail
\end{array} $$
\caption{Type-checking $\capp$.}
\label{fig:typeChecking}
\end{figure}

Compatibility between branches is verified by checking if $\Psicomp{{p : A}}{{q
: B}}$ holds:
\begin{center}
$\begin{array}{rcl}
\fComp{{p : A}}{{q : B}} & \eqdef & (\ttNot \fPComp{{p : A}}{{q : B}}) \ttOr \fSubtype{\varnothing}{B}{A}
\end{array}$
\end{center} 
In $\fPCompName$ we may assume that it has already been checked that $p$ has
type $A$ and $q$ has type $B$. Therefore, if these are compound patterns they
can only be assigned application types, and union types may only appear at leaf
positions of a pattern. We use this correspondence to traverse both pattern and
type simultaneously in linear time, which means the worst-case execution time
of the compatibility check is governed by the complexity of subtyping. $$
\begin{array}{rcl}
\fPComp{{p : A}}{{q : B}} & \eqdef & \ttIf p = \dataterm{p_1}{p_2} \ttAnd q = \dataterm{q_1}{q_2} \ttThen \\
& & \quad \ttLet A = \datatype{A_1}{A_2}, B = \datatype{B_1}{B_2} \ttIn \\
& & \qquad \fPComp{{p_1 : A_1}}{{q_1 : B_1}} \ttAnd \fPComp{{p_2 : A_2}}{{q_2 : B_2}} \\
& & \ttElse \\
& & \qquad (p = \matchable{x}) \ttOr (p = q = \constterm{c}) \ttOr (\lookup{A}{\epsilon} \cap \lookup{B}{\epsilon} \neq \varnothing)
\end{array} $$

The expression $\fUnfold{A}$, in the clause defining
$\fTypeCheck{\Gamma}{\appterm{r}{u}}$, is the result of unfolding type $A$
using rules $\ruleEqalRecL$ and $\ruleEqalRecR$ until the result is an
equivalent type $A' = \maxuniontype{i \in 1..n}{A'_i}$ with $A'_i \neq
\irectype,\iuniontype$, and then simply verifying that $A'_i =
\mathbin{\ifunctype}$ for all $i \in 1..n$. $$
\begin{array}{rcl}
\fUnfold{A} & \eqdef & \ttIf A = \functype{A'}{A''} \ttThen A \\
& & \ttElse\ \ttIf A = \maxuniontype{i \in 1..n}{A_i} \ttAnd n>1 \ttAnd A_i \neq \iuniontype \ttThen \\
& & \quad \ttLet \maxuniontype{j\in 1..m_i}{(\functype{A_{ij}}{B_{ij}})} = \fUnfold{A_i} \ttForeach i \in 1..n \ttIn \\
& & \qquad \maxuniontype{i\in 1..n\\ j\in 1..m_i}{(\functype{A_{ij}}{B_{ij}})} \\

& & \ttElse\ \ttIf A = \rectype{V}{A'} \ttThen \fUnfold{\substitute{V}{\rectype{V}{A}}{A}} \\
& & \ttElse \ttFail
\end{array} $$
Termination is guaranteed by contractiveness of $\mu$-types. In the worst case
it requires exponential time due to the need to unfold types until the desired
equivalent form is obtained (\eg 
$\rectype{X_1}{\ldots\rectype{X_n}{\functype{X_1}{\ldots\functype{X_n}{\consttype{c}}}}}$).


%% file: efficientTypeChecking.tex
\section{Towards Efficient Type-Checking}
\label{sec:efficientTypeChecking}

The algorithms presented so far are clear but inefficient. The number of
recursive calls in $\fEqtypeName$ and $\fSubtypeName$ is not bounded (it
depends on the size of the type) and unfolding recursive types may increment
their size exponentially. This section draws from ideas
in~\cite{Jim97typeinference,DBLP:journals/iandc/PalsbergZ01,DBLP:conf/tlca/CosmoPR05}
and adopts a dag-representation of recursive types which are encoded as
\emph{term automata} (described below). Associativity is handled by resorting to $n$-ary
unions, commutativity and idempotence of $\iuniontype$ is handled by how types
are decomposed in their automaton representation (\cf $\fCheckName$ in
Fig.~\ref{fig:pseudocode:subtyping2:phase2}). The algorithm itself is
$\fTypeCheckName$ of Fig.~\ref{fig:typeChecking} except that:
\begin{enumerate}
  \item The representation of $\mu$-types are now term automata. This renders
  $\fUnfoldName$ linear.
\item The subtyping algorithm is optimized, based on the new representation
  and following ideas
  from~\cite{DBLP:journals/iandc/PalsbergZ01,DBLP:conf/tlca/CosmoPR05}.
\end{enumerate}
The end product is an algorithm with complexity $\O(n^7 d)$ where $n$ is the
size of the input (\ie that of $\Gamma$ plus $t$) and $d$ is the maximum arity
of the $n$-unions occurring in $\Gamma$ and $t$. Note that all the information
needed to type $t$ is either in the context or annotated within the term
itself. Thus, a linear relation can be established between the size of the
input and the size of the resulting type; and we can think of $n$ as the size
of the latter.

\subsection{Term Automata}

$\mu$-types may be
understood as finite dags since their infinite
unfolding produce a regular (infinite) trees. We further simplify the types
whose dags we consider by flattening the union type constructor and switching
to an alphabet where unions are $n$-ary: $\mathfrak{L^n} \eqdef  \set{a^0
\mathrel| a \in \TypeVariable \cup \TypeConstant} \cup \set{\idatatype^2,
\ifunctype^2} \cup \set{\iuniontype^n \mathrel| n > 1}$ and we let $\TreeN$
stand for possibly infinite trees whose nodes are symbols in $\mathfrak{L^n}$.
$\mu$-types may be interpreted in $\TreeN$ simply by unfolding and then
considering maximal union types as their underlying $n$-ary union types. We
do so by resorting to our previous interpretation of $\mu$-types over $\Tree$
and the following translation.

\begin{definition}
\label{def:toNTree}
The function $\toNTree{\bullet} : \Tree \to \TreeN$ is defined inductively as
follows: $$
\begin{array}{r@{\quad\eqdef\quad}l@{\quad}l}
\toNTree{a}(\epsilon)                                & a \\
\toNTree{\tA_1 \star \tA_2}(\epsilon)                & \star                & \text{for $\star \in \set{\idatatype, \ifunctype}$} \\
\toNTree{\tA_1 \star \tA_2}(i\pi)                    & \toNTree{\tA_i}(\pi) & \text{for $\star \in \set{\idatatype, \ifunctype}$, $i \in \set{1,2}$} \\
\toNTree{\maxuniontype{i \in 1..n}{\tA_i}}(\epsilon) & \nuniontype{}{n}{}   & \text{$\tA_i \neq \iuniontype$}\\
\toNTree{\maxuniontype{i \in 1..n}{\tA_i}}(i\pi)     & \toNTree{\tA_i}(\pi) & \text{$\tA_i \neq \iuniontype$}
\end{array} $$
\end{definition}

Types in $\TreeN$ may be represented as \emph{term
automata}~\cite{DBLP:journals/toplas/AmadioC93}.

\begin{definition}  
A \emphdef{term automaton} is a tuple $\M = \langle Q, \Sigma, q_0, \delta,
l\rangle$ where:
\begin{enumerate}
  \itemsep0em
  \item $Q$ is a finite set of states.
  \item $\Sigma$ is an alphabet where each symbol has an associated arity.
  \item $q_0$ is the initial state.
  \item $\delta: Q \times \mathbb{N} \to Q$ is a partial transition function
  between states, defined over $1..k$, where $k$ is the arity of the symbol
  associated by $\ell$ to the origin state.
  \item $\ell: Q \to \Sigma$ is a total function that associates a symbol in
  $\Sigma$ to each state.
\end{enumerate}
\end{definition}

We write $\M_A$ for the automaton associated to type $A$. $\M_A$ recognizes all
paths from the root of $A$ to any of its sub-expressions. 
Fig.~\ref{fig:listTypeTreeAndAutomaton} illustrates an example type, namely
$\listType =
\rectype{\alpha}{\uniontype{\nil}{(\datatype{\datatype{\cons}{A}}{\alpha})}}$,
represented as an infinite tree and as a term automaton $\M_{\listType}$. If
$q_0$ is the initial state of $\M_{\listType}$ and $\widehat\delta$ denotes the
natural extension of $\delta$ to sequences of symbols, then
$\ell(\widehat\delta(q_0,211)) = \cons$. As mentioned, the regular
structure of trees arising from types yields automata with a finite number of states. 

\begin{figure}
\begin{center}
\begin{minipage}{.4\textwidth}
\begin{tikzpicture}
\tikzset{level 4/.style={distance from root=125pt}}
\qTree [.{$\iuniontype^2$} {\nil}
                           [.$@$ [.$@$ {\cons}
                                       $\tA$ ]
                                 [.{$\iuniontype^2$} {\nil}
                                                     [.$@$ [.$@$ {\cons}
                                                                 $\tA$ ]
                                                           {$\vdots$} ] ] ] ]
\end{tikzpicture}
\end{minipage}
\begin{minipage}{.4\textwidth}
\begin{tikzpicture}[shorten >=1pt,node distance=1.8cm,on grid,auto, every node/.style={scale=0.9},initial text=]
   \node[state,initial]     (union) {$\nuniontype{}{2}{}$};
   \node[state]             (nil)  [below left=of union]  {$\nil$};
   \node[state]             (app1) [below right=of union] {$\idatatype$};
   \node[state]             (app2) [below left=of app1]   {$\idatatype$};
   \node[state]             (cons) [below left=of app2]   {$\cons$};
   \node[state,draw=none]   (a)    [below right=of app2]  {$\M_A$};
   \path[->]
    (union) edge              node [swap] {\footnotesize{1}} (nil)
            edge [bend right] node [swap] {\footnotesize{2}} (app1)
    (app1)  edge              node        {\footnotesize{1}} (app2)
            edge [bend right] node [swap] {\footnotesize{2}} (union)
    (app2)  edge              node        {\footnotesize{1}} (cons)
            edge              node [swap] {\footnotesize{2}} (a);
\end{tikzpicture}
\end{minipage}
\end{center}
\caption{The type $\listType$ represented as an infinite tree and as a term automaton.}
\label{fig:listTypeTreeAndAutomaton}
\end{figure}


\input{scala_subtipado}


\input{scala_equivalencia}


\subsection{Type Checking}


Let us revisit type-checking ($\fTypeCheckName$). As already discussed,
it linearly traverses the input term, the most costly operations
being those that deal with application and abstraction. These cases involve
calling $\fSubtypeName$. Notice that these calls do not depend directly on
the input to $\fTypeCheckName$. However, a linear correspondence can be
established between the size of the types being considered in $\fSubtypeName$
and the input to the algorithm, since such expressions are built from elements
of $\Gamma$ (the input context) or from annotations in the input term itself.
Consider for instance $\fSubtype{\varnothing}{A}{B}$ with $a$ and $b$ the size
of each type resp. This has complexity $\O(a^2 b^2 d)$ and, from the
discussion above, we can refer to it as $\O(n^4 d)$, where $n$ is the size of the input
to $\fTypeCheckName$ (\ie that of $\Gamma$ plus $t$). Similarly, we may say
that $\fUnfoldName$ is linear in $n$.

We now analyze the application and abstraction cases of the algorithm in
detail:

\begin{description}
  \item[Application]
  First it performs a linear check on the type to verify if it is a datatype.
  If so it returns. Otherwise, a second linear check is required
  ($\fUnfoldName$) in order to then perform as many calls to $\fSubtypeName$ as
  elements there are in the union of the functional types. This yields a local
  complexity of $\O(n^4 d^2)$.

  \item[Abstraction]
  First there are as many calls to $\fTypeCheckPName$ (the algorithm for
  type-checking patterns) as branches the abstraction has. Note that
  $\fTypeCheckPName$ has linear complexity in the size of its input and this
  call is instantiated with arguments $p_i$ and $\theta_i$ which occur in the
  original term. All these calls, taken together, may thus be considered to
  have linear time complexity with respect to the input of $\fTypeCheckPName$.
  Then it is necessary to perform a quadratic number (in the number of
  branches) of checks on compatibility. We have already analyzed that
  compatibility in the worst case involves checking subtyping. If we assume a
  linear number of branches w.r.t. the input, we obtain a total complexity of
  $\O(n^6 d)$ for this case.

\end{description}

Finally, the total complexity of $\fTypeCheckName$ is governed by the case of
the abstraction, and is therefore $\O(n^7 d)$.

\subsection{Prototype implementation}

A prototype in Scala is available~\cite{EV:2015:Prototipo}. It implements
$\fTypeCheckName$ but resorts to the efficient algorithm for subtyping and type
equivalence described above. It also includes further optimizations. For
example, following a suggestion in~\cite{DBLP:conf/tlca/CosmoPR05}, the
order in which elements in $W$ are selected for evaluation relies on detecting
strongly connected components, using
Tarjan's~\cite{DBLP:journals/jacm/DowneyST80} algorithm of linear cost and
topologically sorting them in reverse order. In the absence of cycles this
results in evaluating every pair only after all its children have already been
considered. For cyclic types pairs for which no order can be determined are
encapsulated within the same strongly connected component.



%% file: scala_subtipado.tex
\subsection{Subtyping and Subtype Checking}

We next present a coinductive notion of subtyping over $\TreeN$. It is a binary
relation $\subtypeup{\R}$ \emph{up-to} a set of hypothesis $\R$
(Fig.~\ref{fig:subtypingSchemesNary}). For $\R = \varnothing$, $\subtypeup{\R}$
coincides with $\subtypemu$, modulo the proper translation.

\begin{figure} $$
\begin{array}{c}
\RuleCo{}{a \subtypeup{\R} a}{\ruleSubupRefl}
\\
\\
\RuleCo{\tD \mathrel{(\subtypeup{\R} \cup \mathrel{\R})} \tD'
        \quad
        \tA \mathrel{(\subtypeup{\R} \cup \mathrel{\R})} \tA'
       }
       {\datatype{\tD}{\tA} \subtypeup{\R} \datatype{\tD'}{\tA'}}
       {\ruleSubupComp}
\\
\\
\RuleCo{\tA' \mathrel{(\subtypeup{\R} \cup \mathrel{\R})} \tA
        \quad
        \tB \mathrel{(\subtypeup{\R} \cup \mathrel{\R})} \tB'
       }
       {\functype{\tA}{\tB} \subtypeup{\R} \functype{\tA'}{\tB'}}
       {\ruleSubupFunc}
\\
\\
\RuleCo{\tA_i \mathrel{(\subtypeup{\R} \cup \mathrel{\R})} \tB_{f(i)}
        \quad
        f : 1..n \to 1..m
        \quad
        \tA_i, \tB_j\neq \iuniontype}
       {\nuniontype{i}{n}{\tA_i} \subtypeup{\R} \nuniontype{j}{m}{\tB_j}}
       {\ruleSubupUnion}
\\
\\
\RuleCo{\tA_i \mathrel{(\subtypeup{\R} \cup \mathrel{\R})} \tB
        \text{ for all $i \in 1..n$}
        \quad
        \tA_i \neq \iuniontype
        \quad
        \tB \neq \iuniontype}
       {\nuniontype{i}{n}{\tA_i} \subtypeup{\R} \tB}
       {\ruleSubupUnionL}
\\
\\
\RuleCo{\tA \mathrel{(\subtypeup{\R} \cup \mathrel{\R})} \tB_{k}
        \text{ for some $k \in 1..m$}
        \quad
        \tA\neq \iuniontype
        \quad
        \tB_j\neq \iuniontype
        \quad}
       {\tA \subtypeup{\R} \nuniontype{j}{m}{\tB_j}}
       {\ruleSubupUnionR}
\end{array} $$
\caption{Subtyping relation \emph{up-to} $\R$ over $\TreeN$.}
\label{fig:subtypingSchemesNary}
\end{figure}

\begin{proposition}
\label{prop:subtypeaciSoundnessAndCompleteness}
$A \subtypemu B$ iff $\toNTree{\toBTree{A}} \subtypeaci \toNTree{\toBTree{B}}$.
\end{proposition}

So we can use $\subtypeaci$ to determine whether types are related via
$\subtypemu$: take two types, construct their automaton representation and
check whether these are related via $\subtypeaci$. Moreover, our formulation
of $\subtypeup{\R}$ will prove convenient for proving correctness of our
subtyping algorithm.

The proof of Prop.~\ref{prop:subtypeaciSoundnessAndCompleteness} resorts to
Prop.~\ref{prop:subtypeSoundnessAndCompleteness}
and~\ref{prop:subtypecoIffSubstypeaci} below.

\begin{proposition}
\label{prop:subtypecoIffSubstypeaci}
$\tA \subtypeco \tB$ iff $\toNTree{\tA} \subtypeaci \toNTree{\tB}$.
\end{proposition}

\begin{proof}
$\Rightarrow)$ We prove this part by showing that $\R \eqdef
\set{\pair{\toNTree{\tA}}{\toNTree{\tB}} \mathrel| \tA \subtypeco \tB}$ is
$\Phisubtypeaci$-dense. We proceed by analyzing the shape of any possible
element of $\R$.
\begin{itemize}
  \item $\pair{\toNTree{\tA}}{\toNTree{\tB}} = \pair{a}{a}$. Then
  $\pair{\toNTree{\tA}}{\toNTree{\tB}} \in \Phisubtypeaci(\R)$.
  
  \item $\pair{\toNTree{\tA}}{\toNTree{\tB}} =
  \pair{\datatype{\tT'}{\tT''}}{\datatype{\tS'}{\tS''}}$. Note that the
  translation of a type in $\Tree$ can only have symbol $\idatatype$ at the
  root if the original type already does. Then, we know that $\tA$ must be of
  the form $\datatype{\tD}{\tA'}$, where $\toNTree{\tD} = \tT'$ and
  $\toNTree{\tA'} = \tT''$. Similarly, we have $\tB = \datatype{\tD'}{\tB'}$
  with $\toNTree{\tD'} = \tS'$ and $\toNTree{\tB'} = \tS''$.

  By definition of $\R$ we have $\tA \subtypeco \tB$, and by rule
  $\ruleSubcoComp$ both $\tD \subtypeco \tD'$ and $\tA' \subtypeco \tB'$ hold.
  Then $\pair{\toNTree{\tD}}{\toNTree{\tD'}},
  \pair{\toNTree{\tA'}}{\toNTree{\tB'}} \in \R$. Thus,
  $\pair{\datatype{\toNTree{\tD}}{\toNTree{\tA'}}}{\datatype{\toNTree{\tD'}}{\toNTree{\tB'}}}
  \in \Phisubtypeaci(\R)$. Finally, by Def.~\ref{def:toNTree}, we have $$
\begin{array}{r@{\quad=\quad}l}
\pair{\datatype{\toNTree{\tD}}{\toNTree{\tA'}}}{\datatype{\toNTree{\tD'}}{\toNTree{\tB'}}}
 & \pair{\toNTree{\datatype{\tD}{\tA'}}}{\toNTree{\datatype{\tD'}{\tB'}}} \\
 & \pair{\toNTree{\tA}}{\toNTree{\tB}} \in \Phisubtypeaci(\R)
\end{array} $$
  
  \item $\pair{\toNTree{\tA}}{\toNTree{\tB}} =
  \pair{\functype{\tT'}{\tT''}}{\functype{\tS'}{\tS''}}$. Similarly to the
  previous case, by Def.~\ref{def:toNTree}, we have $\pair{\tA}{\tB} =
  \pair{\functype{\tA'}{\tA''}}{\functype{\tB'}{\tB''}}$ with $\toNTree{\tA'} =
  \tT'$, $\toNTree{\tA''} = \tT''$, $\toNTree{\tB'} = \tS'$ and
  $\toNTree{\tB''} = \tS''$. Then, from $\tA \subtypeco \tB$ we get
  $\pair{\toNTree{\tB'}}{\toNTree{\tA'}},
  \pair{\toNTree{\tA''}}{\toNTree{\tB''}} \in \R$ and thus $$
\begin{array}{r@{\quad=\quad}l}
\pair{\functype{\toNTree{\tA'}}{\toNTree{\tA''}}}{\functype{\toNTree{\tB'}}{\toNTree{\tB''}}}
 & \pair{\toNTree{\functype{\tA'}{\tA''}}}{\toNTree{\functype{\tB'}{\tB''}}} \\
 & \pair{\toNTree{\tA}}{\toNTree{\tB}} \in \Phisubtypeaci(\R)
\end{array} $$
  
  \item $\pair{\toNTree{\tA}}{\toNTree{\tB}} =
  \pair{\nuniontype{i}{n}{\tT_i}}{\tS}$ with $\tS \neq \iuniontype$. A tree of
  the form $\nuniontype{i}{n}{\tT_i}$ can only result from translating a
  maximal union of $n > 1$ elements, thus $\tA = \maxuniontype{i \in
  1..n}{\tA_i}$ where $\toNTree{\tA_i} = \tT_i$. On the other hand, $\tS \neq
  \iuniontype$ implies $\tB \neq \iuniontype$. From $\tA \subtypeco \tB$, by
  $\ruleSubcoUnion$, we know $\tA_i \subtypeco \tB$ for every $i \in 1..n$.
  Thus, $\pair{\toNTree{\tA_i}}{\toNTree{\tB}} \in \R$ for every $i \in 1..n$
  and we conclude $$
\begin{array}{r@{\quad=\quad}l}
\pair{\nuniontype{i}{n}{\toNTree{\tA_i}}}{\toNTree{\tB}}
 & \pair{\toNTree{\maxuniontype{i \in 1..n}{\tA_i}}}{\toNTree{\tB}} \\
 & \pair{\toNTree{\tA}}{\toNTree{\tB}} \in \Phisubtypeaci(\R)
\end{array} $$
  
  \item $\pair{\toNTree{\tA}}{\toNTree{\tB}} =
  \pair{\tT}{\nuniontype{j}{m}{\tS_j}}$ with $\tT \neq \iuniontype$. This case
  is similar to the previous one. We have $\tA \neq \iuniontype$ and $\tB =
  \maxuniontype{j \in 1..m}{\tB_j}$ with $m > 1$ and $\toNTree{\tB_j} = \tS_j$.
  From $\tA \subtypeco \tB$ we know there exists $k \in 1..m$ such that $\tA
  \subtypeco \tB_k$ by rule $\ruleSubcoUnion$. Thus,
  $\pair{\toNTree{\tA}}{\toNTree{\tB_j}} \in \R$ and we conclude $$
\begin{array}{r@{\quad=\quad}l}
\pair{\toNTree{\tA}}{\nuniontype{j}{m}{\toNTree{\tB_j}}}
 & \pair{\toNTree{\tA}}{\toNTree{\maxuniontype{j \in 1..m}{\tB_j}}} \\
 & \pair{\toNTree{\tA}}{\toNTree{\tB}} \in \Phisubtypeaci(\R)
\end{array} $$
  
  \item $\pair{\toNTree{\tA}}{\toNTree{\tB}} =
  \pair{\nuniontype{i}{n}{\tT_i}}{\nuniontype{j}{m}{\tS_j}}$. With a similar
  argument to the two previous cases we have maximal union types $\tA =
  \maxuniontype{i \in 1..n}{\tA_i}$ and $\tB = \maxuniontype{j \in
  1..m}{\tB_j}$ with $\toNTree{\tA_i} = \tT_i$ and $\toNTree{\tB_j} = \tS_j$.
  By $\ruleSubcoUnion$ once again, from $\tA \subtypeco \tB$ we know that there
  exists a function $f : 1..n \to 1..m$ such that $\tA_i \subtypeco \tB_{f(i)}$
  for every $i \in 1..n$. Thus, $\pair{\toNTree{\tA_i}}{\toNTree{\tB_{f(i)}}}
  \in \R$ and we resort to $\Phisubtypeaci$ and Def.~\ref{def:toNTree} to
  conclude $$
\begin{array}{r@{\quad=\quad}l}
\pair{\nuniontype{i}{n}{\toNTree{\tA_i}}}{\nuniontype{j}{m}{\toNTree{\tB_j}}}
 & \pair{\toNTree{\maxuniontype{i \in 1..n}{\tA_i}}}{\toNTree{\maxuniontype{j \in 1..m}{\tB_j}}} \\
 & \pair{\toNTree{\tA}}{\toNTree{\tB}} \in \Phisubtypeaci(\R)
\end{array} $$
\end{itemize}

$\Leftarrow)$ Similarly, we define $\R \eqdef \set{\pair{\tA}{\tB} \mathrel|
\toNTree{\tA} \subtypeaci \toNTree{\tB}}$ and show that it is
$\Phisubtypeco$-dense. We proceed my analyzing the shape of $\pair{\tA}{\tB}
\in R$:
\begin{itemize}
  \item $\pair{\tA}{\tB} = \pair{a}{a}$. Then $\pair{\tA}{\tB} \in
  \Phisubtypeco(\R)$ trivially.
  
  \item $\pair{\tA}{\tB} = \pair{\datatype{\tD}{\tA'}}{\datatype{\tD'}{\tB'}}$.
  By Def.~\ref{def:toNTree}, $\pair{\toNTree{\tA}}{\toNTree{\tB}} =
  \pair{\datatype{\toNTree{\tD}}{\toNTree{\tA'}}}{\datatype{\toNTree{\tD'}}{\toNTree{\tB'}}}$.
  Moreover, since $\toNTree{\tA} \subtypeaci \toNTree{\tB}$, we have
  $\toNTree{\tD} \subtypeaci \toNTree{\tD'}$ and $\toNTree{\tA'} \subtypeaci
  \toNTree{\tB'}$. Then, $\pair{\tD}{\tD'}, \pair{\tA'}{\tB'} \in \R$ and $$
\begin{array}{r@{\quad=\quad}l}
\pair{\tA}{\tB} & \pair{\datatype{\tD'}{\tA'}}{\datatype{\tD'}{\tB'}} \in \Phisubtypeco(\R)
\end{array} $$
  
  \item $\pair{\tA}{\tB} =
  \pair{\functype{\tA'}{\tA''}}{\functype{\tB'}{\tB''}}$. As in the previous
  case we resort to Def.~\ref{def:toNTree} to get
  $\pair{\toNTree{\tA}}{\toNTree{\tB}} =
  \pair{\functype{\toNTree{\tA'}}{\toNTree{\tA''}}}{\functype{\toNTree{\tB'}}{\toNTree{\tB''}}}$.
  Then, from $\toNTree{\tA} \subtypeaci \toNTree{\tB}$ we get $\toNTree{\tB'}
  \subtypeaci \toNTree{\tA'}$ and $\toNTree{\tA''} \subtypeaci \toNTree{\tB''}$
  by $\ruleSubupFunc$. Thus, $\pair{\tB'}{\tA'}, \pair{\tA''}{\tB''} \in \R$
  and $$
\begin{array}{r@{\quad=\quad}l}
\pair{\tA}{\tB} & \pair{\functype{\tA'}{\tA''}}{\functype{\tB'}{\tB''}} \in \Phisubtypeco(\R)
\end{array} $$

  \item $\pair{\tA}{\tB} = \pair{\maxuniontype{i \in
  1..n}{\tA_i}}{\maxuniontype{j \in 1..m}{\tB_j}}$ with $n + m > 2$ and $\tA_i,
  \tB_i \neq \iuniontype$. Then we need to distinguish three cases:
  \begin{itemize}
    \item If $m = 1$, then $n > 1$ and $\toNTree{\tA} =
    \nuniontype{i}{n}{\toNTree{\tA_i}}$. From $\toNTree{\tA} \subtypeaci
    \toNTree{\tB}$ we have $\toNTree{\tA_i} \subtypeaci \toNTree{\tB}$ for
    every $i \in 1..n$, by $\ruleSubupUnionL$. Then, $\pair{\tA_i}{\tB} \in \R$
    and by definition of $\Phisubtypeco$ we get $$
\begin{array}{r@{\quad=\quad}l}
\pair{\tA}{\tB} & \pair{\maxuniontype{i \in 1..n}{\tA_i}}{\tB} \in \Phisubtypeco(\R)
\end{array} $$
    
    \item If $n = 1$, then $m > 1$ and $\toNTree{\tB} =
    \nuniontype{j}{m}{\toNTree{\tB_j}}$. Here, by $\ruleSubupUnionR$, we have
    $\toNTree{\tA} \subtypeaci \toNTree{\tB_k}$ for some $k \in 1..m$, and thus
    $\pair{\tA}{\tB_k} \in \R$. Finally we conclude $$
\begin{array}{r@{\quad=\quad}l}
\pair{\tA}{\tB} & \pair{\tA}{\maxuniontype{j \in 1..m}{\tB_j}} \in \Phisubtypeco(\R)
\end{array} $$
    
    \item If $n > 1$ and $m > 1$, then both $\tA$ and $\tB$ are maximal union
    types with at least two elements each. Here it applies $\ruleSubupUnion$
    and we have $\toNTree{\tA} = \nuniontype{i}{n}{\toNTree{\tA_i}}$,
    $\toNTree{\tB} = \nuniontype{j}{m}{\toNTree{\tB_j}}$, and there exists $f :
    1..n \to 1..m$ such that $\toNTree{\tA_i} \subtypeaci
    \toNTree{\tB_{f(i)}}$. Then, $\pair{\tA_i}{\tB_{f(i)}} \in \R$ and we
    conclude $$
\begin{array}{r@{\quad=\quad}l}
\pair{\tA}{\tB} & \pair{\maxuniontype{i \in 1..n}{\tA_i}}{\maxuniontype{j \in 1..m}{\tB_j}} \in \Phisubtypeco(\R)
\end{array} $$
  \end{itemize}
\end{itemize}
\end{proof}

\subsubsection{Algorithm Description}

The algorithm that checks whether types are related by the new subtyping
relation builds on ideas from~\cite{DBLP:conf/tlca/CosmoPR05}. 
Our presentation is more general than required for subtyping; this general
scheme will also be applicable to type equivalence, as we shall later see.
Call $p \in \TreeN \times \TreeN$ \emph{valid} if $p \in
\mathbin{\subtypeaci}$. The algorithm consists of two phases. The aim of the
first one is to construct a set $U \subseteq \TreeN \times \TreeN$ that
delimits the universe of pairs of types that will later be refined to obtain a
set of only valid pairs. It starts off with an initial pair (\cf
Fig.~\ref{fig:pseudocode:subtyping2:phase1}, $\fBuildUniverseName$) and
then explores pairs of sub-terms of both types in this pair by decomposing the
type constructors (\cf Fig.~\ref{fig:pseudocode:subtyping2:phase1},
$\fChildrenName$). Note that, given $p$, the algorithm may add invalid
pairs in order to prove the validity of $p$. The second phase shall be in
charge of eliminating these invalid pairs. Note that the first phase can easily be
adapted to other relations by simply redefining function $\fChildrenName$.

$U$ may be interpreted as a directed graph where an edge from pair $p$ to
$q$ means that $q$ might belong to the support set of $p$ in the final relation
$\subtypeaci$. In this case we say that $p$ is a \emphdef{parent} of $q$. Since
types could have cycles, a pair may be added to $U$ more than once and hence
have more than one parent. Set $u(p)$ to be the \emphdef{incoming degree} of
$p$, \ie the number of parents.

\begin{figure} $$
\begin{array}{cc}
\begin{array}{rcl}
\multicolumn{3}{l}{\fBuildUniverse{p_0}:} \\
 & & U = \emptyset \\
 & & W = \set{p_0} \\
 & & \ttWhile W \neq \emptyset: \\
 & & \quad p := \fTakeOne{W} \\
 & & \quad \ttIf p \notin U \\
 & & \qquad \setinsert{p}{U} \\
 & & \qquad \ttForeach q \in \fChildren{p} \\
 & & \qqquad \setinsert{q}{W} \\
 & & \ttReturn\ U
\end{array}
&
\begin{array}{rcl}
\multicolumn{3}{l}{\fChildren{p}:}\\ 
 & & \ttCase\ p\ \ttOf\\
 & & \quad \pair{\datatype{\tD}{\tA}}{\datatype{\tD'}{\tB}} \rightarrow \\
 & & \qqquad \set{\pair{\tD}{\tD'},\pair{\tA}{\tB}} \\
 & & \quad \pair{\functype{\tA'}{\tA''}}{\functype{\tB'}{\tB''}} \rightarrow \\
 & & \qqquad \set{\pair{\tB'}{\tA'},\pair{\tA''}{\tB''}} \\
 & & \quad \pair{\nuniontype{i}{n}{\tA_i}}{\nuniontype{j}{m}{\tB_j}} \rightarrow \\
 & & \qqquad \set{\pair{\tA_i}{\tB_j} \ |\ i\in1..n,\ j\in1..m} \\
 & & \quad \pair{\nuniontype{i}{n}{\tA_i}}{\tB},\ \tB \neq \iuniontype \rightarrow \\
 & & \qqquad \set{\pair{\tA_i}{\tB} \ |\ i\in1..n} \\
 & & \quad \pair{\tA}{\nuniontype{j}{m}{\tB_j}},\ \tA \neq \iuniontype \rightarrow \\
 & & \qqquad \set{\pair{\tA}{\tB_j} \ |\ j\in1..m} \\
 & & \quad \mathtt{otherwise} \rightarrow \\
 & & \qqquad \emptyset
\end{array}
\end{array} $$
\caption{Pseudo-code of the first phase of the algorithm (construction of the universe $U$).}
\label{fig:pseudocode:subtyping2:phase1}
\end{figure}

During the second phase (Fig.~\ref{fig:pseudocode:subtyping2:phase2},
$\fGfpName$) the following sets are maintained, all of which conform a
partition of $U$:
\begin{itemize}
  \itemsep0em
  \item $W$: pairs whose validity has yet to be determined
  \item $S$: pairs considered conditionally valid
  \item $F$: invalid pairs
\end{itemize}

The algorithm repeatedly takes elements in $W$ and, in each iteration,
transfers to $S$ the selected pair $p$ if its validity can be proved
assuming valid only those pairs which have not been discarded up until now
(\ie those in $W \cup S$). Otherwise, $p$ is transferred to $F$ and all of its
parents in $S$ need to be reconsidered (their validity up-to $W$ may have
changed). Thus these parents are moved back to $W$ (\cf
Fig.~\ref{fig:pseudocode:subtyping2:phase2}, $\fInvalidateName$).
Intuitively, $S$ contains elements in $\subtypeup{W}$. The process ends when
$W$ is empty. The only aspect of this second phase specific to $\subtypeup{W}$
is function $\fCheckName$, which may be redefined to be other
suitable \emph{up-to} relations.

\begin{figure} $$
\begin{array}{c@{\qquad}c}
\begin{array}{rcl}
\multicolumn{3}{l}{\fGfp{\tA}{\tB}:} \\
 & & W = \fBuildUniverse{\pair{\tA}{\tB}} \\
 & & S = \emptyset \\
 & & F = \emptyset \\
 & & \ttWhile W \neq \emptyset: \\
 & & \quad p := \fTakeOne{W} \\
 & & \quad \ttIf \fCheck{p}{F} \\
 & & \qquad  \ttThen \setinsert{p}{S} \\
 & & \qquad  \ttElse \fInvalidate{p}{S}{F}{W} \\
 & & \ttReturn\ p \in S \\
\\
\multicolumn{3}{l}{\fInvalidate{p}{S}{F}{W}:} \\
 & & \setinsert{p}{F} \\
 & & \ttForeach q \in \parents{p} \cap S \\
 & & \quad \setmove{q}{S}{W}
\end{array}
&
\begin{array}{rcl}
\multicolumn{3}{l}{\fCheck{p}{F}:} \\ 
  & & \ttCase\ p\ \ttOf \\
  & & \quad \pair{a}{a} \rightarrow \\
  & & \qqquad \ttTrue \\
  & & \quad \pair{\datatype{\tD}{\tA}}{\datatype{\tD'}{\tB}} \rightarrow \\
  & & \qqquad \pair{\tD}{\tD'} \notin F \ttAnd \pair{\tA}{\tB} \notin F \\
  & & \quad \pair{\functype{\tA'}{\tA''}}{\functype{\tB'}{\tB''}} \rightarrow \\
  & & \qqquad \pair{\tB'}{\tA'} \notin F \ttAnd \pair{\tA''}{\tB''} \notin F \\
  & & \quad \pair{\nuniontype{i}{n}{\tA_i}}{\nuniontype{j}{m}{\tB_j}} \rightarrow \\
  & & \qqquad \forall i. \exists j.\ \pair{\tA_i}{\tB_j} \notin F \\
  & & \quad \pair{\nuniontype{i}{n}{\tA_i}}{\tB},\ \tB \neq \iuniontype \rightarrow \\
  & & \qqquad \forall i.\ \pair{\tA_i}{\tB} \notin F \\
  & & \quad \pair{\tA}{\nuniontype{j}{m}{\tB_j}},\ \tA \neq \iuniontype \rightarrow \\
  & & \qqquad \exists m.\ \pair{\tA}{\tB_m} \notin F
\end{array}
\end{array} $$
\caption{Pseudo-code of the second phase (relation refinement).}
\label{fig:pseudocode:subtyping2:phase2}
\end{figure}

\subsubsection{Correctness}

It is based on the fact that $S$ may be considered a set of valid pairs
\emph{assuming the validity of those in  $W$}. More generally, the following
holds:

\begin{proposition}
The algorithm preserves the following invariant:
\begin{itemize}
  \itemsep0em
  \item \triple{W}{S}{F} is a partition of $U$
  \item $F$ is composed solely of invalid pairs
  \item $S \subseteq \Phisubtypeup{W}(S)$
\end{itemize}
\end{proposition}

\begin{proof}
The tree conditions clearly hold at the beginning of the second phase, where $W
= U$, and both $S$ and $F$ are empty. Moreover, during the process, pairs are
transferred between the three sets preserving the first condition.

Notice that, in each iteration, the decision of whether to invalidate a pair
$p$ or move it to $S$ is made analyzing if there are enough elements in $S \cup
W$ to prove its validity (\cf Fig.~\ref{fig:pseudocode:subtyping2:phase2},
$\fCheckName$). We show this by analyzing the structure of $p$:

\begin{itemize}
  \item $p = \pair{a}{a}$. As we can see in $\ruleSubupRefl$, the validity of
  $p$ does not depend on any other pair (\ie~its support set is empty). As
  expected, the algorithm transfers $p$ to $S$ directly, without further
  checks.
  
  \item $p = \pair{\datatype{\tD}{\tA}}{\datatype{\tD'}{\tA'}}$. By
  $\ruleSubupComp$, $p$ is only valid if $\pair{\tD}{\tD'}$ and
  $\pair{\tA}{\tA'}$ also are. Those are exactly the pairs that are evaluated
  to decide whether to invalidate $p$ or not.
  
  \item $p = \pair{\functype{\tA}{\tB}}{\functype{\tA'}{\tB'}}$. This case
  corresponds to $\ruleSubupFunc$, $p$ will be moved to $S$ only if
  $\pair{\tA'}{\tA}$ and $\pair{\tB}{\tB'}$ has not been discarded yet.

  \item $p = \pair{\nuniontype{i}{n}{\tA_i}}{\nuniontype{j}{m}{\tB_j}}$.
  Here the algorithm checks if for each $\tA_i$ there exists $\tB_j$ not yet
  invalidated, to continue considering $p$ a valid pair. This is exactly the
  hypothesis of rule $\ruleSubupUnion$.
  
  \item $p = \pair{\nuniontype{i}{n}{\tA_i}}{\tB}$. Similarly, as indicated in
  $\ruleSubupUnionL$, the process checks that no pair $\pair{\tA_i}{\tB}$ was
  already invalidated.
  
  \item $p = \pair{\tA}{\nuniontype{j}{m}{\tB_j}}$. Here, by rule
  $\ruleSubupUnionR$, it is enough to find a non-invalid pair
  $\pair{\tA}{\tB_j}$ to keep $p$ in the set of conditionally valid pairs.
\end{itemize}

If there is not enough elements in $S \cup W$ to consider $p$ a valid pair,
then it is considered invalid and moved to $F$, preserving the second condition
of the invariant.

We now check the last condition. We use indexes 0 and 1 to distinguish between
the state of the sets at the beginning and end of each iteration respectively.
We consider two cases depending on the decision made on $p$:
\begin{itemize}
  \item If $p$ is moved to $S$ (\ie all the necessary elements to prove its
  validity are in $W_0 \cup S_0$) we have: $$
\begin{array}{r@{\quad=\quad}l}
W_1 & W_0 \setminus \set{p} \\
S_1 & S_0 \cup \set{p} \\
F_1 & F_0
\end{array} $$
  By definition we have $S_0 \subseteq \Phisubtypeup{W_0}(S_0)$, hence $$S_0
  \cup \set{p} \subseteq \Phisubtypeup{W_0}(S_0) \cup \set{p}$$ Notice that
  $S_1 \cup W_1 = S_0 \cup W_0$. Then, by definition of $\Phisubtypeup{\R}$ we
  get $\Phisubtypeup{W_0}(S_0) = \Phisubtypeup{W_1}(S_1)$. Thus, $$S_0 \cup
  \set{p} \subseteq \Phisubtypeup{W_1}(S_1) \cup \set{p} $$ The set on the
  left-hand side is exactly $S_1$. Moreover, since the necessary elements to
  consider $p$ a valid pair are in $S_1 \cup W_1$, we have $p \in
  \Phisubtypeup{W_1}(S_1)$. Finally, we conclude $$S_1 \subseteq
  \Phisubtypeup{W_1}(S_1)$$
  
  \item If $p$ is invalidated we have: $$
\begin{array}{r@{\quad=\quad}l}
W_1 & W_0 \setminus \set{p} \cup Q\\
S_1 & S_0 \setminus Q \\
F_1 & F_0 \cup \set{p}
\end{array} $$
  where $Q$ contains the parents of $p$ that belong to $S_0$. We proceed to
  show that $S_1 \subseteq \Phisubtypeup{W_1}(S_1)$. Assume toward a
  contradiction that there exists $s \in S_1$ such that $s \notin
  \Phisubtypeup{W_1}(S_1)$. Then $s$ cannot have an empty support set.
  Moreover, from $S_1 \subseteq S_0$ we know $s \in S_0$.
  
  On the other hand, from $s \notin \Phisubtypeup{W_1}(S_1)$ we deduce that
  there exists an element in the support set of $s$ that is not in $S_1 \cup
  W_1$ but belongs to $S_0 \cup W_0$, since $S_0$ is
  $\Phisubtypeup{W_0}(S_0)$-dense. By the equalities above, the only possible
  such element is $p$. Thus, $s$ is a parent of $p$. Given that the algorithm
  moves to $W_1$ all the parents of $p$ that belong to $S_0$, we should have $s
  \in W_1$. This leads to a contradiction since $s \in S_1$ and the sets $W_1,
  S_1, F_1$ form a partition of $U$.
  
  The contradiction comes from assuming the existence of such $s$, hence we
  conclude $$S_1 \subseteq \Phisubtypeup{W_1}(S_1)$$
\end{itemize}
\end{proof}

When the main cycle ends we know that $W$ is empty, and therefore that $S
\subseteq \Phisubtypeup{\varnothing}(S)$. The coinduction principle then
yields $S \subseteq \mathbin{\subtypeaci}$ (\ie~every pair in $S$ is valid) and
therefore we are left to  verify whether the original pair of types is in $S$
or $F$.

\subsubsection{Complexity}
\label{subtyping2:complexity}

The first phase consists of identifying relevant pairs of sub-terms in both
types being evaluated. If we call $N$ and $N'$ the size of such types
(considering nodes and edges in their representations) we have that the size
and cost of building the universe $U$ can be bounded by $\O(N N')$. As we shall
see, the total cost of the algorithm is governed by the amount of operations in
the second phase.

As stated in~\cite{DBLP:conf/tlca/CosmoPR05}, since any pair can be
invalidated at most once (in which case $u(p)$ nodes are transferred back to
$W$ for reconsideration) the amount of iterations in the refinement stage can
be bounded by $$
\begin{array}{rcccl}
  \sum_{p \in U} 1 + \sum_{p \in U} u(p) & = & \sum_{p \in U} (1 + u(p)) & = & size(U)
\end{array}  
$$

Assuming that set operations can be performed in constant time, the cost of
evaluating each node in the main loop is that of deciding whether to suspend or
invalidate the pair and, in the later case, the cost of the invalidation
process. The decision of where to transfer the node is computed in the function
$\fCheckName$, which always performs a constant amount of operations for pairs
of non-union types. The worst case involves checking pairs of the form
$\pair{\nuniontype{i}{n}A_i}{\nuniontype{j}{m}B_i}$, which can be resolved by
maintaining in each node a table indicating, for every $A_i$, the amount of
pairs $\pair{A_i}{B_j}$ that have not yet been invalidated. Using this approach,
this check can be performed in $\O(d)$ operations, where $d$ is a bound on the
size of both unions. Whenever a pair is invalidated, all tables present in its
immediate parents are updated as necessary.

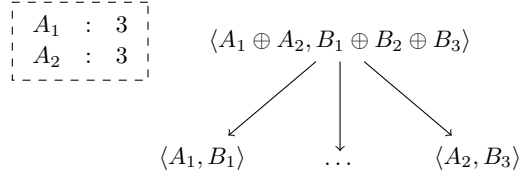
\begin{figure}[h]
\begin{center}
\begin{tikzpicture}[shorten >=1pt,node distance=1.8cm,on grid,auto, every node/.style={scale=0.9},initial text=]
  \node[] (p1)                      {$\pair{\uniontype{A_1}{A_2}}{\uniontype{B_1}{\uniontype{B_2}{B_3}}}$};
  \node[] (p2) [below left=of p1,yshift=-0.3cm,xshift=-0.6cm]   {$\pair{A_1}{B_1}$};
  \node[] (p3) [below of=p1]        {\ldots};
  \node[] (p4) [below right=of p1,yshift=-0.3cm,xshift=0.6cm]  {$\pair{A_2}{B_3}$};

  \node[state,rectangle,dashed] (t1) [left=of p1,xshift=-1.8cm] {$
    \begin{array}{rcl}
    A_1 & : & 3 \\
    A_2 & : & 3
    \end{array}
  $};
  \node[state,rectangle,draw=none] (t2) [right=of p1,xshift=1.8cm] {
  \hphantom{$
    \begin{array}{rcl}
    A_1 & : & 3 \\
    A_2 & : & 3
    \end{array}
  $}};
   
  \draw [->](p1) -- (p2);
  \draw [->](p1) -- (p3);
  \draw [->](p1) -- (p4);
\end{tikzpicture}
\end{center}
\caption{Verification of invalidated descendants.}
\label{fig:subtyping2:failureTableOptimization}
\end{figure}

Finally we resort to an argument introduced
in~\cite{DBLP:conf/tlca/CosmoPR05} to argue that the cost of
invalidating an element can be seen as $\O(1)$: a new iteration will be
performed for each of the $u(p)$ pairs added to $W$ when invalidating $p$.
Because of this, a more precise bound for the cost of the complete execution of
the algorithm can be obtained if we consider the cost of adding each of these
elements to $W$ as part of the iteration itself, yielding an amortized cost of
$\O(d)$ operations for each iteration. This leaves a total cost of $\O(size(U)
d)$ for the subtyping check, expressed as $\O(N N' d)$ in terms of the size of
the input automata.

Let us call $n$ and $n'$ the amount of constructors in types $A$ and $B$,
respectively. $N$ and $N'$ are the sizes of automata representing these types,
and can consequently be bounded by $\O(n^2)$ and $\O(n'^2)$. Therefore, the
complexity of the algorithm can be expressed as $\O(n^2 n'^2 d)$.


%% file: scala_equivalencia.tex
\subsection{Equivalence Checking}

In this section we adapt the previous algorithm to obtain one proper of
equivalence checking with the same complexity cost.
Fig.~\ref{fig:equalitySchemesNary} introduces an equivalence relation
\emph{up-to} $\R$ over $\TreeN$ which can be used to compute $\eqtypemu$ via
the translation $\toNTree{\bullet}$.

\begin{figure} $$
\begin{array}{c}
\RuleCo{}{a \eqtypeup{\R} a}{\ruleEqupRefl}
\\
\\
\RuleCo{\tD \mathrel{(\eqtypeup{\R} \cup \mathrel{\R})} \tD'
        \quad
        \tA \mathrel{(\eqtypeup{\R} \cup \mathrel{\R})} \tA'
       }
       {\datatype{\tD}{\tA} \eqtypeup{\R} \datatype{\tD'}{\tA'}}
       {\ruleEqupComp}
\\
\\
\RuleCo{\tA' \mathrel{(\eqtypeup{\R} \cup \mathrel{\R})} \tA
        \quad
        \tB \mathrel{(\eqtypeup{\R} \cup \mathrel{\R})} \tB'
       }
       {\functype{\tA}{\tB} \eqtypeup{\R} \functype{\tA'}{\tB'}}
       {\ruleEqupFunc}
\\
\\
\RuleCo{\begin{array}{ll}
          \tA_i \mathrel{(\eqtypeup{\R} \cup \mathrel{\R})} \tB_{f(i)} & \quad f : 1..n \to 1..m \\
          \tA_{g(j)} \mathrel{(\eqtypeup{\R} \cup \mathrel{\R})} \tB_j & \quad g : 1..m \to 1..n
        \end{array}
        \quad
        \tA_i, \tB_j\neq \iuniontype}
       {\nuniontype{i}{n}{\tA_i} \eqtypeup{\R} \nuniontype{j}{m}{\tB_j}}
       {\ruleEqupUnion}
\\
\\
\RuleCo{\tA_i \mathrel{(\eqtypeup{\R} \cup \mathrel{\R})} \tB
        \text{ for all $i \in 1..n$}
        \quad
        \tA_i \neq \iuniontype
        \quad
        \tB \neq \iuniontype}
       {\nuniontype{i}{n}{\tA_i} \eqtypeup{\R} \tB}
       {\ruleEqupUnionL}
\\
\\
\RuleCo{\tA \mathrel{(\eqtypeup{\R} \cup \mathrel{\R})} \tB_{j}
        \text{ for all $j \in 1..m$}
        \quad
        \tA\neq \iuniontype
        \quad
        \tB_j\neq \iuniontype}
       {\tA \eqtypeup{\R} \nuniontype{j}{m}{\tB_j}}
       {\ruleEqupUnionR}
\end{array} $$
\caption{Equivalence relation \emph{up-to} $\R$ over $\TreeN$.}
\label{fig:equalitySchemesNary}
\end{figure}

\begin{proposition}
\label{prop:eqtypeaciSoundnessAndCompleteness}
$A \eqtypemu B$ iff $\toNTree{\toBTree{A}} \eqtypeaci \toNTree{\toBTree{B}}$.
\end{proposition}

As done before for subtyping, we prove this by resorting to
Prop.~\ref{prop:eqtypeSoundnessAndCompleteness}
and~\ref{prop:eqtypecoIffEqtypeaci} below:

\begin{proposition}
\label{prop:eqtypecoIffEqtypeaci}
$\tA \eqtypeco \tB$ iff $\toNTree{\tA} \eqtypeaci \toNTree{\tB}$.
\end{proposition}

\begin{proof}
$\Rightarrow)$
We prove this part by showing that $\R \eqdef
\set{\pair{\toNTree{\tA}}{\toNTree{\tB}} \mathrel| \tA \eqtypeco \tB}$ is
$\Phieqtypeaci$-dense. We proceed by analyzing the shape of any possible
element of $\R$.
\begin{itemize}
  \item $\pair{\toNTree{\tA}}{\toNTree{\tB}} = \pair{a}{a}$. Then
  $\pair{\toNTree{\tA}}{\toNTree{\tB}} \in \Phieqtypeaci(\R)$.
  
  \item $\pair{\toNTree{\tA}}{\toNTree{\tB}} =
  \pair{\datatype{\tT'}{\tT''}}{\datatype{\tS'}{\tS''}}$. Note that the
  translation of a type in $\Tree$ can only have symbol $\idatatype$ at the
  root if the original type already does. Then, we know that $\tA$ must be of
  the form $\datatype{\tD}{\tA'}$, where $\toNTree{\tD} = \tT'$ and
  $\toNTree{\tA'} = \tT''$. Similarly, we have $\tB = \datatype{\tD'}{\tB'}$
  with $\toNTree{\tD'} = \tS'$ and $\toNTree{\tB'} = \tS''$.

  By definition of $\R$ we have $\tA \eqtypeco \tB$, and by rule
  $\ruleEqcoComp$ both $\tD \eqtypeco \tD'$ and $\tA' \eqtypeco \tB'$ hold.
  Then $\pair{\toNTree{\tD}}{\toNTree{\tD'}},
  \pair{\toNTree{\tA'}}{\toNTree{\tB'}} \in \R$. Thus,
  $\pair{\datatype{\toNTree{\tD}}{\toNTree{\tA'}}}{\datatype{\toNTree{\tD'}}{\toNTree{\tB'}}}
  \in \Phieqtypeaci(\R)$. Finally, by Def.~\ref{def:toNTree}, we have $$
\begin{array}{r@{\quad=\quad}l}
\pair{\datatype{\toNTree{\tD}}{\toNTree{\tA'}}}{\datatype{\toNTree{\tD'}}{\toNTree{\tB'}}}
 & \pair{\toNTree{\datatype{\tD}{\tA'}}}{\toNTree{\datatype{\tD'}{\tB'}}} \\
 & \pair{\toNTree{\tA}}{\toNTree{\tB}} \in \Phieqtypeaci(\R)
\end{array} $$
  
  \item $\pair{\toNTree{\tA}}{\toNTree{\tB}} =
  \pair{\functype{\tT'}{\tT''}}{\functype{\tS'}{\tS''}}$. Similarly to the
  previous case, by Def.~\ref{def:toNTree}, we have $\pair{\tA}{\tB} =
  \pair{\functype{\tA'}{\tA''}}{\functype{\tB'}{\tB''}}$ with $\toNTree{\tA'} =
  \tT'$, $\toNTree{\tA''} = \tT''$, $\toNTree{\tB'} = \tS'$ and
  $\toNTree{\tB''} = \tS''$. Then, from $\tA \eqtypeco \tB$ we get
  $\pair{\toNTree{\tA'}}{\toNTree{\tB'}},
  \pair{\toNTree{\tA''}}{\toNTree{\tB''}} \in \R$ and thus $$
\begin{array}{r@{\quad=\quad}l}
\pair{\functype{\toNTree{\tA'}}{\toNTree{\tA''}}}{\functype{\toNTree{\tB'}}{\toNTree{\tB''}}}
 & \pair{\toNTree{\functype{\tA'}{\tA''}}}{\toNTree{\functype{\tB'}{\tB''}}} \\
 & \pair{\toNTree{\tA}}{\toNTree{\tB}} \in \Phieqtypeaci(\R)
\end{array} $$
  
  \item $\pair{\toNTree{\tA}}{\toNTree{\tB}} =
  \pair{\nuniontype{i}{n}{\tT_i}}{\tS}$ with $\tS \neq \iuniontype$. A tree of
  the form $\nuniontype{i}{n}{\tT_i}$ can only result from translating a
  maximal union of $n > 1$ elements, thus $\tA = \maxuniontype{i \in
  1..n}{\tA_i}$ where $\toNTree{\tA_i} = \tT_i$. On the other hand, $\tS \neq
  \iuniontype$ implies $\tB \neq \iuniontype$. From $\tA \eqtypeco \tB$, by
  $\ruleEqcoUnion$, we know $\tA_i \eqtypeco \tB$ for every $i \in 1..n$.
  Thus, $\pair{\toNTree{\tA_i}}{\toNTree{\tB}} \in \R$ for every $i \in 1..n$
  and we conclude $$
\begin{array}{r@{\quad=\quad}l}
\pair{\nuniontype{i}{n}{\toNTree{\tA_i}}}{\toNTree{\tB}}
 & \pair{\toNTree{\maxuniontype{i \in 1..n}{\tA_i}}}{\toNTree{\tB}} \\
 & \pair{\toNTree{\tA}}{\toNTree{\tB}} \in \Phieqtypeaci(\R)
\end{array} $$
  
  \item $\pair{\toNTree{\tA}}{\toNTree{\tB}} =
  \pair{\tT}{\nuniontype{j}{m}{\tS_j}}$ with $\tT \neq \iuniontype$. This case
  is similar to the previous one. We have $\tA \neq \iuniontype$ and $\tB =
  \maxuniontype{j \in 1..m}{\tB_j}$ with $m > 1$ and $\toNTree{\tB_j} = \tS_j$.
  From $\tA \eqtypeco \tB$ we know, by $\ruleEqcoUnion$, that $\tA \eqtypeco
  \tB_j$ for every $j \in 1..m$. Thus, $\pair{\toNTree{\tA}}{\toNTree{\tB_j}}
  \in \R$ and we conclude $$
\begin{array}{r@{\quad=\quad}l}
\pair{\toNTree{\tA}}{\nuniontype{j}{m}{\toNTree{\tB_j}}}
 & \pair{\toNTree{\tA}}{\toNTree{\maxuniontype{j \in 1..m}{\tB_j}}} \\
 & \pair{\toNTree{\tA}}{\toNTree{\tB}} \in \Phieqtypeaci(\R)
\end{array} $$
  
  \item $\pair{\toNTree{\tA}}{\toNTree{\tB}} =
  \pair{\nuniontype{i}{n}{\tT_i}}{\nuniontype{j}{m}{\tS_j}}$. With a similar
  argument to the two previous cases we have maximal union types $\tA =
  \maxuniontype{i \in 1..n}{\tA_i}$ and $\tB = \maxuniontype{j \in
  1..m}{\tB_j}$ with $\toNTree{\tA_i} = \tT_i$ and $\toNTree{\tB_j} = \tS_j$.
  By $\ruleEqcoUnion$ once again, from $\tA \eqtypeco \tB$ we know that there
  exists a function $$
\begin{array}{r@{\quad\text{such that}\quad}l}
f : 1..n \to 1..m & \tA_i \eqtypeco \tB_{f(i)} \\
g : 1..m \to 1..n & \tA_{g(j)} \eqtypeco \tB_j
\end{array} $$
  Thus, $\pair{\toNTree{\tA_i}}{\toNTree{\tB_{f(i)}}},
  \pair{\toNTree{\tA_{g(j)}}}{\toNTree{\tB_j}} \in \R$. and we resort to 
  $\Phieqtypeaci$ and Def.~\ref{def:toNTree} to
  conclude $$
\begin{array}{r@{\quad=\quad}l}
\pair{\nuniontype{i}{n}{\toNTree{\tA_i}}}{\nuniontype{j}{m}{\toNTree{\tB_j}}}
 & \pair{\toNTree{\maxuniontype{i \in 1..n}{\tA_i}}}{\toNTree{\maxuniontype{j \in 1..m}{\tB_j}}} \\
 & \pair{\toNTree{\tA}}{\toNTree{\tB}} \in \Phieqtypeaci(\R)
\end{array} $$
\end{itemize}

$\Leftarrow)$ Similarly, we define $\R \eqdef \set{\pair{\tA}{\tB} \mathrel|
\toNTree{\tA} \eqtypeaci \toNTree{\tB}}$ and show that it is
$\Phieqtypeco$-dense. We proceed my analyzing the shape of $\pair{\tA}{\tB}
\in R$:
\begin{itemize}
  \item $\pair{\tA}{\tB} = \pair{a}{a}$. Then $\pair{\tA}{\tB} \in
  \Phieqtypeco(\R)$ trivially.
  
  \item $\pair{\tA}{\tB} = \pair{\datatype{\tD}{\tA'}}{\datatype{\tD'}{\tB'}}$.
  By Def~\ref{def:toNTree}, $\pair{\toNTree{\tA}}{\toNTree{\tB}} =
  \pair{\datatype{\toNTree{\tD}}{\toNTree{\tA'}}}{\datatype{\toNTree{\tD'}}{\toNTree{\tB'}}}$.
  Moreover, since $\toNTree{\tA} \eqtypeaci \toNTree{\tB}$, we have
  $\toNTree{\tD} \eqtypeaci \toNTree{\tD'}$ and $\toNTree{\tA'} \eqtypeaci
  \toNTree{\tB'}$. Then $\pair{\tD}{\tD'}, \pair{\tA'}{\tB'} \in \R$ and $$
\begin{array}{r@{\quad=\quad}l}
\pair{\tA}{\tB} & \pair{\datatype{\tD'}{\tA'}}{\datatype{\tD'}{\tB'}} \in \Phieqtypeco(\R)
\end{array} $$
  
  \item $\pair{\tA}{\tB} =
  \pair{\functype{\tA'}{\tA''}}{\functype{\tB'}{\tB''}}$.  As in the previous
  case we resort to Def.~\ref{def:toNTree} to get
  $\pair{\toNTree{\tA}}{\toNTree{\tB}} =
  \pair{\functype{\toNTree{\tA'}}{\toNTree{\tA''}}}{\functype{\toNTree{\tB'}}{\toNTree{\tB''}}}$.
  Then, from $\toNTree{\tA} \eqtypeaci \toNTree{\tB}$ we get $\toNTree{\tA'}
  \eqtypeaci \toNTree{\tB'}$ and $\toNTree{\tA''} \eqtypeaci \toNTree{\tB''}$,
  by $\ruleEqupFunc$. Thus, $\pair{\tA'}{\tB'}, \pair{\tA''}{\tB''} \in \R$
  and $$
\begin{array}{r@{\quad=\quad}l}
\pair{\tA}{\tB} & \pair{\functype{\tA'}{\tA''}}{\functype{\tB'}{\tB''}} \in \Phieqtypeco(\R)
\end{array} $$
  
  \item $\pair{\tA}{\tB} = \pair{\maxuniontype{i \in
  1..n}{\tA_i}}{\maxuniontype{j \in 1..m}{\tB_j}}$ with $n + m > 2$ and $\tA_i,
  \tB_i \neq \iuniontype$. Then we need to distinguish three cases:
  \begin{itemize}
    \item If $m = 1$, then $n > 1$ and $\toNTree{\tA} =
    \nuniontype{i}{n}{\toNTree{\tA_i}}$. From $\toNTree{\tA} \eqtypeaci
    \toNTree{\tB}$ we have $\toNTree{\tA_i} \eqtypeaci \toNTree{\tB}$ for every
    $i \in 1..n$, by $\ruleEqupUnionL$. Then, $\pair{\tA_i}{\tB} \in \R$. and
    by definition of $\Phieqtypeco$ we get $$
\begin{array}{r@{\quad=\quad}l}
\pair{\tA}{\tB} & \pair{\maxuniontype{i \in 1..n}{\tA_i}}{\tB} \in \Phieqtypeco(\R)
\end{array} $$
    
    \item If $n = 1$, then $m > 1$ and $\toNTree{\tB} =
    \nuniontype{j}{m}{\toNTree{\tB_j}}$. Here, by $\ruleEqupUnionR$, we have
    $\toNTree{\tA} \eqtypeaci \toNTree{\tB_j}$ for every $j \in 1..m$, and thus
    $\pair{\tA}{\tB_j} \in \R$. Finally we conclude $$
\begin{array}{r@{\quad=\quad}l}
\pair{\tA}{\tB} & \pair{\tA}{\maxuniontype{j \in 1..m}{\tB_j}} \in \Phieqtypeco(\R)
\end{array} $$
    
    \item If $n > 1$ and $m > 1$, then both $\tA$ and $\tB$ are maximal union
    types with at least two elements each. Here it applies $\ruleEqupUnion$ and
    we have $\toNTree{\tA} = \nuniontype{i}{n}{\toNTree{\tA_i}}$,
    $\toNTree{\tB} = \nuniontype{j}{m}{\toNTree{\tB_j}}$, and there exists $$
\begin{array}{r@{\quad=\quad}l}
f : 1..n \to 1..m & \toNTree{\tA_i} \eqtypeco \toNTree{\tB_{f(i)}} \\
g : 1..m \to 1..n & \toNTree{\tA_{g(j)}} \eqtypeco \toNTree{\tB_j}
\end{array} $$
    Then, $\pair{\tA_i}{\tB_{f(i)}}, \pair{\tA_{g(j)}}{\tB_j} \in \R$ and we
    conclude $$
\begin{array}{r@{\quad=\quad}l}
\pair{\tA}{\tB} & \pair{\maxuniontype{i \in 1..n}{\tA_i}}{\maxuniontype{j \in 1..m}{\tB_j}} \in \Phieqtypeco(\R)
\end{array} $$
  \end{itemize}
\end{itemize}
\end{proof}

The algorithm is the result of adapting the scheme presented for
subtyping to the new relation $\subtypeal{\R}$. This is done by redefining
functions $\fChildrenName$ and $\fCheckName$ from the first and second
phase respectively (\cf Fig.~\ref{fig:pseudocode:equality2}). For the former
the only difference is on rule $\ruleEqupFunc$, where we need to add pair
$\pair{\tA'}{\tB'}$ instead of $\pair{\tB'}{\tA'}$, added for subtyping. We
could have omitted this by using the same rule for functional types as before
and resorting to the symmetry of the resulting relation (which does not depend
on this rule), but we wanted to emphasize the fact that phase one can easily be
adapted if needed. For the refinement phase we need to properly check the
premises of rules $\ruleEqupUnion$ and $\ruleEqupUnionR$, while the others
remain the same.

\begin{figure}
{\small $$
\begin{array}{c@{\kern-1em}c}
\begin{array}{rcl}
\multicolumn{3}{l}{\fChildren{p}:} \\ 
 & & \ttCase\ p\ \ttOf \\
 & & \quad \pair{\datatype{\tD}{\tA}}{\datatype{\tD'}{\tB}} \rightarrow \\
 & & \qqquad \set{\pair{\tD}{\tD'},\pair{\tA}{\tB}} \\
 & & \quad \pair{\functype{\tA'}{\tA''}}{\functype{\tB'}{\tB''}} \rightarrow \\
 & & \qqquad \set{\pair{\tA'}{\tB'},\pair{\tA''}{\tB''}} \\
 & & \quad \pair{\nuniontype{i}{n}{\tA_i}}{\nuniontype{j}{m}{\tB_j}} \rightarrow \\
 & & \qqquad \set{\pair{\tA_i}{\tB_j} \ |\ i\in1..n,\ j\in1..m} \\
 & & \quad \pair{\nuniontype{i}{n}{\tA_i}}{\tB},\ \tB \neq \iuniontype \rightarrow \\
 & & \qqquad \set{\pair{\tA_i}{\tB} \ |\ i\in1..n} \\
 & & \quad \pair{\tA}{\nuniontype{j}{m}{\tB_j}},\ \tA \neq \iuniontype \rightarrow \\
 & & \qqquad \set{\pair{\tA}{\tB_j} \ |\ j\in1..m} \\
 & & \quad \mathtt{otherwise} \rightarrow \\
 & & \qqquad \emptyset
\end{array}
&
\begin{array}{rcl}
\multicolumn{3}{l}{\fCheck{p}{F}:} \\
  & & \ttCase\ p\ \ttOf \\
  & & \quad \pair{a}{a} \rightarrow \\
  & & \qqquad \ttTrue \\
  & & \quad \pair{\datatype{\tD}{\tA}}{\datatype{\tD'}{\tB}} \rightarrow \\
  & & \qqquad \pair{\tD}{\tD'} \notin F \ttAnd \pair{\tA}{\tB} \notin F \\
  & & \quad \pair{\functype{\tA'}{\tA''}}{\functype{\tB'}{\tB''}} \rightarrow \\
  & & \qqquad \pair{\tA'}{\tB'} \notin F \ttAnd \pair{\tA''}{\tB''} \notin F \\
  & & \quad \pair{\nuniontype{i}{n}{\tA_i}}{\nuniontype{j}{m}{\tB_j}} \rightarrow \\
  & & \qqquad \forall i. \exists j.\ \pair{\tA_i}{\tB_j} \notin F \ttAnd \forall j. \exists i.\ \pair{\tA_i}{\tB_j} \notin F \\
  & & \quad \pair{\nuniontype{i}{n}{\tA_i}}{\tB},\ \tB \neq \iuniontype \rightarrow \\
  & & \qqquad \forall i.\ \pair{\tA_i}{\tB} \notin F \\
  & & \quad \pair{\tA}{\nuniontype{j}{m}{\tB_j}},\ \tA \neq \iuniontype \rightarrow \\
  & & \qqquad \forall j.\ \pair{\tA}{\tB_j} \notin F
\end{array}
\end{array} $$
} 
\caption{Pseudo-code of first (left) and second (right) phase for equivalence checking.}
\label{fig:pseudocode:equality2}
\end{figure}

With these considerations is easy to see that, in each iteration, $S$ consists
of pairs in the relation $\eqtypeup{W}$, getting $S \subseteq
\mathbin{\eqtypeaci}$ at the end of the process.

\begin{proposition}
The algorithm preserves the following invariant:
\begin{itemize}
  \itemsep0em
  \item \triple{W}{S}{F} is a partition of $U$
  \item $F$ is composed solely of invalid pairs
  \item $S \subseteq \Phieqtypeup{W}(S)$
\end{itemize}
\end{proposition}

\begin{proof}
Analysis here is exactly the same as the case for subtyping. The only
difference is when proving the second condition. For that, we just need to
make sure that pairs are considered valid according to the rules of
Fig.\ref{fig:equalitySchemesNary}. We present next the cases that differ
between $\subtypeup{W}$ and $\eqtypeup{W}$:
\begin{itemize}
  \item $p = \pair{\functype{\tA}{\tB}}{\functype{\tA'}{\tB'}}$. This case
  corresponds to $\ruleEqupFunc$. Here the algorithm checks for
  $\pair{\tA}{\tA'}$ instead of $\pair{\tA'}{\tA}$. That is the reason why
  $\fChildrenName$ was redefined (\cf Fig.~\ref{fig:pseudocode:equality2}).
  
  \item $p = \pair{\nuniontype{i}{n}{\tA_i}}{\nuniontype{j}{m}{\tB_j}}$. In the
  case of subtyping, the algorithm checks if for each $\tA_i$ there exists
  $\tB_j$ not yet invalidated, to assure the existence of $f$ (\cf
  $\ruleSubupUnion$). Here it extends the check to verify that for each $\tB_j$
  there is a $\tA_i$ not yet invalidating, thus assuring the existence of $g$
  too (\cf $\ruleEqupUnion$).
  
  \item $p = \pair{\tA}{\nuniontype{j}{m}{\tB_j}}$. Here, by rule
  $\ruleEqupUnionR$, there must be a pair $\pair{\tA}{\tB_j}$ in $S \cup W$ for
  each $j \in 1..m$, to keep $p$ in the set of conditionally valid pairs.
\end{itemize}
\end{proof}

For the complexity analysis, notice that the size of the built universe is the
same as before and phase one is governed by phase two, which has at most
$\O(N N')$ iterations. For the cost of each iteration it is enough to analyze
the complexity of $\fCheckName$, since the rest of the scheme remains the same.
As we remarked before, the only difference in $\fCheckName$ between subtyping
and equality is in the cases involving unions. Here the worst case is when
checking rule $\ruleEqupUnion$ that requires the existence of two functions $f$
and $g$ relating elements of each type. This can be done in linear time by
maintaining tables with the count of non-invalidated pairs of descendants, as
indicated in Sec.~\ref{subtyping2:complexity}. Thus, the cost of an iteration
is $\O(d)$, resulting in an overall cost of $\O(N N' d)$ as before.


%% file: conclusions.tex
\section{Conclusions}
\label{sec:conclusions}

We address efficient type-checking for path polymorphism. We start off with the
type system of~\cite{DBLP:journals/entcs/VisoBA16} which includes singleton types, union types,
type application and recursive types. The union type constructor is assumed
associative, commutative and idempotent.  First we formulate a syntax-directed
presentation. Then we devise invertible coinductive presentations of
type-equivalence and subtyping. This yields a na\"ive but correct type-checking
algorithm. However, it proves to be inefficient (exponential in the size of the
type). This prompts us to change the representation of type expressions and use
automata techniques to considerably improve the efficiency. Indeed, the final
algorithm has complexity $\O(n^7 d)$ where $n$ is the size of the input and $d$
is the maximum arity of the unions occurring in it.

Regarding future work an outline of possible avenues follows. These are aimed
at enhancing the expressiveness of $\capp$ itself and then adapting the
techniques presented here to obtain efficient type checking algorithms.

\begin{itemize}  
  \item Addition of parametric polymorphism (presumably in the style of
  F$_{<:}$~\cite{DBLP:conf/tacs/CardelliMMS91,DBLP:books/daglib/0005958,DBLP:journals/iandc/ColazzoG05}). We
  believe this should not present major difficulties. 
  
  \item Strong normalization requires devising a notion of positive/negative
  occurrence in the presence of strong $\mu$-type equality, which is known not
  to be obvious~\cite[page 515]{DBLP:books/daglib/0032840}.
  
  \item A more ambitious extension is that of \emph{dynamic patterns}, namely
  patterns that may be computed at run-time, \ppc\ being the prime example of a
  calculus supporting this feature.
\end{itemize}
